\documentclass[
notitlepage
,floatfix,
aps,
pra,
reprint,
onecolumn,
superscriptaddress,
,10pt
]{revtex4-2}
\usepackage[utf8]{inputenc}

\usepackage{amssymb,amsthm}

\usepackage[caption=false]{subfig}
\usepackage{graphicx}
\usepackage{here}
\usepackage{epstopdf}
\usepackage{array}
\usepackage{physics}
\usepackage{verbatim}
\usepackage{amsmath,amsfonts,amscd,mathtools}

\usepackage{tabularx}
\usepackage{stmaryrd}
\usepackage{enumerate}
\usepackage{wasysym}
\usepackage{braket}
\usepackage{mathrsfs}
\usepackage{microtype}
\usepackage{hyperref}
\usepackage[dvipsnames]{xcolor}
\hypersetup{
    bookmarksnumbered=true, 
    unicode=false, 
    pdfstartview={FitH}, 
    pdftitle={}, 
    pdfauthor={}, 
    pdfsubject={}, 
    pdfcreator={}, 
    pdfproducer={}, 
    pdfkeywords={}, 
    pdfnewwindow=true, 
    colorlinks=true, 
    linkcolor=NavyBlue, 
    citecolor=NavyBlue, 
    filecolor=NavyBlue, 
    urlcolor=NavyBlue 
}
\usepackage{tikz}
\usepackage[lmargin=.7in,rmargin=.7in,tmargin=.7in,bmargin=1in]{geometry}
\usepackage{newtxtext,newtxmath}

\usepackage{cleveref}

\usepackage{youngtab}

\usepackage{algorithm}
\usepackage{algorithmicx}
\usepackage{algpseudocode}



\usepackage{booktabs}

\theoremstyle{plain}
\newtheorem{thm}{Theorem}

\newtheorem{lem}[thm]{Lemma}
\newtheorem{pro}[thm]{Proposition}

\theoremstyle{definition}
\newtheorem{defn}[thm]{Definition}

\usepackage[most,breakable]{tcolorbox}
\tcbset{ sharp corners}
  {\expandafter\ifstrequal\expandafter{#1}{filled}{\begin{tcolorbox}[colback=MidnightBlue!70!black!70!TealBlue!2!white,colframe=MidnightBlue!70!black!70!TealBlue!30!white,breakable,enhanced,left=5.75pt,right=5.75pt,grow sidewards by=10pt]}{\begin{tcolorbox}[colback=white,colframe=gray!0,breakable,enhanced,left=5.75pt,right=5.75pt,grow sidewards by=10pt]}}%
  {\end{tcolorbox}}

\newenvironment{defnboxed}[1][white]
  {\expandafter\ifstrequal\expandafter{#1}{filled}{\begin{tcolorbox}[colback=MidnightBlue!70!black!70!TealBlue!2!white,colframe=MidnightBlue!70!black!70!TealBlue!30!white,breakable,enhanced,left=5.75pt,right=5.75pt,grow sidewards by=10pt]}{\begin{tcolorbox}[colback=MidnightBlue!10,colframe=MidnightBlue!0,breakable,enhanced,left=5.75pt,right=5.75pt,grow sidewards by=10pt]}}%
  {\end{tcolorbox}}

\newenvironment{proboxed}[1][white]
  {\expandafter\ifstrequal\expandafter{#1}{filled}{\begin{tcolorbox}[colback=MidnightBlue!70!black!70!TealBlue!2!white,colframe=MidnightBlue!70!black!70!TealBlue!30!white,breakable,enhanced,left=5.75pt,right=5.75pt,grow sidewards by=10pt]}{\begin{tcolorbox}[colback=OliveGreen!10,colframe=OliveGreen!0,breakable,enhanced,left=5.75pt,right=5.75pt,grow sidewards by=10pt]}}%
  {\end{tcolorbox}}

  \newenvironment{thmboxed}[1][white]
  {\expandafter\ifstrequal\expandafter{#1}{filled}{\begin{tcolorbox}[colback=MidnightBlue!70!black!70!TealBlue!2!white,colframe=MidnightBlue!70!black!70!TealBlue!30!white,breakable,enhanced,left=5.75pt,right=5.75pt,grow sidewards by=10pt]}{\begin{tcolorbox}[colback=Maroon!10,colframe=Maroon!0,breakable,enhanced,left=5.75pt,right=5.75pt,grow sidewards by=10pt]}}%
  {\end{tcolorbox}}
  
   \newenvironment{lemboxed}[1][white]
  {\expandafter\ifstrequal\expandafter{#1}{filled}{\begin{tcolorbox}[colback=MidnightBlue!70!black!70!TealBlue!2!white,colframe=MidnightBlue!70!black!70!TealBlue!30!white,breakable,enhanced,left=5.75pt,right=5.75pt,grow sidewards by=10pt]}{\begin{tcolorbox}[colback=Gray!20,colframe=Maroon!0,breakable,enhanced,left=5.75pt,right=5.75pt,grow sidewards by=10pt]}}%
  {\end{tcolorbox}}

  {\expandafter\ifstrequal\expandafter{#1}{filled}{\begin{tcolorbox}[colback=MidnightBlue!70!black!70!TealBlue!2!white,colframe=MidnightBlue!70!black!70!TealBlue!30!white,breakable,enhanced,left=5.75pt,right=5.75pt,grow sidewards by=10pt]}{\begin{tcolorbox}[colback=Yellow!20,colframe=Maroon!0,breakable,enhanced,left=5.75pt,right=5.75pt,grow sidewards by=10pt]}}%
  {\end{tcolorbox}}

\newcommand{\eq}[1]{(\hyperref[eq:#1]{\ref*{eq:#1}})}

\renewcommand{\sec}[1]{\hyperref[sec:#1]{Section~\ref*{sec:#1}}}
\newcommand{\thrm}[1]{\hyperref[thrm:#1]{Theorem~\ref*{thrm:#1}}}
\newcommand{\lemm}[1]{\hyperref[lemm:#1]{Lemma~\ref*{lemm:#1}}}
\newcommand{\prop}[1]{\hyperref[prop:#1]{Proposition~\ref*{prop:#1}}}
\newcommand{\corr}[1]{\hyperref[corr:#1]{Corollary~\ref*{corr:#1}}}
\newcommand{\fig}[1]{\hyperref[fig:#1]{~\ref*{fig:#1}}}
\newcommand{\deff}[1]{\hyperref[deff:#1]{~\ref*{deff:#1}}}

\newcommand{\mE}{\mathcal{E}}
\newcommand{\mN}{\mathcal{N}}
\newcommand{\mU}{\mathcal{U}}

\newcommand{\mD}{\mathcal{D}}
\newcommand{\mI}{\mathcal{I}}

\newcommand{\mH}{\mathcal{H}}
\newcommand{\mM}{\mathcal{M}}

\newcommand{\mB}{\mathcal{B}}
\newcommand{\mK}{\mathcal{K}}

\newcommand{\mR}{\mathcal{R}}
\newcommand{\mS}{\mathcal{S}}
\newcommand{\mX}{\mathcal{X}}

\newcommand{\mbE}{\mathbb{E}}

\newcommand{\mbN}{\mathbb{N}}

\newcommand{\mfS}{\mathfrak{S}}

\newcommand{\aX}{{\abs{\mX}}}

\newcommand{\sD}{\widetilde{D}}

\newcommand{\pD}{\overline{D}}

\newcommand{\tX}{\widetilde{X}}
\newcommand{\tA}{{\widetilde{A}}}
\newcommand{\tB}{{\widetilde{B}}}
\newcommand{\tC}{{\widetilde{C}}}

\newcommand{\tE}{{\widetilde{E}}}
\newcommand{\tF}{{\widetilde{F}}}
\newcommand{\tS}{{\widetilde{S}}}

\newcommand{\sI}{\widetilde{I}}
\newcommand{\pI}{\overline{I}}

\newcommand{\mW}{\mathcal{W}}

\newcommand{\ve}{\varepsilon}

\DeclareMathOperator{\supp}{supp}

\newcommand{\ba}{\begin{eqnarray}}
\newcommand{\ea}{\end{eqnarray}}
\newcommand{\bann}{\begin{eqnarray*}}
\newcommand{\eann}{\end{eqnarray*}}
\newcommand{\bal}{\begin{equation}\begin{aligned}}
\newcommand{\eal}{\end{aligned}\end{equation}}

\newcolumntype{L}[1]{>{\raggedright}p{#1}}
\newcolumntype{C}[1]{>{\centering}p{#1}}
\newcolumntype{R}[1]{>{\raggedleft}p{#1}}
\newcolumntype{D}{>{\centering\arraybackslash}X}

\begin{document}

\title {All you need is the universal correlation detector: \\
A unified approach to universalize communication protocols over quantum channels
}

\author{Kaito Watanabe}
\email{watanabe715@g.ecc.u-tokyo.ac.jp}
\affiliation{Department of Basic Science, The University of Tokyo, 3-8-1 Komaba, Meguro-ku, Tokyo 153-8902, Japan}
\affiliation{RIKEN Center for Quantum Computing (RQC), Hirosawa 2-1, Wako, Saitama 351-0198, Japan}

\author{Takaya Matsuura}
\email{takayamatsuura@gmail.com}
\affiliation{RIKEN Center for Quantum Computing (RQC), Hirosawa 2-1, Wako, Saitama 351-0198, Japan}

\author{Ryuji Takagi}
\email{ryujitakagi@g.ecc.u-tokyo.ac.jp}
\affiliation{Department of Basic Science, The University of Tokyo, 3-8-1 Komaba, Meguro-ku, Tokyo 153-8902, Japan}


\begin{abstract}
Constructing optimal quantum information protocols without complete knowledge of the underlying states or channels is a central challenge. We address this challenge by developing universal correlation detection as a common building block for quantum communication. Our detectors distinguish a bipartite state from the product of its marginals and attain the same first-order asymptotic performance as optimal tests constructed with complete state information. For general bipartite quantum states, knowledge of a single marginal suffices: a detector depending only on that marginal is first-order optimal for every compatible state. For classical–quantum states, the detector is fully universal and requires no prior state information. By combining these detectors with position-based decoding and convex splitting, we construct channel-independent coding schemes that achieve capacity for a range of communication tasks. These results turn universal correlation detection into a systematic tool for universal protocol design, providing a unified route to designing capacity-achieving protocols for many communication tasks over unknown quantum channels.
\end{abstract}
\maketitle
\tableofcontents

\section{Introduction}

Correlation is one of the most fundamental objects of study in both classical and quantum information theory. In particular, correlations between multiple quantum states, such as quantum entanglement~\cite{Horodecki_2009_quantum_entanglement} and quantum discord~\cite{L_Henderson_2001, Harold_quantum_discord}, are central to quantum mechanics and play an important role not only to understand the fundamental difference in the classical and quantum information theory, but also in applications in quantum Shannon theory~\cite{Wilde_2016_book, khatri_2024_book} and quantum resource theories~\cite{Chitamber_gour, Gour_2025_book}.

In this context, correlation detection, the task of testing whether a given bipartite state is correlated, is of fundamental importance.
In particular, it has been extensively studied within quantum hypothesis testing, the task of discriminating between two states~\cite{Hayashi_2016_correlation_detection, girardi2025umlaut_information, girardi2025quantum_umlaut, Dasgupta_2025_universal_tester}. In correlation detection, these states are the correlated bipartite quantum state $\rho_{AB}$ and the corresponding uncorrelated state $\rho_A\otimes \rho_B$.
With complete information about $\rho_{AB}$, one can use previous results on quantum hypothesis testing to characterize the optimal performance of correlation detection.

However, the assumption of complete information about the state whose correlations are to be detected is quite strong and limits applicability in practical settings, as the quantum state may not be prepared precisely as intended and might experience unknown noise. 
This motivates us to study state-agnostic correlation detection, where one is given an unknown bipartite quantum state and must test whether the state is correlated. At best, we can hope to achieve optimal performance without complete information, but it is unclear at first sight whether this is possible.

In this work, we show that there indeed exists
a correlation-detection test achieving the  optimal correlation detection in terms of the first order for any bipartite state $\rho_{AB}$ only with knowledge of the marginal state $\rho_B=\Tr_A\rho_{AB}$. This means that correlation detection can be performed optimally using only knowledge of the marginal state. Furthermore, we show that correlation detection for a classical-quantum state is universally achievable, meaning that first-order optimal correlation detection is achievable without any knowledge about the state. In the latter case, the test was proposed previously in Ref.~\cite{Dasgupta_2025_universal_tester} for a different purpose.

Furthermore, we apply the universal correlation detector more broadly to establish universal protocols for various information-theoretic tasks.
For instance, the goal of classical-quantum channel coding is to send messages reliably by applying preprocessing---called encoding---and the measurement on the output quantum system---called decoding---to a given noisy channel.
Our goal is to design the coding scheme so we can send messages as efficiently as possible. Designing the coding scheme to achieve the channel capacity---the maximum number of bits one can send through the noisy channel---is a central task in quantum communication theory~\cite{Wilde_2016_book, khatri_2021_second_order, Holevo_1996_capacity, SW_theorem, Bennett_2002_Entanglement_assisted, Holevo_2002_entanglement_assisted, Dupuis_2009_capacity_side_information, Devetak_2004_private_classical_capacity, Lloyd_1997_capacity}.

In the standard channel-coding setting, we often assume complete knowledge about the noisy channel. However, in practice, this knowledge may be unavailable due to unknown noise, which motivates us to study universal channel coding: a coding scheme that achieves the optimal communication rate without knowledge of the channel.

We show that the universal correlation detector provides a unified approach to constructing universal coding schemes for various communication tasks. In fact, the position-based decoding strategy, a known proof strategy for deriving capacity formulas for many communication tasks~\cite{Anshu_2019_building, Anshu_2019_on_the_near, Wilde_2017_position_based, khatri_2021_second_order}, combined with the universal correlation detector, yields a universal decoder that achieves optimal performance.
Specifically, this idea yields a simple construction of universal coding schemes for the following communication problems:
\begin{enumerate}
    \item Classical-quantum channel coding assisted by shared randomness with an unknown distribution
    \item Classical-quantum wiretap channel coding assisted by shared randomness with an unknown distribution
    \item One-way secret key distillation from classical-quantum-quantum states
    \item Entanglement-assisted classical communication
    \item Entanglement-assisted Gel'fand-Pinsker channel coding 
    \item Entanglement-assisted Marton's bound for a quantum broadcast channel with $L\geq 2$ receivers.
\end{enumerate}

To the best of our knowledge, universal coding schemes for the latter two settings had not been established previously, but arise naturally within our framework.
Our findings reveal that universal correlation detection offers a simple and unified way to construct universal channel coding schemes for various kinds of communication tasks, indicating the possibility of extending the technique not only to quantum communication tasks but also to a broad range of quantum information-theoretic tasks in quantum learning theory and quantum resource theories.

\section{Preliminaries}
Throughout this paper, we consider finite-dimensional Hilbert spaces.
We denote the set of density matrices on a Hilbert space $\mH$ by $\mD(\mH)$. We sometimes abbreviate Hilbert-space notation such as $\mH_A,\mH_B$ as $A, B$.
Furthermore, we write $[d]\coloneqq\qty{1,\ldots, d}$ for any $d\in\mbN$. $T(\rho,\sigma)\coloneq\frac{1}{2}\|\rho-\sigma\|_1$ is the trace distance, and $P(\rho,\sigma)\coloneqq\sqrt{1-F(\rho,\sigma)}=\sqrt{1-\|\sqrt{\rho}\sqrt{\sigma}\|^2_1}$ is the purified distance.

\subsection{Information-theoretic quantities}
We first define several relevant information-theoretic quantities.
For a density matrix $\rho\in\mD(\mH)$, the von Neumann entropy $S(\rho)$ is defined as
\bal
S(\rho)\coloneqq-\Tr[\rho\log \rho].
\eal
The Umegaki relative entropy of $\rho\in\mD(\mH)$ with respect to $\sigma\in\mD(\mH)$ is defined as~\cite{Umegaki_relative}
\begin{equation}
    D(\rho\|\sigma)\coloneqq
    \left\{\,
\begin{aligned}
    &\Tr[\rho\log \rho-\rho\log \sigma]~~~(\supp\rho\subset\supp\sigma)\\
    &+\infty~~~(\mbox{otherwise})
\end{aligned}
\right.
\end{equation}
The Petz Rényi relative entropy of $\rho\in\mD(\mH)$ with respect to $\sigma\in\mD(\mH)$ of order $\alpha\in(0,1)\cup(1,\infty]$ is defined as~\cite{Petz_1986_quasi_entropy}
\begin{equation}
    \pD_\alpha(\rho\|\sigma)\coloneqq\left\{\,
    \begin{aligned}
        &\frac{1}{\alpha-1}\log\Tr\qty[\rho^\alpha\sigma^{1-\alpha}]~~~(\supp\rho\subset\supp\sigma\mbox{ or }\supp\rho\not\perp\supp\sigma, ~0< \alpha<1)\\
        &+\infty~~~(\mbox{otherwise})
    \end{aligned}
    \right.
\end{equation}
The Petz Rényi relative entropy coincides with the Umegaki relative entropy in the limit $\alpha\to 1$.

Another quantum extension of Rényi relative entropy is the sandwiched Rényi relative entropy of $\rho\in\mD(\mH)$ with respect to $\sigma\in\mD(\mH)$ of order $\alpha\in(0,1)\cup(1,\infty]$, defined as~\cite{Muller_Lennert_2013_on_quantum_renyi, Wilde_2014_strong_converse}
\begin{equation}
    \sD_\alpha(\rho\|\sigma)\coloneqq\left\{\,
    \begin{aligned}
        &\frac{1}{\alpha-1}\log\Tr\qty[\qty(\sigma^{\frac{1-\alpha}{2\alpha}}\rho\sigma^{\frac{1-\alpha}{2\alpha}})^\alpha]~~~(\supp\rho\subset\supp\sigma\mbox{ or }\supp\rho\not\perp\supp\sigma, ~0<\alpha<1)\\
        &+\infty~~~(\mbox{otherwise})
    \end{aligned}
    \right.
\end{equation}
The sandwiched Rényi relative entropy also coincides with the Umegaki relative entropy in the limit $\alpha\to 1$.

Now, we define the mutual information for a bipartite state.
For a bipartite state $\rho_{AB}\in\mD(\mH_A\otimes \mH_B)$, the mutual information $I(A:B)_\rho$ of $\rho_{AB}$ is defined as 
\bal
I(A:B)_\rho\coloneqq D(\rho_{AB}\|\rho_A\otimes \rho_B)=S(A)_\rho+S(B)_\rho-S(AB)_\rho.
\eal
Mutual information can also be generalized to a quantity called total correlation, defined as follows.
For a given $L$-partite state $\rho_{A_1\cdots A_L}$, the total correlation $I[S]$ of a subset $S\subset [L]$ of the parties is defined as~\cite{Modi_2010_unified, Avis_2008_distributed_compression}
    \bal\label{eq: def of total correlation}
    I[S]_\rho\coloneqq D\qty(\Tr_{\backslash A_S}\rho_{A_1\cdots A_L}\middle\| \bigotimes_{l\in S}\Tr_{\backslash A_l}\rho_{A_1\cdots A_L})=\sum_{l\in S}S(A_l)_\rho-S(A_S)_\rho.
    \eal
    Here, we use $A_S$ to denote the composite system of $A_l$ with $l\in S$.

The Petz Rényi mutual information of order $\alpha\in(0,1)\cup(1,\infty)$ is defined as
\bal
\pI^\downarrow_\alpha(A:B)_\rho\coloneqq\inf_{\sigma_B\in\mD(\mH_B)}\pD_\alpha(\rho_{AB}\|\rho_A\otimes\sigma_B)
\eal
The Petz Rényi mutual information admits a concise closed form given by the quantum Sibson identity~\cite{Sharma_2013_fundamental_bound, Hayashi_2016_correlation_detection}
\bal\label{eq: fully quanutm sibsion}
\pI_\alpha^{\downarrow}(A:B)_\rho=\frac{\alpha}{\alpha-1}\log \Tr[\Tr_A[\rho_{AB}^\alpha\rho_A^{1-\alpha}]^{\frac{1}{\alpha}}]
\eal
In particular, for a classical-quantum state $\rho_{XB}$, the quantum Sibson identity reduces to
\bal\label{eq: cq sibson}
\pI_\alpha^\downarrow(X:B)_\rho=\frac{\alpha}{\alpha-1}\log\Tr[\qty(\sum_{x\in\mX}p_X(x)\rho_x^\alpha)^{1/\alpha}],
\eal
The sandwiched Rényi mutual information of order $\alpha\in(0,1)\cup(1,\infty)$ is defined as 
\bal
\sI^\downarrow_\alpha(A:B)_\rho\coloneqq\inf_{\sigma_B\in\mD(\mH_B)}\sD_\alpha(\rho_{AB}\|\rho_A\otimes\sigma_B).
\eal
The sandwiched Rényi total correlation of order $\alpha\in(0,1)\cup(1,\infty)$ for the subsystems labeled by $S\subset[L]$ is defined as~\cite{Cheng_2023_quantum_broadcast}
\bal
\sI^\downarrow_\alpha[S]_\rho\coloneqq\sD_\alpha\qty(\Tr_{\backslash A_S}\rho_{A_1\cdots A_L}\middle\| \bigotimes_{l\in S}\Tr_{\backslash A_l}\rho_{A_1\cdots A_L})
\eal
Note that all the Rényi information quantities defined above reduce to the standard mutual information or total correlation in the limit $\alpha\to 1$:
\bal
\lim_{\alpha\to 1}\pI^\downarrow_\alpha(A:B)_\rho=\lim_{\alpha\to 1}\sI^\downarrow_\alpha(A:B)_\rho=I(A:B)_\rho, ~~\lim_{\alpha\to 1}\sI^\downarrow_\alpha[S]_\rho=I[S]_\rho.
\eal

\subsection{Universal symmetric states}
In the following, we review the notion of a universal symmetric state, a key ingredient in constructing a measurement for universal correlation detection. The universal state was originally introduced in Ref.~\cite{Hayashi_2009_universal_coding} to construct a universal decoder for classical-quantum channel coding and has since been widely used in quantum information theory~\cite{Hayashi_2016_correlation_detection, Matsuura_2025_universal_resolvability, Fang_2026_error_exponent}.

Consider a Hilbert space $\mH^{\otimes m}$.
Due to Schur-Weyl duality, we can decompose this as 
\bal
\mH^{\otimes m}=\bigoplus_{\lambda\in Y^m_d}\mW_\lambda\otimes \mU_\lambda,
\eal
where $\mW_\lambda$ is an irreducible representation (irrep.)\ of the general linear group, and $\mU_\lambda$ is an irrep.\ of the symmetric group. $Y^m_d$ denotes the set of all Young diagrams with $m$ boxes and depth at most $d$. 
Let $\Pi_\lambda$ denote the projector onto the subspace $\mW_\lambda\otimes \mU_\lambda$. 
We define two quantum states $\sigma^\lambda,\sigma^{U,m}$ as 
\bal
\sigma^\lambda&\coloneqq\frac{\Pi_\lambda}{\dim(\mW_\lambda\otimes \mU_\lambda)}\\
\sigma^{U,m}&\coloneqq\frac{1}{\abs{Y^m_d}}\sum_{\lambda\in Y^m_d}\sigma^\lambda.
\eal
The second state is called a \emph{universal symmetric state}.

Since any permutation-invariant state $\rho_m$ on $\mH^{\otimes m}$ can be decomposed as 
\bal
\rho_m=\bigoplus_{\lambda\in Y^m_d}\rho_\lambda\otimes \frac{I_{\mU_\lambda}}{\dim \mU_\lambda},
\eal
it follows that 
\bal
\Pi_\lambda\rho_m\Pi_\lambda=\rho_\lambda\otimes \frac{I_{\mU_\lambda}}{\dim \mU_\lambda}\leq \frac{\Pi_\lambda}{\dim \mU_\lambda}\leq (\dim\mW_\lambda) \sigma^\lambda.
\eal
It follows that, for any permutation-invariant state $\rho_m$,
\bal
\rho_m&\leq \sum_{\lambda\in Y^m_d}(\dim\mW_\lambda) \sigma^\lambda\\
&\leq \max_{\lambda\in Y^m_d}(\dim\mW_\lambda) \abs{Y^m_d}\sigma^{U,m}
\eal
Due to Weyl's character formula, we have
\bal\label{eq: universal for single alphabet}
\dim \mW_\lambda\leq (m+1)^{d(d-1)/2}
\eal
for any $\lambda\in Y^m_d$. Moreover, due to the type-counting argument~\cite{cover_1999_elements}, we have $\abs{Y^m_d}\leq (m+1)^{d-1}$.
From this, we have $\rho_m\leq (m+1)^{(d+2)(d-1)/2}\sigma^{U,m}$.

For a classical-quantum state $\rho_{XB}$, the bound takes the following form.
For a string $x^n\in \mX^n$ of length $n$, the type $P$ of $x^n$ is defined as the distribution
\bal
P(a|x^n)=\frac{1}{n}\#\qty{i\in[n]~|~x^n(i)=a}.
\eal
The set $T^n_P\subset \mX^n$ is defined as the set of strings with type $P$.

Now, suppose that we have a quantum state $\rho_{x^n}\coloneqq\rho_{x^n(1)}\otimes \cdots\otimes \rho_{x^n(n)}$ corresponding to the sequence $x^n\in T^n_P$ with type $P$.
We first consider a string
\bal
x^n_P\coloneqq1^{m_1}2^{m_2}\cdots \abs{\mX}^{m_\abs{\mX}},
\eal
where $m_x\coloneqq n P(x)$ for any $x\in\mX$. Let us see how $\rho_{x^n_P}$ can be bounded from above.
Due to Eq.~\eqref{eq: universal for single alphabet}, we have
\bal
\rho_{B^n}^{x^n_P}=\rho_1^{\otimes m_1}\otimes \cdots\rho_{\aX}^{\otimes m_{\aX}}&\leq \bigotimes_{x=1}^{\aX}(m_x+1)^{(d+2)(d-1)/2}\sigma_{U, m_1}\otimes \cdots\otimes \sigma_{U, m_\aX}\\
&\leq (n+1)^{\aX(d+2)(d-1)/2}\sigma_{U, m_1}\otimes \cdots\otimes \sigma_{U, m_\aX}.
\eal
Using this idea, we can obtain a bound on $\rho_{B^n}^{x^n}$ with $x^n\in T^n_P$.
Take a permutation $\pi\in\mfS_n$ such that $\pi(x^n_P)=x^n$. Denoting by $V_\pi$ the unitary representing the permutation $\pi\in\mfS_n$, we have
\bal
\rho_{B^n}^{x^n}=V_{\pi}\rho_{B^n}^{x^n_P}V^\dagger_\pi\leq (n+1)^{\aX(d+2)(d-1)/2}V_\pi \qty(\sigma^{U, m_1}\otimes \cdots\otimes \sigma^{U, m_\aX})V_\pi^\dagger.
\eal
Defining $\sigma_{B^n}^{x^n}$ as $\sigma_{B^n}^{x^n}=V_\pi \qty(\sigma_{B^{m_1}}^{U, m_1}\otimes \cdots\otimes \sigma_{B^{m_\aX}}^{U, m_\aX})V_\pi^\dagger$, we have
\bal
\rho_{B^n}^{x^n}\leq {\rm poly }(n)\sigma_{B^n}^{x^n}.
\label{eq:universal state bound}
\eal

Furthermore, we define another quantum state
\bal
\sigma_{X^nB^n}^{U,p}=\sum_{x^n}p^n(x^n)\ketbra{x^n}{x^n}_{X^n}\otimes \sigma_{B^n}^{x^n}
\eal
Since $\sigma_{B^n}^{U,p}$ is permutation symmetric and commutes with $U^{\otimes n}$ for any $U\in\mU(\mH)$, we can decompose $\sigma_{B^n}^{U,p}$ as 
\bal
\sigma_{B^n}^{U,P}=\bigoplus_{\lambda\in Y^n_d}c_\lambda \Pi_{\lambda}.
\eal
In particular, note that $[\sigma_{B^n}^{U,p},\sigma_{B^n}^{x^n}]=0$ holds for any $x_n\in \mX^n$.

\subsection{Convex splitting}
In addition to universal correlation detection, discussed below, we use convex splitting~\cite{Anshu2017quantum,Anshu_2019_building, Anshu_2019_on_the_near, Cheng_2023_quantum_broadcast, Cheng_2023_tight_convex_splitting}, a powerful technique for decoupling multipartite states by mixing quantum states.
We first review the simplest case, unipartite convex splitting, and then review the multipartite case.

\begin{lemboxed}
    \begin{lem}[Unipartite convex splitting~{\cite{Cheng_2023_tight_convex_splitting}}]\label{lem: unipartite convex splitting}
        Let $\rho_{AB}$ be a quantum state, and let $\tau_{A_1\ldots A_MB}$ be
        \bal
        \tau_{A_1\ldots A_MB}\coloneqq\frac{1}{M}\sum_{m=1}^M\rho_{A_1}\otimes \cdots\otimes \rho_{A_{m-1}}\otimes \rho_{A_mB}\otimes \rho_{A_{m+1}}\otimes \rho_{A_M}=\frac{1}{M}\sum_{m=1}^M \rho_{A_mB}\otimes \bigotimes_{m'\in[M]~|~m'\neq m}\rho_{A_{m'}}.
        \eal
        Then, it holds that 
        \bal
        \frac{1}{2}\left\|\tau_{A_1\ldots A_MB}-\rho_A^{\otimes M}\otimes \rho_B \right\|_1\leq 2^{-\sup_{1\leq \alpha \leq 2}\frac{\alpha-1}{\alpha}(\log M-\sI^\downarrow_\alpha(A:B)_\rho)}.
        \eal
    \end{lem}
\end{lemboxed}

The following is a special case of the multipartite convex splitting lemma established in Ref.~\cite{Cheng_2023_quantum_broadcast}, which is relevant to universal entanglement-assisted classical-quantum broadcast channel coding in Section~\ref{sec: universal marton}.
\begin{lemboxed}
    \begin{lem}[Multipartite convex splitting~{\cite{Cheng_2023_quantum_broadcast}}]\label{lem: multipartite convex splitting}
        Let $\rho_{A_1\cdots A_L}$ be an $L$-partite state, and let $\tau_{A_1^{M_1}\cdots A_L^{M_L}}$ be a quantum state defined as
        \bal
        \tau_{A_1^{M_1}\cdots A_L^{M_L}}\coloneqq\frac{1}{\prod_{l=1}^LM_l}\sum_{(m_1,\ldots, m_L)\in[M_1]\times \cdots\times[M_L]}\rho_{A_{1, m_1}\cdots A_{L, m_L}}\otimes \qty(\bigotimes_{l=1}^L\bigotimes_{m_l'\in[M_l]~|~m_l'\neq m_l}\rho_{A_{m_l'}}).
        \eal
        Then, it holds that 
        \bal
        \left\|\tau_{A_1^{M_1}\cdots A_L^{M_L}}-\rho_{A_1}^{\otimes M_1}\otimes \cdots\otimes\rho_{A_L}^{\otimes M_L} \right\|_1\leq \sum_{\emptyset\neq S\subset[L]}2^{\abs{S}}2^{-\sup_{1\leq \alpha\leq 2}\frac{\alpha-1}{\alpha}(\sum_{l\in S}\log M_l-\sI_\alpha^\downarrow[S])}
        \eal
    \end{lem}
\end{lemboxed}

\subsection{Other technical lemmas}
In this section, we present technical lemmas that will be used frequently in the subsequent discussion.

\begin{lemboxed}
    \begin{lem}[{Lemma 1 in Ref.~\cite{Anshu_2019_building}}]\label{lem: continuity of the prbability}
        For any quantum states $\rho,\sigma\in\mD(\mH)$ and any operator $0\leq \Lambda\leq I$, it holds that 
        \bal
        \abs{\sqrt{\Tr\Lambda\rho}-\sqrt{\Tr\Lambda\sigma}}\leq P(\rho,\sigma).
        \eal
    \end{lem}
\end{lemboxed}

\begin{lemboxed}
    \begin{lem}[Relation between the trace distance and the purified distance~{\cite{Fuchs_van_de_graaf}}]
        Let $\rho, \sigma\in\mD(\mH)$ be two quantum states. Then, it holds that
        \bal
        T(\rho,\sigma)\leq P(\rho,\sigma)\leq \sqrt{2T(\rho,\sigma)}.
        \eal
        Here, $T(\rho,\sigma)\coloneqq\frac{1}{2}\|\rho-\sigma\|_1$ is the trace distance, and $P(\rho,\sigma)\coloneqq\sqrt{1-F(\rho,\sigma)}=\sqrt{1-\|\sqrt{\rho}\sqrt{\sigma}\|^2_1}$ is the purified distance.
    \end{lem}
\end{lemboxed}

\begin{lemboxed}
    \begin{lem}[Hayashi-Nagaoka inequality,~{Ref.~\cite[Lemma 2]{Hayashi_Nagaoka}}]\label{lem: Hayashi Nagaoka inequality}
        For any operators $0\leq S\leq I$, $T\geq 0$, and any positive number $c>0$,
\bal
I-(S+T)^{-\frac{1}{2}} S (S+T)^{-\frac{1}{2}}&\leq (1+c)(I-S)+(2+c+c^{-1})T\\
&=c_1(I-S)+c_2T,
\eal
holds. Here, we write $c_1\coloneqq1+c, ~c_2\coloneqq2+c+c^{-1}$.
    \end{lem}
\end{lemboxed}

\begin{lemboxed}
    \begin{lem}[{Ref.~\cite[Lemma 2]{Hayashi_2009_universal_coding}}]\label{lem: variational form}
        For any operator $X\geq 0$ and any number $0\leq t<1$, we have
        \bal
        \max_{\sigma\in\mD(\mH)}\Tr[X\sigma^t]=\qty(\Tr[X^{\frac{1}{1-t}}])^{1-t}.
        \eal
        For $t=1$, we have
        \bal
        \max_{\sigma\in\mD(\mH)}\Tr[X\sigma]=\|X\|_\infty.
        \eal
    \end{lem}
\end{lemboxed}

\section{Universal correlation detection}\label{sec:correlation detection}

\subsection{Quantum hypothesis testing and correlation detection}
Here, we review the setting of quantum hypothesis testing and the task of correlation detection as a special case.
Hypothesis testing aims to determine whether the null hypothesis or the alternative hypothesis is correct. It is known to be connected to various aspects of the performance of quantum information-processing tasks, such as channel coding~\cite{Anshu_2019_building,Anshu_2019_on_the_near,Anshu_2019_hypothesis_testing, Polyanskiy_2010_channel_coding, Cheng_2023_simple_and_tighter} and resource distillation in quantum resource theories~\cite{Liu_one_shot, Takagi_One_shot, Regula_Takagi_2021, Regula_benchmarking, hayashi_generalized_2025, Lami_2025_gqsl,lami_2024_asymptotic_quantification}.
In the standard scenario of quantum hypothesis testing, the hypotheses are represented by quantum states: the null hypothesis corresponds to a quantum state $\rho$, and the alternative hypothesis corresponds to another quantum state $\sigma$. We infer which state is actually given by performing a binary POVM $\qty{M, I-M}$.
Here, the outcome corresponding to the POVM element $M$ indicates a guess that the given state is $\rho$, whereas the outcome corresponding to $I-M$ indicates a guess that the given state is $\sigma$.

There are two possible errors in this hypothesis testing setting: 
The type I error means that the given state is $\rho$ but one infers that $\sigma$ is given.
The type II error means that the given state is $\sigma$, but one infers that $\rho$ is given.
A standard figure of merit representing how well one can discriminate between $\rho$ and $\sigma$ is the hypothesis testing divergence, defined as
\bal
D^\ve_H(\rho\|\sigma)\coloneqq-\log \min_{\substack{0\leq M\leq I \\ \Tr[\rho (I-M)]\leq \ve}}\Tr[\sigma M].
\eal
This means that one minimizes the type II error probability while the type I error probability is kept no larger than a constant $\ve$. 
In particular, the direct part of quantum hypothesis testing implies that, for any $0<r<R$, there exists a sequence of POVMs $\qty{M_n, I-M_n}$ such that both the type I error and the type II error vanish in the asymptotic limit.

One can also consider the multicopy setting, where one is given either $\rho^{\otimes n}$ or $\sigma^{\otimes n}$.
In fact, quantum Stein's lemma~\cite{hiai_1991_proper, Ogawa_2000_strong} states that the hypothesis testing divergence satisfies
\bal
\lim_{n\to\infty}\frac{1}{n}D^\ve_H(\rho^{\otimes n}\|\sigma^{\otimes n})=D(\rho\|\sigma),~~~\forall\ve\in(0,1).
\eal

Correlation detection can be formulated in the language of hypothesis testing. For a bipartite quantum state $\rho_{AB}$, suppose that we are given either the correlated state $\rho_{AB}$ or the uncorrelated state $\rho_A\otimes \rho_B$, the tensor product of the marginals, and we perform a measurement to discriminate between these two.
In the multicopy setting where we are given $n$ copies of $\rho_{AB}$ or $n$ copies of $\rho_A\otimes \rho_B$, the hypothesis testing divergence in this setting converges to the mutual information as 
\bal
\lim_{n\to\infty}\frac{1}{n}D^\ve_H(\rho_{AB}^{\otimes n}\|\rho_A^{\otimes n}\otimes \rho_B^{\otimes n})=I(A:B)_\rho.
\eal

\subsection{Universal correlation detection for classical-quantum states}

Consider an unknown cq channel $W:\mX\to\mD(\mH), ~x\mapsto\rho_B^x$.
If we choose a probability distribution $p_X$ over $\mX$ and feed a sampled symbol to $W$, we get a cq state $\rho_{XB}=\sum_{x\in\mX}p_X(x)\ketbra{x}{x}_X\otimes \rho_B^x$.
Correlation detection is the task of discriminating between the correlated cq state and the uncorrelated product cq state. That is, we would like to discriminate between 
\bal
H_0: \mbox{Given state is }\rho^{\otimes n}_{XB},~~H_1: \mbox{Given state is }\rho^{\otimes n}_{X}\otimes \rho^{\otimes n}_{B}.
\eal
This problem has been studied in Refs.~\cite{Hayashi_2016_correlation_detection, Dasgupta_2025_universal_tester, girardi2025umlaut_information, girardi2025quantum_umlaut} using the techniques of composite hypothesis testing. Here, we consider universal detection, in which the test does not depend on the details of the cq channel $W$ but achieves the optimal discrimination rate.

\begin{proboxed}
    \begin{pro}\label{pro: universal correlation detection cq}
        For any $a>0$, there exists a test $P^n_{X^nB^n}(a)$ depending on $a$ such that, for any cq state $\rho_{XB}$, the inequalities
        \bal
        \Tr[\rho_{XB}^{\otimes n}(I-P^n_{X^nB^n}(a))]&\leq \frac{1}{{\rm poly}(n)}2^{nt(a-\pI^\downarrow_{1-t}(X:B)_\rho)}\\
        \Tr[\rho_X^{\otimes n}\otimes\rho_B^{\otimes n} P^n_{X^nB^n}(a)]&\leq {\rm poly}(n)\cdot 2^{-na}.
        \eal
        hold for any $t\in(0,1)$.
    \end{pro}
\end{proboxed}
Let us remark on an immediate corollary of Proposition~\ref{pro: universal correlation detection cq}.
In the state-aware correlation detection setting, it follows from quantum Stein's lemma that, for any $R<I(X:B)_\rho$, there exists a sequence of POVMs $\qty{M_n, I_n-M_n}$ depending on the description of the state $\rho_{XB}$ such that the type I error goes to $0$ as $n\to\infty$, while the type II error decays exponentially as $2^{-nR}$.
Proposition~\ref{pro: universal correlation detection cq} states that this can also be achieved with a fixed POVM designed independently of $\rho_{XB}$.

Note that the idea of constructing $P^n(a)$ was already discussed in Ref.~\cite[Theorem 1]{Dasgupta_2025_universal_tester}, but in a different setting. Specifically, the authors consider the arbitrarily varying scenario where the cq channel used to send $n$ messages may vary in every round, resulting in a doubly composite hypothesis testing problem where the null hypothesis is a tensor-product black box of cq states, while the alternative hypothesis consists of states with an arbitrary permutation-symmetric state on the register $B$.

\begin{proof}
    We choose the test $P^n_{X^nB^n}(a)$ as 
    \bal
    P^n_{X^nB^n}(a):&=\sum_{x^n\in\mX^n}\ketbra{x^n}{x^n}_{X^n}\otimes \qty{\sigma_{B^n}^{x^n}\geq 2^{na}\sigma_{B^n}^{U,n}}\\
    \eal
    Note that $P^n_{X^nB^n}(a)$ itself does not depend on the probability distribution.
    We also remark that this test was proposed in Ref.~\cite{Dasgupta_2025_universal_tester} to detect a different type of correlation.
    The type I error is bounded as follows:
    \bal
    \Tr[\qty(\rho_{XB})^{\otimes n}(I-P^n_{X^nB^n}(a))] &=\sum_{x^n\in\mX^n}p^n(x^n)\Tr[\rho^{x^n}_{B^n}\qty{\sigma^{x^n}_{B^n}<2^{na}\sigma^{U,n}_{B^n}}]\\
    &\leq \sum_{x^n\in\mX^n}p^n(x^n)\Tr[\rho^{x^n}_{B^n}2^{nta}(\sigma^{x^n}_{B^n})^{-t}(\sigma^{U,n}_{B^n})^t]
    \eal
    for any $t\in (0,1)$, which follows from $[\sigma^{x^n}_{B^n},\sigma^{U,n}_{B^n}]=0$. Since it holds that $\rho_{B^n}^{x^n}\leq {\rm poly}(n)\sigma^{x^n}_{B^n}$ (we write ${\rm poly}(n)=(n+1)^{\aX(d+2)(d-1)/2}$) as in \eqref{eq:universal state bound}, we have
    \bal
     \sum_{x^n\in\mX^n}p^n(x^n)\Tr[\rho^{x^n}_{B^n}2^{nta}(\sigma^{x^n}_{B^n})^{-t}(\sigma^{U,n}_{B^n})^t]&\leq 
     \frac{2^{nta}}{{\rm poly}(n)}\sum_{x^n\in\mX^n}p^n(x^n)\Tr[(\rho^{x^n}_{B^n})^{1-t}(\sigma^{U,n}_{B^n})^t]\\
     &=\frac{2^{nta}}{{\rm poly}(n)}\Tr[\qty(\sum_{x\in\mX}p(x)(\rho_B^x)^{1-t})^{\otimes n}\qty(\sigma^{U,n}_{B^n})^t].
    \eal
    Now, due to Lemma~\ref{lem: variational form}, we have
    \bal
    \frac{2^{nta}}{{\rm poly}(n)}\Tr[\qty(\sum_{x\in\mX}p(x)(\rho_B^x)^{1-t})^{\otimes n}\qty(\sigma^{U,n}_{B^n})^t]&\leq 
    \frac{2^{nta}}{{\rm poly}(n)}\Tr[\qty(\qty(\sum_{x\in\mX}p(x)(\rho_B^x)^{1-t})^{\otimes n})^{\frac{1}{1-t}}]^{1-t}\\
    &=\frac{1}{{\rm poly}(n)}2^{nt(a-\pI^\downarrow_{1-t}(X:B)_\rho)}.
    \eal

    Now, let us consider the type II error. 
    Noting that 
    \bal
    \Tr[\rho_X^{\otimes n}\otimes\rho_B^{\otimes n} P^n_{X^nB^n}(a)]=\Tr[\rho_X^{\otimes n}\otimes\rho_B^{\otimes n} \qty{\sigma^{U,p}_{X^nB^n}\geq 2^{na}\rho_X^{\otimes n}\otimes \sigma_{B^n}^{U,n}}]
    \eal
    and  $\rho^{\otimes n}_{B}\leq {\rm poly}(n)\sigma_{B^n}^{U,n}$, we have
    \bal
    \Tr[\rho_X^{\otimes n}\otimes\rho_B^{\otimes n} P^n_{X^nB^n}(a)]&\leq {\rm poly}(n)\Tr[\rho_X^{\otimes n}\otimes \sigma_{B^n}^{U,n}P^n_{X^nB^n}(a)]\\
    &\leq {\rm poly}(n)\cdot 2^{-na}\Tr[\sigma_{X^nB^n}^{U,p}P^n_{X^nB^n}(a)]\leq {\rm poly}(n)\cdot 2^{-na},
    \eal
    where the second inequality is due to the definition of $P^n_{X^nB^n}(a)$.
\end{proof}
In Ref.~\cite{Dasgupta_2025_universal_tester}, they remark that the test $P^n_{X^nB^n}(a)$ constructed in the proof satisfies
\bal
P^n_{X^nB^n}(a):&=\sum_{x^n\in\mX^n}\ketbra{x^n}{x^n}_{X^n}\otimes \qty{\sigma_{B^n}^{x^n}\geq 2^{na}\sigma_{B^n}^{U,n}}=\qty{\sigma^{U,p}_{X^nB^n}\geq 2^{na}\rho_X^{\otimes n}\otimes \sigma_{B^n}^{U,n}}.
\eal
However, this holds only in the case where the probability distribution $p(x)$ satisfies $p(x)>0$ for every $x\in \mX$.

\subsection{Semi-universal correlation detection for fully quantum states}
Fix a quantum state $\rho_{AB}$ on a composite system $AB$. We assume that we are not given the full description of the global state $\rho_{AB}$, but are given knowledge about the marginal $\rho_B$.
In this setting, we can construct a binary POVM achieving the optimal performance guaranteed by quantum Stein's lemma without the full description of $\rho_{AB}$.

\begin{proboxed}
    \begin{pro}\label{pro: semi universal detector}
        For any $a>0$ and any quantum state $\rho_B\in\mD(\mH_B)$, there exists a projector $P^n(a, \rho_B)$ such that, for any bipartite quantum state $\rho_{AB}$ satisfying $\Tr_A\rho_{AB}=\rho_B$ and for any $t\in(0,1)$, the inequalities
        \bal
        \Tr\qty[\rho_{AB}^{\otimes n}\qty(I-P^n(a,\rho_B))]&\leq \frac{1}{{\rm poly}(n)}2^{-nt(\pI^{\downarrow}_{1-t}(B:A)_\rho-a)},\\
        \Tr[\rho_A^{\otimes n}\otimes\rho_B^{\otimes n} P^n(a,\rho_B)]&\leq {\rm poly}(n)2^{-na}
        \eal
        hold.
    \end{pro}
\end{proboxed}
Proposition~\ref{pro: semi universal detector} implies that, for any $a<I(A:B)_\rho$, the test $P^n(a,\rho_B)$ can achieve exponential decay $\sim 2^{-na}$ of the type II error probability, while the type I error also decays exponentially to zero, even though the test does not depend on the full description of $\rho_{AB}$.

\begin{proof}
    We define a test $P^n(a,\rho_B)$ as 
    \bal
    P^n(a,\rho_B)\coloneqq\qty{\sigma_{A^nB^n}^{U,n}> 2^{na}\sigma_{A^n}^{U,n} \otimes \rho_{B}^{\otimes n}}.
    \eal
    First, we bound the type I error.
    \bal
    \Tr\qty[\rho_{AB}^{\otimes n}\qty(I-P^n(a,\rho_B))]&=\Tr\qty[\rho_{AB}^{\otimes n}\qty{\sigma_{A^nB^n}^{U,n}\leq 2^{na}\sigma_{A^n}^{U,n}\otimes \rho_{B}^{\otimes n}}]\\
    &\leq 2^{tna}\Tr\qty[\rho_{AB}^{\otimes n}\qty(\sigma_{A^nB^n}^{U,n})^{-t}\qty(\sigma_{A^n}^{U,n}\otimes \rho_{B}^{\otimes n})^t],\\
    \eal
    for any $t\in (0,1)$. Here, we used the fact that $[\sigma_{A^nB^n}^{U,n},\sigma_{A^n}^{U,n}\otimes \rho_{B}^{\otimes n}]=0$, which implies $\qty{\sigma_{A^nB^n}^{U,n}\leq 2^{na}\sigma_{A^n}^{U,n}\otimes \rho_{B}^{\otimes n}}\leq 2^{tna}\qty(\sigma_{A^nB^n}^{U,n})^{-t}\qty(\sigma_{A^n}^{U,n}\otimes \rho_{B}^{\otimes n})^t$. From $\rho_{AB}^{\otimes n}\leq {\rm poly}(n)\sigma_{A^nB^n}^{U,n}$, we further have 
    \bal
    2^{tna}\Tr\qty[\rho_{AB}^{\otimes n}\qty(\sigma_{A^nB^n}^{U,n})^{-t}\qty(\sigma_{A^n}^{U,n}\otimes \rho_{B}^{\otimes n})^t]&\leq
    \frac{2^{tna}}{{\rm poly}(n)}\Tr\qty[\qty(\rho_{AB}^{\otimes n})^{1-t}\qty(\sigma_{A^n}^{U,n}\otimes \rho_{B}^{\otimes n})^t]\\
    &=\frac{2^{tna}}{{\rm poly}(n)}\Tr\qty[\qty(\rho_{AB}^{\otimes n})^{1-t}\qty(I_{A^n}\otimes \rho_{B}^{\otimes n})^t\qty(\sigma_{A^n}^{U,n}\otimes I_{B^n})^t]\\
    &=\frac{2^{tna}}{{\rm poly}(n)}\Tr[\Tr_{B^n}\qty[\qty(\rho_{AB}^{\otimes n})^{1-t}(\rho_{B}^{\otimes n})^t]\qty(\sigma_{A^n}^{U,n})^t].
    \eal
    Due to Lemma~\ref{lem: variational form}, we have
    \bal
    \frac{2^{tna}}{{\rm poly}(n)}\Tr[\Tr_{B^n}\qty[\qty(\rho_{AB}^{\otimes n})^{1-t}(\rho_{B}^{\otimes n})^t]\qty(\sigma_{A^n}^{U,n})^t]
    &\leq \frac{2^{tna}}{{\rm poly}(n)}\Tr[\Tr_{B^n}\qty[\qty(\rho_{AB}^{\otimes n})^{1-t}(\rho_{B}^{\otimes n})^t]^{\frac{1}{1-t}}]^{1-t}.
    \eal
    Now, by the quantum Sibson identity,
    \bal
    \pI_\alpha^{\downarrow}(B:A)_\rho=\inf_{\sigma_A\in\mD(\mH_A)}\pD_\alpha(\rho_{AB}\|\sigma_A\otimes \rho_B)=\frac{\alpha}{\alpha-1}\log \Tr[\Tr_B[\rho_{AB}^\alpha\rho_B^{1-\alpha}]^{\frac{1}{\alpha}}]
    \eal
    it holds that
    \bal
     \frac{2^{tna}}{{\rm poly}(n)}\Tr[\Tr_{B^n}\qty[\qty(\rho_{AB}^{\otimes n})^{1-t}(\rho_{B}^{\otimes n})^t]^{\frac{1}{1-t}}]^{1-t}=\frac{1}{{\rm poly}(n)}2^{-nt(\pI^{\downarrow}_{1-t}(B:A)_\rho-a)}
    \eal
    We bound the type II error.
    \bal
    \Tr[\rho_A^{\otimes n}\otimes\rho_B^{\otimes n} P^n(a,\rho_B)]
    &\leq{\rm poly}(n)\Tr[\sigma_{A^n}^{U,n}\otimes \rho_B^{\otimes n} P^n(a,\rho_B)]\\
    &\leq {\rm poly}(n)2^{-na}\Tr[\sigma_{A^nB^n}^{U,n}P^n(a,\rho_B)]\\
    &\leq {\rm poly}(n)2^{-na}.
    \eal
\end{proof}

Here, we remark on the achievable error exponent of the correlation detection in the state-aware/state-agnsotic scenarios.
If one is informed of the complete information about the given state $\rho_{AB}$, then the optimal exponent of the type I decay where the exponent of the type II error is kept no smaller than $r>0$ is characterized as~\cite{nagaoka_2006_converse_theorem, Hayashi_2007_error_exponent, audenaert_2008}
\bal
\sup_{0<\alpha<1}\frac{\alpha-1}{\alpha}\qty(r-\pD_\alpha(\rho_{AB}\|\rho_A\otimes \rho_B)).
\eal
On the other hand, in either the fully quantum case or the classical-quantum case in Proposition\ref{pro: universal correlation detection cq} and Proposition~\ref{pro: semi universal detector}, the error exponent of the type I error is not optimal. At this point, it is not clear whether the universal correlation detectors above also achieve the optimal error exponent.

\section{Warm-up 1: Universal cq channel coding with unknown shared-randomness assistance}\label{sec: cq channel coding}

We now apply the universal correlation detection technique established in Sec.~\ref{sec:correlation detection} to characterize the performance of communication protocols.
We begin our discussion with classical-quantum channel coding.

In channel-aware classical-quantum coding, the channel capacity $C_{\rm cq}(W)$ of the classical-quantum channel $W$ is characterized as~\cite{Holevo_1996_capacity, SW_theorem}
\bal
C_{\rm cq}(W)=\sup_{p}I(X:B)_{\rho_{XB}},~~~\rho_{XB}=\sum_{x\in \mX}p(x)\ketbra{x}{x}\otimes \rho_B^x.
\eal
In the channel-agnostic scenario, it is known that, for any fixed input distribution $p$, one can construct a universal coding scheme achieving the communication rate $I(X:B)_{\rho_{XB}}$~\cite{Hayashi_2009_universal_coding, Matsuura_2025_universal_resolvability}.

Here, we approach the universal classical-quantum channel coding problem based on correlation detection, providing an alternative proof that the mutual information with a fixed input distribution is achievable. 
We remark that our construction requires shared (but unknown) randomness as an additional resource, and thus is slightly weaker than the results in Refs.~\cite{Hayashi_2009_universal_coding, Matsuura_2025_universal_resolvability}.
Nevertheless, our construction---combining correlation detection and position-based decoding~\cite{Anshu_2019_building}---serves as a starting point for the further extension of this idea to various communication settings as we see later.

We first formulate shared-randomness-assisted classical-quantum channel coding.
\begin{defnboxed}
\begin{defn}[Shared-randomness-assisted classical-quantum channel coding]
         Let $W:\mX\to \mD(\mH), x\mapsto \rho^x$ be a classical-quantum channel, and let $\mM=\qty{1,\ldots, M}$ be a message set containing $M\coloneqq\abs{\mM}$ elements. Suppose that the sender (Alice) and the receiver (Bob) share a random variable $S\sim p_X$. We also consider a map $f:\mM\times \mS\to\mX^n$ and a family of POVMs $\qty{\qty{\Lambda_m^{s}}_{m}}_s$.
         The triplet $\qty(p_X, f,\qty{\qty{\Lambda_m^{s}}_{m}}_s)$ is called a shared-randomness-assisted $(n, M, \ve)$-code if
         \bal
         1-\mbE_S\Tr[\Lambda^s_m\rho^{f(m, s)}]\leq \ve, ~~\forall m\in \mM.
         \eal
         Here, $\mbE_S$ denotes the average with respect to $S\sim p_X$.
         $R>0$ is called an achievable rate for a classical-quantum channel $W$ if, for any $\ve>0$ and sufficiently large $n$, there exists a shared-randomness-assisted $(n, 2^{nR},\ve)$-code. The classical-quantum channel capacity $C_{\rm cq}(W)$ is defined as the supremum of the achievable rates for $W$.
\end{defn}
\end{defnboxed}

We then show the following result by employing universal correlation detection as a decoder for position-based decoding.
\begin{thmboxed}
    \begin{thm}[Universal cq channel coding with unknown shared-randomness assistance]\label{thm: universal cq channel coding}
        There exists a sequence of shared-randomness-assisted channel coding schemes achieving the communication rate $R$ for any classical-quantum channel and any probability distribution $p$ on $\mX$ from which one samples the shared randomness, satisfying $R<I(X:B)_\rho$, with $\rho_{XB}\coloneqq\sum_{x\in\mX}p(x)\ketbra{x}{x}_X\otimes \rho^x_B$.

        \begin{itemize}
    \item {\bf The sender and the receiver need to know}:
    \begin{enumerate}
        \item that the target communication rate $R$ satisfies $R<I(X:B)_\rho$.
    \end{enumerate}
    \item {\bf The sender and the receiver do not need to know}:
    \begin{enumerate}
        \item the full description of the classical-quantum channel $W$;
        \item the probability distribution $p$ on $\mX$;
        \item the exact values of mutual information $I(X:B)_\rho$, and the channel capacity $C_{\rm cq}(W)$.
    \end{enumerate}
\end{itemize}
    \end{thm}
\end{thmboxed}

\begin{proof}

Let $\mM=\qty{1,\ldots, M}$ be the message set with $M=2^{nR}$, and let $W:\mX\to \mD(\mH),~x\mapsto\rho^x$ be an unknown cq channel. 
We consider the situation where Alice (sender) and Bob (receiver) share $nM$ copies of the classical state
\bal
\rho_{X'X}=\sum_{x\in\mX}p_X(x)\ketbra{x}{x}_{X'}\otimes \ketbra{x}{x}_{X},
\eal
where Alice has the system $X'$ and Bob has the system $X$. For convenience, we write $\widetilde{X}_m\coloneqq X_{n(m-1)+1}\cdots X_{nm}$ for any $m=1, \ldots M$.
If Alice would like to communicate $m\in\mM$, Alice sends $\tX_m$ via $W^{\otimes n}$. Bob's resulting reduced state is 
\bal
\rho_{\tX^{M}B^n}^m=\rho_{\tX_1}\otimes\cdots \otimes \rho_{\tX_{m-1}} \otimes \rho_{\tX_{m+1}}\otimes \cdots\otimes \rho_{\tX_M}\otimes \rho_{\tX_mB^n}.
\eal

Now, let us construct the decoder.
We first define an operator $\Gamma_{\tX^{M}B^n}^m$ for any $m\in\mM$ as 
\bal
\Gamma_{\tX^{M}B^n}^m\coloneqq I_{\tX_1}\otimes \cdots \otimes I_{\tX_{m-1}}\otimes I_{\tX_{m+1}}\otimes \cdots\otimes I_{\tX_M}\otimes P^n_{\tX_mB^n}(a).
\eal
Note that $\Gamma_{\tX^{M}B^n}^m$ satisfies
\bal
\Tr[\Gamma_{\tX^{M}B^n}^m \rho_{\tX^{M}B^n}^m]&=\Tr[\rho_{XB}^{\otimes n}P^n_{\tX_mB^n}(a)],\\
\Tr[\Gamma_{\tX^{M}B^n}^{m'} \rho_{\tX^{M}B^n}^m]&=\Tr[\rho_X^{\otimes n}\otimes \rho_B^{\otimes n}P^n_{\tX_mB^n}(a)], \qquad m\neq m'.
\eal
$\Gamma_{\tX^{M}B^n}^1$ and $\Gamma_{\tX^{M}B^n}^m$ are related through a permutation $\pi\in\mfS_M$ over the systems $\tX^M$ such that $\pi(1)=m$ as 
\bal
V^\pi_{\tX^M}\Gamma_{\tX^{M}B^n}^1 \qty( V^\pi_{\tX^M})^\dagger=\Gamma_{\tX^{M}B^n}^m.
\eal

From the family $\qty{\Gamma_{\tX^{M}B^n}^m}_{m\in\mM}$ of operators, we define operators $\qty{\widetilde{\Lambda}_{\tX^{M}B^n}^m}_{m\in\mM}$
\bal
\widetilde{\Lambda}_{\tX^{M}B^n}^m=\qty(\sum_{m\in\mM}\Gamma_{\tX^{M}B^n}^m)^{-\frac{1}{2}} \Gamma_{\tX^{M}B^n}^m \qty(\sum_{m\in\mM}\Gamma_{\tX^{M}B^n}^m)^{-\frac{1}{2}}.
\eal
Here, note that 
\bal
\sum_{m\in\mM}\widetilde{\Lambda}_{\tX^{M}B^n}^m=\qty(\sum_{m'\in\mM}\Gamma_{\tX^{M}B^n}^{m'})^0\leq I.
\eal
We define our decoder
$\qty{\Lambda_{\tX^{M}B^n}^m}_{m\in\mM}$ by lifting $\qty{\widetilde{\Lambda}_{\tX^{M}B^n}^m}_{m\in\mM}$ to a valid measurement by 
\bal\label{eq: pretty good cq}
\Lambda_{\tX^{M}B^n}^m=\frac{1}{2^{nR}}\qty(I-\qty(\sum_{m\in\mM}\Gamma_{\tX^{M}B^n}^m)^0)+\widetilde{\Lambda}_{\tX^{M}B^n}^m.
\eal
Note that $\Lambda_{\tX^{M}B^n}^m$ satisfies
\bal
V_\pi^{\tX^M}\Lambda_{\tX^{M}B^n}^1 \qty( V_\pi^{\tX^M})^\dagger=\Lambda_{\tX^{M}B^n}^m.
\eal

Let us analyze the error probability. 
The error probability for the message $m$ is written as
\bal
\Tr[\rho_{\tX^{M}B^n}^m \qty(I-\Lambda_{\tX^{M}B^n}^m)].
\eal
Noting that $\Lambda_{\tX^{M}B^n}^m\geq \widetilde{\Lambda}_{\tX^{M}B^n}^m$ and applying Lemma~\ref{lem: Hayashi Nagaoka inequality}, we have
\bal
\Tr[\rho_{\tX^{M}B^n}^m \qty(I-\Lambda_{\tX^{M}B^n}^m)]&\leq\Tr[\rho_{\tX^{M}B^n}^m \qty(I-\widetilde{\Lambda}_{\tX^{M}B^n}^m)]\\
&\leq c_1\Tr[\rho_{\tX^{M}B^n}^m(I-\Gamma_{\tX^{M}B^n}^m)]+c_2\Tr[\rho_{\tX^{M}B^n}^m\sum_{m'\neq m} \Gamma_{\tX^{M}B^n}^{m'}]\\
&= c_1\Tr[\rho_{XB}^{\otimes n}(I-P_n(a))]+c_2(M-1)\Tr[\rho_X^{\otimes n}\otimes \rho_B^{\otimes n}P_n(a)].
\eal
Due to Proposition~\ref{pro: universal correlation detection cq}, for any $m\in\mM$, we have
\bal\label{eq: cq channel error}
\Tr[\rho_{\tX^{M}B^n}^m \qty(I-\Lambda_{\tX^{M}B^n}^m)]&\leq c_1 \frac{1}{{\rm poly}(n)}2^{-n\sup_{0< t< 1}t(\pI^\downarrow_{1-t}(X:B)_\rho-a)}+c_2{\rm poly}(n)\cdot 2^{-na+\log M}.
\eal
Let us take $M=2^{nR}$.
We can show that any $R<I(X:B)_\rho$ is achievable from Eq.~\eqref{eq: cq channel error} as follows:
Take $a_n\coloneqq R+1/\sqrt{n}$.
Since $\lim_{t\to0}\pI^\downarrow_{1-t}(X:B)_\rho=I(X:B)_\rho$ holds, there exists $t^*\in(0,1)$ such that, for a sufficiently large $n$, $a_n<\pI^\downarrow_{1-t^*}(X:B)_\rho$ holds, implying that
\bal
\sup_{0< t< 1}t\qty(\pI^\downarrow_{1-t}(X:B)_\rho-a_n)\geq t^*\qty(\pI^\downarrow_{1-t^*}(X:B)_\rho-a_n)>0
\eal
for sufficiently large $n$.
Combining this with $a_n>R$, we have $\Tr[\rho_{\tX^{M}B^n}^m \qty(I-\Lambda_{\tX^{M}B^n}^m)]\to 0$ for any $m\in\mM$, which concludes the proof.
\end{proof}

\section{Warm-up 2: Universal wiretap channel coding with unknown shared-randomness assistance}\label{sec: private channel}
Following the same argument as in Ref.~\cite{Wilde_2017_position_based} and combining it with convex splitting, we obtain universal wiretap channel coding.
We consider a wiretap channel $W:\mX\to \mD(\mH_B\otimes \mH_E), ~x\mapsto \rho_{BE}^x$. Moreover, we assume that Alice, Bob, and Eve share some randomness whose distribution might be unknown to Alice and Bob.

The goal of wiretap channel coding is to send a message $m\in\mM$ to Bob while keeping it inaccessible to Eve.

\begin{defnboxed}
\begin{defn}[Shared-randomness-assisted classical-quantum wiretap channel coding]
         Let $W:\mX\to \mD(\mH), x\mapsto \rho_{BE}^x$ be a wiretap channel, and let $\mM=\qty{1,\ldots, M}$ be a message set containing $M\coloneqq\abs{\mM}$ elements, and $\mK=\qty{1,\ldots, K}$ be the set of the keys. Suppose that the sender (Alice), the receiver (Bob), and the eavesdropper (Eve) share a random variable $S\sim p_X$. We also consider a map $f:\mM\times \mK\times \mS\to\mX^n$ and a family of POVMs $\qty{\qty{\Lambda_m^{s}}_{m}}_s$ on Bob's system.
         The triplet $(p_X, f,\qty{\qty{\Lambda_m^{s}}_{m}}_s)$ is called a shared-randomness-assisted $(n, M, \ve_e,\ve_s)$-code if the inequalities 
         \bal
         1-\mbE_S\Tr[\Lambda^s_m\rho_{B^n}^{f(m, s)}]&\leq \ve_e,\\
         \mbE_S\|\rho_{E^n}^{f(m,k, s)}-\rho_{E}^{\otimes n}  \|_1&\leq \ve_s
         \eal
         hold for every $m\in\mM$.
         $R>0$ is called an achievable rate if, for any $\ve_e,\ve_s>0$ and sufficiently large $n$, there exists a shared-randomness-assisted $(n, 2^{nR},\ve_e,\ve_s)$-code. The wiretap channel capacity $C_{WT}(W)$ of a wiretap channel $W$ is defined as the supremum of the achievable rates.
\end{defn}
\end{defnboxed}
In the channel-aware scenario, the classical-quantum wiretap channel capacity was established in Refs.~\cite{Devetak_2004_private_classical_capacity, cai_quantum_2004, Renes_Noisy_channel, Wilde_2017_position_based, Radhakrishnan_2017_one_shot_private} as
\bal
C_{WT}(W)=\lim_{n\to\infty}\frac{1}{n}\max_{T- X^n-B^nE^n}\qty(I(T:B^n)-I(T:E^n)).
\eal
Here, the maximization is taken over Markov chains $T-X^n-B^nE^n$, where $X^n$ and $B^nE^n$ are related by the action of the wiretap channel.

In the channel-agnostic scenario, previous studies have investigated the capacity achievable by wiretap channel coding without the full description of the channel~\cite{Boche_2014_secrecy_capacity_compound, Datta_2010_universal_private, Matsuura_2025_universal_resolvability}. 
In particular, a universal wiretap channel coding protocol with shared randomness was presented in Ref.~\cite{Datta_2010_universal_private}, and a fully fixed protocol for universal wiretap channel coding was constructed in Ref.~\cite{Matsuura_2025_universal_resolvability}. Here, we show that wiretap channel coding is possible with assistance from shared randomness whose distribution is unknown.
The setting we consider here is positioned somewhere between Ref.~\cite{Datta_2010_universal_private} and Ref.~\cite{Matsuura_2025_universal_resolvability}, but we provide a novel approach of the channel-agnostic wiretap channel coding based the idea to incorpolate the universal correlation detector with the position-based decoding, together with the convex split technique to ensure security against Eve~\cite{Wilde_2017_position_based}.

\begin{thmboxed}
    \begin{thm}[Universal wiretap channel coding with unknown shared-randomness assistance]
        There exists a sequence of shared-randomness-assisted wiretap channel codes achieving the communication rate $R_1-R_2$ for any wiretap channel $W$ and any probability distribution $p$ on $\mX$ from which one samples the shared randomness,  satisfying $R_1<I(X:B)_\rho$ and $R_2>I(X:E)_\rho$, where $\rho_{XBE}\coloneqq \sum_xp(x)\rho^x_{BE}$ is the output of the wiretap channel $W$.

        \begin{itemize}
    \item {\bf The sender and the receiver need to know}:
    \begin{enumerate}
        \item that the target rates $R_1, R_2$ satisfy $R_1<I(X:B)_\rho$ and $R_2>I(X:E)_\rho$.
    \end{enumerate}
    \item {\bf The sender and the receiver do not need to know}:
    \begin{enumerate}
        \item the full description of the classical-quantum wiretap channel $W$;
        \item the probability distribution $p$ on $\mX$;
        \item the exact values of the mutual information $I(X:B)_\rho,I(X:E)_\rho$, and the wiretap channel capacity $C_{\rm WT}(W)$.
    \end{enumerate}
\end{itemize}
    \end{thm}
\end{thmboxed}

\begin{proof}
Let the message set be $\mM\coloneqq\qty{1,\ldots, M}$ and the set of random keys be $\mK\coloneqq\qty{1, \ldots, K}$, with $M\coloneqq2^{n(R_1-R_2)}$ and $K\coloneqq2^{nR_2}$.
Let Alice, Bob, and Eve share $nMK$ copies of the classical state
\bal
\rho_{XX'X''}\coloneqq\sum_{x\in\mX}p(x)\ketbra{x}{x}_X\otimes \ketbra{x}{x}_{X'}\otimes \ketbra{x}{x}_{X''}.
\eal
Here, Bob holds the register $X$, Alice holds $X'$, and Eve holds $X''$.
As in the previous section, we write $X^n=\tX, X'^n=\tX', X''^n=\tX''$.
We label $nMK$ copies of the systems ${XX'X''}$ in lexicographic order, and write the state as 
\bal
\rho_{XX'X''}^{\otimes MK}=\rho_{\tX_{1,1}\tX_{1,1}'\tX_{1,1}''}\otimes \rho_{\tX_{1,2}\tX_{1,2}'\tX_{1,2}''} \otimes \cdots\otimes \rho_{\tX_{M,K}\tX_{M,K}'\tX_{M,K}''}.
\eal
 
Alice sends the message $m\in\mM$, together with the locally generated random key $k\in\mK$. To send the pair $(m,k)$, Alice chooses the classical slot $\tX_{m,k}$ and sends it through $n$ copies of the given wiretap channel $W^{\otimes n}$.

The reduced state $\rho_{\tX^{MK}\tX''^{MK}B^nE^n}$ on Bob's and Eve's systems is written as
\bal
\rho_{\tX^{MK}\tX''^{MK}B^nE^n}^{(m,k)}=\rho_{\tX_{1,1}\tX_{1,1}''}\otimes \cdots\otimes\rho_{\tX_{m,k-1}\tX_{m,k-1}''}\otimes \rho_{\tX_{m,k}\tX_{m,k}''B^n E^n}\otimes \rho_{\tX_{m,k+1}\tX_{m,k+1}''}\otimes \cdots\otimes \rho_{\tX_{M,K}\tX_{M,K}''}.
\eal
Bob's reduced state is
\bal
\rho_{\tX^{MK}B^n}^{(m,k)}=\rho_{\tX_{1,1}}\otimes \cdots\otimes\rho_{\tX_{m,k-1}}\otimes \rho_{\tX_{m,k}B^n }\otimes \rho_{\tX_{m,k+1}}\otimes \cdots\otimes \rho_{\tX_{M,K}}.
\eal
Following the same discussion as in Section~\ref{sec: cq channel coding}, Bob can decode $(m,k)$ using the measurement
\bal
\Lambda_{\tX^{MK}B^n}^{(m,k)}=\qty(\sum_{(m,k)\in\mM\times \mK}\Gamma_{\tX^{MK}B^n}^{(m,k)})^{-\frac{1}{2}} \Gamma_{\tX^{MK}B^n}^{(m,k)} \qty(\sum_{(m,k)\in\mM\times \mK}\Gamma_{\tX^{MK}B^n}^{(m,k)})^{-\frac{1}{2}},
\eal
where $\Gamma_{\tX^{MK}B^n}^{(m,k)}$ is an operator defined as
\bal
\Gamma_{\tX^{MK}B^n}^{(m,k)}\coloneqq I_{\tX_{1,1}}\otimes \cdots \otimes I_{\tX_{m, k-1}}\otimes I_{\tX_{m,k+1}}\otimes \cdots\otimes I_{\tX_{M,K}}\otimes P^n_{\tX_{m,k}B^n}(a).
\eal
Following the same argument as in Section~\ref{sec: cq channel coding}, Bob's decoding error probability is bounded as 
\bal
\Tr[\rho_{\tX^{MK}B^n}^{(m,k)} \qty(I-\Lambda_{\tX^{M}B^n}^{(m,k)})]&\leq c_1 \frac{1}{{\rm poly}(n)}2^{-n\sup_{0< t< 1}t(\pI^\downarrow_{1-t}(X:B)_\rho-a)}+c_2{\rm poly}(n)\cdot 2^{-na+\log MK}.
\eal

Now, let us discuss security against Eve.
When $(m,k)$ is fixed, Eve's reduced state is
\bal
\rho_{\tX^{MK}E^n}^{(m,k)}=\rho_{\tX_{1,1}}\otimes \cdots\otimes\rho_{\tX_{m,k-1}}\otimes \rho_{\tX_{m,k}E^n }\otimes \rho_{\tX_{m,k+1}}\otimes \cdots\otimes \rho_{\tX_{M,K}}.
\eal
Here, we relabel $X''$ as $X$.
However, since Alice chooses $k$ uniformly at random, we have
\bal
{}&\rho_{\tX^{MK}E^n}^m\coloneqq\frac{1}{K}\sum_{k=1}^K \rho_{\tX^{MK}E^n}^{(m,k)}\\
&=\rho_{\tX_{1,1}}\otimes \cdots\otimes\rho_{\tX_{m-1,K}}\otimes \qty[\frac{1}{K}\sum_{k=1}^K\rho_{\tX_{m,1}}\otimes \cdots\otimes\rho_{\tX_{m,k-1}}\otimes \rho_{\tX_{m,k}E^n}\otimes\rho_{\tX_{m,k+1}}\cdots\otimes\rho_{\tX_{m,K}}]\otimes \rho_{\tX_{m+1,1}}\otimes \cdots\otimes\rho_{\tX_{M,K}}.
\eal

Due to the convex splitting lemma in Lemma~\ref{lem: unipartite convex splitting}, we have
\bal
\left\|\frac{1}{K}\sum_{k=1}^K\rho_{\tX_{m,1}}\otimes \cdots\otimes\rho_{\tX_{m,k-1}}\otimes \rho_{\tX_{m,k}E^n}\otimes\rho_{\tX_{m,k+1}}\cdots\otimes\rho_{\tX_{m,K}}- \rho_{\tX^K}\otimes \rho^{\otimes n}_{E}\right\|_1\leq  2\cdot2^{-n\sup_{1\leq \alpha \leq 2}\frac{\alpha-1}{\alpha}(\log K-\sI^\downarrow_\alpha(X:E)_\rho)}.
\eal
Here, note that the random variable $S$ corresponds to the classical register $X$ in this case.
Since the trace distance is invariant under taking the tensor product with a fixed state, i.e., $\|\rho-\sigma\|_1=\|\rho\otimes\tau-\sigma\otimes \tau\|_1$,
 it holds that 
\bal
 \mbE_S\|\rho_{E^n}^{f(m,k, s)}-\rho_{E}^{\otimes n}  \|_1=\|\rho_{\tX^{MK}E^n}^m -\rho_{\tX^{MK}}\otimes \rho^{\otimes n}_{E}\|_1\leq 2\cdot 2^{-n\sup_{1\leq \alpha \leq 2}\frac{\alpha-1}{\alpha}(\log K-\sI^\downarrow_\alpha(X:E)_\rho)}.
\eal
Let us show that both error criteria can be satisfied. 
Note that we take $K=2^{nR_2}$ with $R_2>I(X:E)_\rho$. Since $\lim_{\alpha\to 1}\sI^\downarrow_\alpha(X:E)_\rho=I(X:E)_\rho$ holds and $\sI^\downarrow_\alpha(X:E)_\rho$ is monotonically increasing in $\alpha$, there exists an $\alpha^*$ satisfying $R_2-\sI^\downarrow_{\alpha^*}(X:E)_\rho>0$. From this, we can see that $\|\rho_{\tX^{MK}E^n}^m -\rho_{\tX^{MK}}\otimes \rho^{\otimes n}_{E}\|_1$ decays exponentially in the limit $n\to\infty$.
On the other hand, following the discussion of Theorem~\ref{thm: universal cq channel coding}, we can see that the decoding error for the pair $(m,k)$ of the message $m$ and the key $k$ decays exponentially as long as $\frac{1}{n}\log MK=R_1<I(X:B)_\rho$ holds, which concludes the proof.

\end{proof}

\section{Universal one-way secret key distillation from classical-quantum-quantum states}
Secret key distillation is the task of sharing a maximally correlated classical state between two parties, Alice and Bob, while keeping it uncorrelated with the eavesdropper's system~\cite{Horodecki_2005_secret_key, Devetak_2005_distillation, Christandl_2007_unifying_classical_quantum, khatri_2021_second_order}. Since the eavesdropper has no access to the state shared by Alice and Bob, this shared state can serve as a secret key between the two parties, which can be used in a one-time pad, an information-theoretically secure communication protocol.

Suppose that Alice, Bob, and the eavesdropper Eve have a tripartite classical-quantum-quantum state
\bal
\rho_{XBE}\coloneqq\sum_{x\in\mX}p(x)\ketbra{x}{x}_X\otimes \rho_{BE}^x,
\eal
where Alice holds the classical state. We assume that Alice knows the probability distribution $p(x)$, but Alice and Bob do not know anything else about the shared state $\rho$.
This assumption can be justified, for instance, in the following setting:
Suppose that Alice has access to a classical-quantum wiretap channel $W:\mX\to \mD(\mH_B\otimes \mH_E), ~x\mapsto \rho^{BE}_x$, and inputs symbols sampled from the probability distribution $p(x)$. Then, the resulting state is exactly $\rho_{XBE}$.

Specifically, we consider secret key distillation under local operations and one-way public communication, where Alice announces classical information to Bob and Eve.
The goal of secret key distillation is to convert the tripartite state $\rho_{XBE}$ to the state $\Phi^K_{K_AK_B}\otimes \rho_{EL}$, where $\Phi^K_{K_AK_B}\coloneqq\frac{1}{K}\sum_{i=1}^K\ketbra{i}{i}_{K_A}\otimes \ketbra{i}{i}_{K_B}$ is the maximally correlated classical state. The auxiliary system $L$ is the classical register storing the information announced by Alice, and $\rho_{EL}$ is the marginal of the final state. 
\begin{defnboxed}
    \begin{defn}[One-way secret key distillation]
        A one-way secret key distillation protocol consists of an encoder $\mE_{X\to K_AL}$ on Alice's side and a decoder $\mD_{LB\to K_B}$ on Bob's side. Given a quantum state $\rho_{XBE}$, the final state is $\rho_{K_AK_BEL}\coloneqq\mD\circ\mE(\rho_{XBE})$. Given a quantum state $\rho_{XBE}$, a secret key distillation protocol $(\mE,\mD)$ is called a $( K, \ve)$-protocol if the resulting state $\rho_{K_AK_BEL}$ satisfies
        \bal\label{eq:  def security key distillation}
        \left\|\rho_{K_AK_BEL}-\Phi^K_{K_AK_B}\otimes \rho_{EL}\right\|_1\leq \ve,
        \eal
        where $\Phi^K_{K_AK_B}\coloneqq\frac{1}{K}\sum_{i=1}^K\ketbra{i}{i}\otimes \ketbra{i}{i}$ is the maximally correlated state shared by Alice and Bob, and $\rho_{EL}\coloneqq\Tr_{K_AK_B}\rho_{K_AK_BEL}$ is the reduced state on Eve's systems.

        The one-way distillable key $K_D^\rightarrow(\rho_{XBE})$ of the tripartite state $\rho_{XBE}$ is defined as the supremum of the rates $R>0$ such that, for any $\ve>0, \delta>0$ and sufficiently large $n$, there exists a $(2^{n(R-\delta)},\ve)$ secret key distillation protocol for $\rho^{\otimes n}_{XBE}$.
    \end{defn}
\end{defnboxed}
We remark that the definition of secret key distillation given in Eq.~\eqref{eq:  def security key distillation} is known to satisfy composability, implying that one can use secret key distillation as a primitive within a larger cryptographic protocol~\cite{Renner_2004_universally_composable, Benor_2004_universal_composable, Portmann_2022_security}.

In Ref.~\cite{Devetak_2005_distillation}, the one-way distillable key from the classical-quantum-quantum state $\rho_{XBE}$ is characterized as 
\bal\label{eq: state aware key distillaiton}
K^\rightarrow_D(\rho)&=\lim_{n\to\infty}\frac{1}{n}K^{(1)}(\rho^{\otimes n}),\\
K^{(1)}(\rho)&\coloneqq\max_{T-U-X}\qty[I(U:B|T)_{\rho'}-I(U:E|T)_{\rho'}],
\eal
where $\rho'$ is the classical-quantum-quantum state obtained by applying classical processing to the classical register.

In the following theorem, we construct a one-way secret key distillation protocol from a classical-quantum-quantum state achieving the distillable key, which works only with prior knowledge about the classical part $p(x)$.
\begin{thmboxed}
    \begin{thm}[Universal one-way secret key distillation from the output of an unknown wiretap channel]\label{thm: universal secret key distillation}
        For any probability distribution $p$ on $X$, there exists a sequence of one-way secret key distillation protocols depending only on $p$ that achieve the key distillation rate $R_1-R_2$ for any tripartite state $\rho_{XBE}\coloneqq\sum_{x\in\mX}p(x)\otimes \rho_{BE}^x$ and any Markov chain $X\to U\to T$ satisfying $R_1<I(U:B|E')_{\rho'}$ and $R_2>I(U:E|E')_{\rho'}$. Here, $\rho'_{UBEB'E'}$ is the tripartite state defined as
        \bal
        \rho'_{UBEB'E'}:=\sum_{t,u,x}R(t|u)Q(u|x)p(x)\ketbra{u}{u}^U\otimes \rho^x_{BE}\otimes \ketbra{t}{t}^{B'}\otimes \ketbra{t}{t}^{E'},
        \eal
        where $Q(u|x), R(t|u)$ are conditional probability distribution, respectively.
        \begin{itemize}
    \item {\bf The sender and the receiver need to know}:
    \begin{enumerate}
        \item the target rates $R_1, R_2$ satisfy $R_1<I(U:B|E')_{\rho'}$ and $R_2>I(U:E|E')_{\rho'}$;
        \item The probability distribution $p$ on $\mX$.
    \end{enumerate}
    \item {\bf The sender and the receiver do not need to know}:
    \begin{enumerate}
        \item the full description of the classical-quantum-quantum state $\rho_{XBE}$;
        \item the exact values of the mutual information $I(U:B|E')_{\rho'}, I(U:E|E')_{\rho'}$, and the one-way distlllable key $K^\to_D(\rho_{XBE})$.
    \end{enumerate}
\end{itemize}
    \end{thm}
\end{thmboxed}
In Ref.~\cite{Boche_2016_secret_key}, the performance of one-way secret key distillation is studied in the compound setting. Although our construction does not directly allow us to maximize over the possible Markov chains because we do not know the input state, we present a much simpler proof of the existence of a universal one-way secret key distillation protocol. Moreover, this result strengthens the previous results in Ref.~\cite{khatri_2021_second_order} on compound secret key distillation with a fixed marginal.

\begin{proof}
    Following the discussion in the proof of Theorem 6 in Ref.~\cite{Devetak_2005_distillation}, it suffices to show that one can universally achieve the key distillation rate $I(X:B)_\rho-I(X:E)_\rho$ from the state $\rho_{XBE}$.
    The idea of the universal construction follows the strategy in Ref.~\cite{khatri_2021_second_order}. 
\begin{enumerate}
    \item Alice prepares $n$ copies of the state $\rho_{XBE}$ and duplicates the classical part, resulting in $\rho_{XX'X''BE}^{\otimes n}$. She labels her register $X^n$ by $(k,r)$, where $k\in\mK$ is the key which she would like to send to Bob, and $r\in\mR$ is an additional randomizer to guarantee security against Eve.
    \item Alice also prepares $\abs{\mK}\abs{\mR}-1$ blocks of a dummy classical state $\rho_{XX'X''}^{\otimes n}$ using her knowledge of the input distribution. Then, she sends all $\abs{\mK}\abs{\mR}$ blocks of $X'$ and $X''$ through the public channel in lexicographic order, while putting the registers $X'$ and $X''$ of the state $\rho_{XX'X''BE}^{\otimes n}$ in the $(k,r)$\,th block.
    \item Bob decodes the pair $(k,r)$ using the universal correlation detector, while Eve cannot due to convex splitting. Following the discussion in Section~\ref{sec: private channel}, we obtain the key distillation rate $I(X:B)_\rho-I(X:E)_\rho$.
\end{enumerate}
\end{proof}

As discussed above, Theorem~\ref{thm: universal secret key distillation} shows that the optimal one-way key distillation rate discussed in Ref.~\cite{Devetak_2005_distillation} is possible universally. 
Our statement does not maximize over all possible Markov chains for preprocessing and regularization. One can, of course, do so without any information about the given state. However, experimenters without state information do not know how to preprocess the classical part or how large to choose the blocklength $n$ to achieve a key distillation rate beyond the single-letter expression in Eq.~\eqref{eq: state aware key distillaiton}.

Several problems along this line remain open: starting from a fully quantum tripartite state, the optimal protocol for one-way secret key distillation is for Alice to perform a POVM and then apply the optimal protocol for the resulting classical-quantum-quantum post-measurement state. In the state-aware scenario, one can tailor a POVM to maximize the distillable key, while in the state-agnostic setting one cannot.
Moreover, designing a universal protocol for two-way secret key distillation remains open for a similar reason. The optimal performance of two-way secret key distillation is achieved by applying preprocessing using two-way local operations and public communication, followed by the one-way secret key distillation protocol for fully quantum states~\cite{Khatri_textbook}. However, in the state-agnostic scenario, we do not know which preprocessing to apply to optimize the distillable key.

\section{Universal entanglement-assisted classical communication}
Entanglement-assisted classical communication is the task of sending classical information through a given noisy quantum channel, while the sender and the receiver have access to an unlimited amount of preshared entanglement. In Refs.~\cite{Bennett_2002_Entanglement_assisted, Holevo_2002_entanglement_assisted}, it is shown that the capacity of entanglement-assisted classical communication is fully characterized by the mutual information $I(\mN)\coloneqq\sup_{\psi_{AC}}I(B:C)_{\mN(\psi)}$ of the quantum channel $\mN:\mD(\mH_A)\to\mD(\mH_B)$.

Entanglement-assisted classical communication in the channel-agnostic scenario has been studied in the compound setting~\cite{boche_entanglement-assisted_2017, Berta_2017_compound_entanglement_assisted}, where one is given a quantum channel from a fixed set and maximizes the worst-case communication rate without complete information about the given channel. It is shown in Ref.~\cite{boche_entanglement-assisted_2017, Berta_2017_compound_entanglement_assisted} that the worst-case communication rate for the set $\mI\coloneqq\qty{\mN_s}_{s\in S}$ with the index set $S$ is characterized as
\bal
C_{\rm EA}(\mI)\coloneqq\sup_{\rho_{AC}\in\mD(\mH_{AC})}\inf_{s\in S}I(B:C)_{\mN_s(\rho)}
\eal
Furthermore, in Ref.~\cite{Anshu_2019_hypothesis_testing}, the one-shot capacity of entanglement-assisted classical communication in the compound setting is investigated under the assumption that the set of possible channels is finite.

We show that universal entanglement-assisted classical communication is achievable using the semi-universal correlation detector in Proposition~\ref{pro: semi universal detector}. This recovers the result of Ref.~\cite{boche_entanglement-assisted_2017, Berta_2017_compound_entanglement_assisted} with a simpler proof.

\subsection{Fixed-length communication}
We first formalize entanglement-assisted classical communication.
\begin{defnboxed}
    \begin{defn}[Entanglement-assisted classical communication]
        Let $\mN:\mD(\mH_A)\to\mD(\mH_B)$ be a quantum channel. An $(M, \ve)$ entanglement-assisted classical communication protocol for a quantum channel $\mN$ is specified by a tuple $(\theta_{EC},\qty{\mE^m_{E\to A}}_m,\qty{\Lambda_{BC}^m}_m)$ satisfying 
        \bal
        1-\Tr\qty[\Lambda^{m}_{BC}\qty(\mN_{A\to B}\circ\mE^m_{E\to A}( \theta_{EC}))]\leq \ve,~~\forall m\in\mM=[M]
        \eal
        Here, $\theta_{EC}$ is the pre-shared entanglement between the sender and the receiver, $\mE^m$ is the encoder which encodes the message $m$, and the POVM $\qty{\Lambda_{BC}^m}_m$ is the decoder.
        We call $R>0$ an achievable rate for the channel $\mN$ if, for any $\ve>0$ and sufficiently large $n$, there exists a $(2^{nR},\ve) $ entanglement-assisted classical communication protocol for the quantum channel $\mN^{\otimes n}$.
        The entanglement-assisted classical capacity $C_{\rm EA}(\mN)$ of a quantum channel $\mN$ is defined as the supremum of the achievable rates.
    \end{defn}
\end{defnboxed}

As mentioned above, the entanglement-assisted classical capacity of $\mN$ is represented as
\bal
C_{\rm EA}(\mN)=I(\mN)\coloneqq\sup_{\psi_{AC}}I(B:C)_{\mN(\psi)}.
\eal
Here, we show that the channel mutual information with a fixed preshared entangled state is universally achievable, using the semi-universal correlation detector.

\begin{thmboxed}
    \begin{thm}[Universal entanglement-assisted classical communication]\label{thm: universal EACC fixed length}
        Let $\theta_{AC}$ be an arbitrary bipartite state. Then, there exists a fixed encoder-decoder pair for entanglement-assisted classical communication that achieves the communication rate $R$ for any quantum channel $\mN:\mD(\mH_A)\to\mD(\mH_{B})$ with $R<I(B:C)_{\mN(\theta)}$.

    \begin{itemize}
    \item {\bf The sender and the receiver need to know}:
    \begin{enumerate}
        \item The full-description of the shared state $\theta_{EC}$;
        \item the target rate $R$ satisfies $R<I(B:C)_{\mN(\theta)}$.
    \end{enumerate}
    \item {\bf The sender and the receiver do not need to know}:
    \begin{enumerate}
        \item the full description of the channel $\mN$;
        \item the exact values of the mutual information $I(B:C)_{\mN(\theta)}$, and the entanglement-assisted classical capacity $C_{\rm EA}(\mN)$.
    \end{enumerate}
\end{itemize}
    \end{thm}
\end{thmboxed}

\begin{proof}
Let us first describe how the universal protocol for entanglement-assisted classical communication works.
Let $\mM_n=\qty{1,\ldots,M_n}$ be the message set with $M_n=2^{nR}$.
Before the protocol, Alice and Bob share $nM_n$ copies of $\theta_{AC}$. As before, we write $\tA\coloneqq A^n,\tB\coloneqq B^n$ and so on.
Alice encodes her message $m\in\mM_n$ by choosing the $m$-th slot $\tA_m$, and she sends it through $\mN^{\otimes n}$.
The resulting state on Bob's side is 
\bal
\eta_{\tB\tC^{M_n}}\coloneqq\theta_{\tC_1}\otimes \cdots\otimes (\mN^{\otimes n}\otimes {\rm id}_{\tC})(\theta_{\tA\tC_m})\otimes \cdots\otimes \theta_{\tC_{M_n}}.
\eal
Now, Bob performs a measurement $\qty{\Lambda_{\tB\tC}^m}_{m=1}^{M_n}$, which is constructed just as in Eq.~\eqref{eq: pretty good cq}, by replacing $P^n(a)$ with $P^n(a,\theta_{C})$.
Following the same discussion as in the proof of Theorem~\ref{thm: universal cq channel coding}, we have
\bal
1-\Tr\qty[\Lambda^{m}_{\tB\tC}\qty(\mN_{A\to B}\circ\mE^m_{E\to A}( \theta_{EC}))]&\leq c_1 \frac{1}{{\rm poly}(n)}2^{-n\sup_{0< t< 1}t(\pI^\downarrow_{1-t}(C:B)_{\mN(\theta)}-a)}+c_2{\rm poly}(n)\cdot 2^{-na+\log M}.
\eal
Again, taking $a_n\coloneqq R+1/\sqrt{n}$ and following the same discussion as in the proof of Theorem~\ref{thm: universal cq channel coding}, we obtain the claim.
\end{proof}

Here, note that universal entanglement-assisted classical communication implies universal entanglement-assisted quantum communication when combined with teleportation. From this, $\frac{1}{2}I(B:C)_{\mN(\theta)}$ is a universally achievable entanglement-assisted quantum communication rate.

Let us remark on the connection to the so-called father protocol~\cite{Devetak_2004_family, Devetak_2006_triangle_of_duality, Abeyesinghe_2009}. The father protocol refers to entanglement-assisted quantum communication, which achieves the following resource inequality:
\bal
\braket{\mN}+\frac{1}{2}I(A:E)_\phi [qq]\geq \frac{1}{2}I(A:B)_\phi[q\to q],
\eal
which also implies the Lloyd-Shor-Devetak quantum capacity theorem~\cite{Lloyd_1997_capacity, Devetak_2004_private_classical_capacity}.
However, the construction for entanglement-assisted classical/quantum communication in Theorem~\ref{thm: universal EACC fixed length} has an entanglement cost much greater than $\frac{1}{2}I(A:E)_\phi$.

\subsection{Variable-length communication with classical feedback}
In the following, we show that the mutual information 
\bal
I(\mN)\coloneqq\sup_{\rho_{AC}\in\mD(\mH_{AC})}I(B:C)_{\mN(\rho)}
\eal
is a universally achievable communication rate for entanglement-assisted classical communication with backward communication, even without knowledge of the channel capacity $I(\mN)$.
To this end, we first review some known facts about learning quantum channels and the continuity bound for mutual information.
\begin{lemboxed}
    \begin{lem}[Sample complexity of process tomography, {Ref.~\cite{bravoprieto_2026_quantum_memory}}]
        Let $\mN$ be a quantum channel with input dimension $d_{\rm in}$ and output dimension $d_{\rm out}$. Fix $\ve\in(0,1]$ and $\eta\in(0,1)$. Then, there exists a non-adaptive incoherent protocol requiring
        \bal
        N=\left\lceil C\frac{(d_{\rm in}d_{\rm out})^3+(d_{\rm in}d_{\rm out})^2\log \frac{1}{\eta}}{\ve^2}\right\rceil
        \eal
        queries to the channel $\mN$. The protocol succeeds with probability at least $1-\eta$ and achieves $\|\mN-\hat{\mN}\|_\diamond\leq \ve$. Here, $C$ is a universal constant.
    \end{lem}
\end{lemboxed}

\begin{lemboxed}
    \begin{lem}[{Uniform continuity bound of the optimized mutual information, Ref.~\cite[Proposition 10]{Shirokov_2017}}]\label{lem: cont capacity}
        Let $\mM,\mN:\mD(\mH_A)\to\mD(\mH_B)$ be two quantum channels with $\|\mM-\mN\|_\diamond\leq \ve$. Then, the channel mutual information satisfies the uniform continuity bound 
        \bal
        \abs{I(\mM)-I(\mN)}\leq 2\ve\log d+g(\ve),
        \eal
        where $g(\ve)\coloneqq(1+\ve)h_2\qty(\frac{\ve}{1+\ve})$, and $d=\min\qty{d_A,d_B}$.
    \end{lem}
\end{lemboxed}

The goal is to construct a universal protocol achieving $I(\mN)-\delta$ for a fixed $\delta$, where $\delta>0$ is chosen in advance. We begin the protocol with Alice and Bob sharing infinitely many copies of the maximally entangled state. An outline of the protocol is as follows:
\begin{enumerate}
    \item Given $n$ copies of the channel $\mN$, Alice sacrifices $m_n=o(n)$ copies of $\mN$ for process tomography.
    \item Bob performs tomography on his side and calculates the estimated channel $\hat{\mN}$, $I(\hat{\mN})$, and a pure state $\hat{\psi}$ achieving the channel capacity.
    \item Bob sends the estimated information to Alice, which requires a constant amount of backward classical communication.
    \item Alice and Bob run the entanglement dilution protocol using backward communication and obtain $\hat{\psi}$.
    \item Alice and Bob run the fixed-length universal entanglement-assisted classical communication protocol for the worst-case guaranteed capacity $R'_n$.
\end{enumerate}
Now, we take a closer look at the protocol. We take $\eta_n=\frac{1}{n}$, $\ve_n=\frac{1}{n^{1/4}}$. Then, process tomography succeeds with probability $1-\frac{1}{n}$, and the estimated channel $\hat{\mN}$ satisfies 
\bal
\|\mN-\hat{\mN}\|_\diamond\leq \ve=n^{-1/4},
\eal
where $\|\cdot\|_{\diamond}$ is the diamond norm.
Here, note that the success probability of the estimation is taken into account in the error probability of the channel coding itself. Furthermore, the query complexity for the channel $\mN$ is at most sublinear. Therefore, it suffices that the success probability of the estimation converges to $1$.
From the estimated channel description, we can find a pure resource state $\hat{\psi}$ satisfying $I(\hat{\mN})>I(B:C)_{\hat{\mN}\otimes{\rm id}(\hat{\psi})}>I(\hat{\mN})-\ve_n$.
Employing Lemma~\ref{lem: cont capacity}, we have
\bal
I(\mN)-2\ve_n\log d_{\rm in}-g(\ve_n)-\ve_n\leq I(B:C)_{\hat{\mN}\otimes{\rm id}(\hat{\psi})}\leq I(\mN)-2\ve_n\log d_{\rm in}-g(\ve_n)
\eal

From this, we take $R'_n\coloneqq I(B:C)_{\hat{\mN}\otimes{\rm id}(\hat{\psi})}-3\ve_n\log d_{\rm in}-g(\ve_n)-\ve_n$, and run the fixed-length protocol. In the asymptotic limit, we achieve $R_n'\to I(\mN)$.

\section{Universal entanglement-assisted quantum Gel'fand-Pinsker channel coding}
We review quantum Gel'fand-Pinsker channel coding considered in Refs.~\cite{Dupuis_2009_capacity_side_information, dupuis_2010_phd_thesis,Anshu_2019_building, Anshu_2019_on_the_near}. In the classical case, the Gel'fand-Pinsker channel has several possible states, with the state drawn according to a probability distribution; the sender also has access to another random variable correlated with the state of the channel. In the quantum case, the probability distribution shared by the sender and the channel is replaced by the quantum state $\ket{\tau}_{SS'}$. (See also FIG.~\ref{fig: def of Gel'fand Pinsker}.)

\begin{defnboxed}
    \begin{defn}[Entanglement-assisted quantum Gel'fand-Pinsker channel coding]
        Consider a state $\ket{\tau}^{SS'}$ shared by Alice and the channel, with Alice holding the system $S'$, and a quantum channel $\mN:\mD(\mH_A\otimes \mH_S)\to\mD(\mH_B)$. A triple $\qty(\theta^{EB'}, \qty{\mE^m_{ES'\to A}}_m,\qty{\Lambda^m_{BC}}_m)$ is called an $(M,\ve) $-code for the entanglement-assisted Gel'fand-Pinsker channel $(\mN,\tau)$ if the following conditions hold: $\theta_{EC}$ is a bipartite quantum state shared by Alice and Bob, where the system $E$ is held by Alice and the system $C$ is held by Bob.
        $\mE^m:\mD(\mH_E\otimes \mH_{S'})\to\mD(\mH_A)$ is the encoder for messages $m\in\mM=[M]$ and $\qty{\Lambda^m_{BC}}_{m\in\mM}$ is the decoder whose output $\hat{M}$ satisfies
        \bal
        1-\Tr\qty[\Lambda^m_{BC}\mN\circ\mE^m(\theta_{EC}\otimes \tau_{SS'})]\leq \ve
        \eal
        for any $m\in\mM$.
        $R>0$ is called an achievable rate for $(\mN,\tau)$ if, for any $\ve>0$, there exists a $(2^{nR},\ve)$-code for the entanglement-assisted Gel'fand-Pinsker channel $(\mN^{\otimes n},\tau^{\otimes n})$ for sufficiently large $n$.
        The capacity $C_{\rm GP}(\mN,\tau)$ of the entanglement-assisted Gel'fand-Pinsker channel is defined as the supremum of the achievable rates for $(\mN,\tau)$.
    \end{defn}
\end{defnboxed}

In the quantum case, with the shared state denoted by $\tau_{SS'}$, the capacity of the Gel'fand-Pinsker channel $\mN_{AS\to B}$ is known to be characterized as~\cite{Dupuis_2009_capacity_side_information, dupuis_2010_phd_thesis}
\bal
C_{\rm GP}(\mN,\tau)=\sup_{\substack{\theta_{ASC}\\ \theta_S=\tau_S}}\qty(I(B:C)_{\mN(\theta)}-I(S:C)_\theta)
\eal

\begin{figure}
    \centering
    \includegraphics[width=0.6\linewidth]{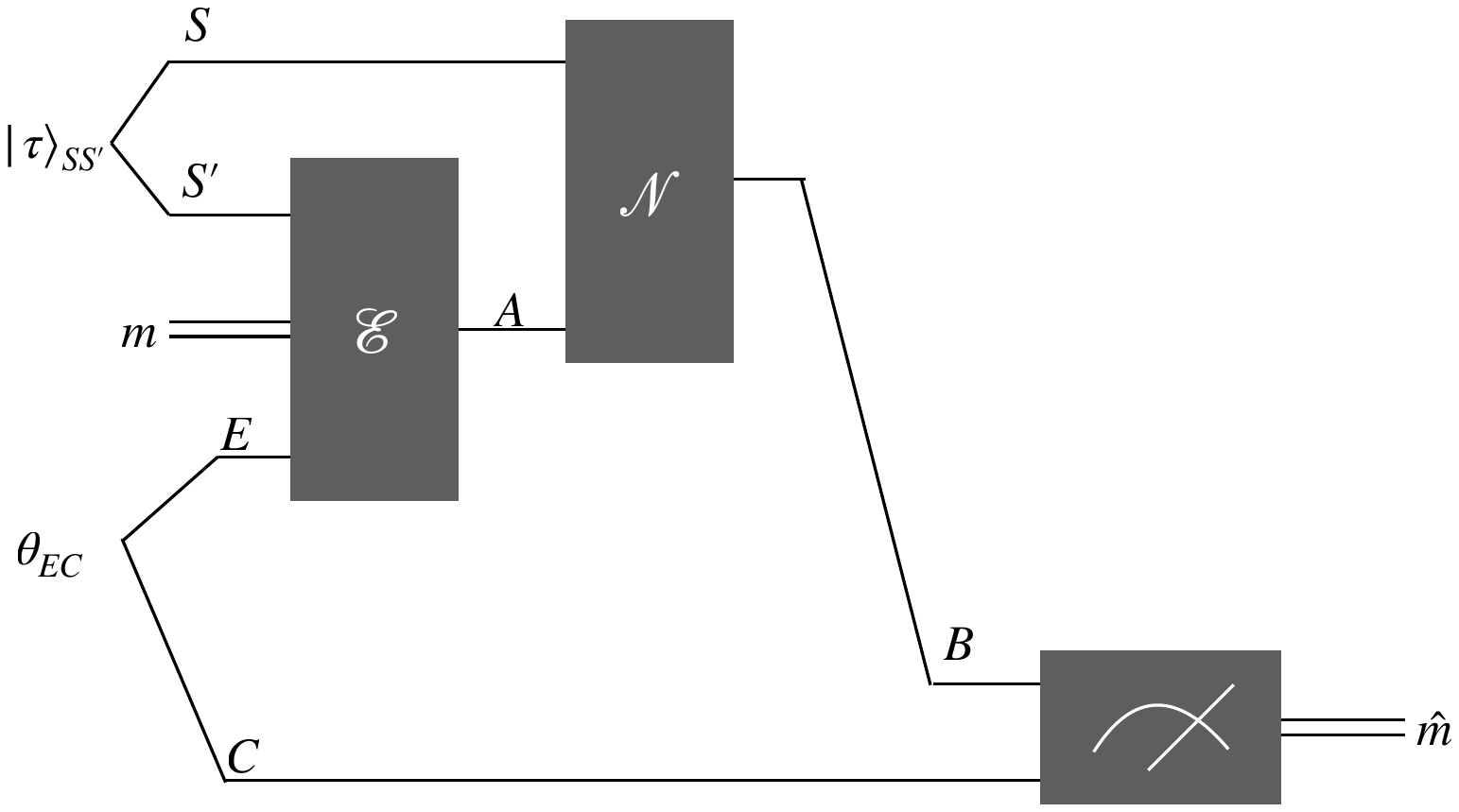}
    \caption{Schematic illustration of quantum Gel'fand-Pinsker channel coding. The state of the quantum channel $\mN$ varies depending on the register $S$, and the encoder is given access to another register $S'$, which is correlated with the channel state $S$.}
    \label{fig: def of Gel'fand Pinsker}
\end{figure}

Using semi-universal correlation detection for the classical-quantum channel and the convex splitting lemma, we obtain a universal coding scheme for the quantum Gel'fand-Pinsker channel.

\begin{thmboxed}
    \begin{thm}[Universal entanglement-assisted quantum Gel'fand-Pinsker channel coding]
        Fix an arbitrary state $\theta_{ASC}$ satisfying $\theta_S=\tau_S$. Then, there exists a sequence of codes for entanglement-assisted quantum Gel'fand-Pinsker channel coding that can depend on the description of $\theta_{ASC}$, achieving the rate $R_1-R_2$ in the asymptotic limit for any quantum channel $\mN$ satisfying $R_1<I(B:C)_{\mN^{AS\to B}\otimes {\rm id}^{C}(\theta_{ASC})}$ and $R_2>I(S:C)_\theta$.

        \begin{itemize}
    \item {\bf The sender and the receiver need to know}:
    \begin{enumerate}
        \item the quantum state $\ket{\tau}_{SS'}$ shared between the sender and the channel;
        \item the full-description of the shared state $\theta_{EC}$;
        \item the target rate $R_1, R_2$ satisfies $R_1<I(B:C)_{\mN^{AS\to B}\otimes {\rm id}^{C}(\theta_{ASC})}$ and $R_2>I(S:C)_\theta$.
    \end{enumerate}
    \item {\bf The sender and the receiver do not need to know}:
    \begin{enumerate}
        \item the full description of the channel $\mN$;
        \item the exact values of the mutual information $I(B:C)_{\mN^{AS\to B}\otimes {\rm id}^{C}(\theta_{ASC})}, I(S:C)_\theta$, and the capacity $C_{\rm GP}(\mN)$ of the entanglement-assisted Gel'fand-Pinsker channel coding.
    \end{enumerate}
\end{itemize}
    \end{thm}
\end{thmboxed}
We remark that encoding and decoding do not depend on the channel $\mN$, but depend on the quantum state $\ket{\tau}_{SS'}$ representing the channel state information. A fully universal protocol independent of the shared state $\ket{\tau}$ remains an open problem.
\begin{proof}
The proof idea follows Ref.~\cite{Anshu_2019_building}.   

Choose $\ket{\theta}_{EC}$ so that $\ket{\theta}_{EC}$ is a purification of $\theta_{C}$, and let Alice and Bob share $n2^{nR_1}=nM_1$ copies of $\ket{\theta}_{EC}$. As before, group the systems into blocks of $n$ and write $\tE=E^n$ and $\tC=C^n$; the shared state is then written as 
\bal
\theta_{\tE^{M_1}\tC^{M_1}}\coloneqq\theta_{\tE_1\tC_1}\otimes \cdots\otimes \theta_{\tE_{M_1}\tC_{M_1}}.
\eal
We divide $nM_1$ copies into $2^{n(R_1-R_2)}$ blocks, each of which is composed of $n2^{nR_2}=nM_2$ copies.
For $j=1,\ldots,2^{n(R_1-R_2)}$, the $j$-th block has the registers $\tE_{s_j},\ldots, \tE_{e_j}, \tC_{s_j},\ldots, \tC_{e_j}$. Here, we write $s_j\coloneqq(j-1)M_2+1, ~e_j=jM_2$. Furthermore, define $\mB(j)$ by $\mB(j)=\qty{s_j,s_j+1,\ldots, e_j}$.
Let us consider the following state
\bal
\tau^j_{\tS\tC_{s_j}\cdots \tC_{e_j}}\coloneqq\frac{1}{M_2}\sum_{k\in\mB(j)}\theta_{\tS\tC_{k}}\otimes \theta_{\tC_{s_j}}\otimes \cdots\theta_{\tC_{k-1}}\otimes \theta_{\tC_{k+1}}\otimes \cdots\otimes \theta_{\tC_{e_j}}.
\eal
Now, due to the unipartite convex splitting in Lemma~\ref{lem: unipartite convex splitting}, we have
\bal
\left\| \tau^j_{\tS\tC_{s_j}\cdots \tC_{e_j}}- \theta_{\tS}\otimes \theta_{\tC_{s_j}}\otimes \cdots\otimes \theta_{\tC_{e_j}} \right\|_1\leq 2\cdot 2^{-n\sup_{1\leq \alpha \leq 2}\frac{\alpha-1}{\alpha}(R_2-\sI^\downarrow_\alpha(C:S)_\theta)},
\eal
which decays exponentially in $n$ because we assume $R_2>I(S:C)_\theta$.
 
Let us take a purification $\ket{\Psi}_{FASC}$ of $\theta_{ASC}$.
Then, 
\bal
\ket{\tau^j}_{K\tA \tF\tS\tE_{s_j}\ldots\tE_{e_j}\tC_{s_j}\ldots\tC_{e_j}}=\frac{1}{\sqrt{M_2}}\sum_{k\in\mB(j)}\ket{k}_K\ket{\Psi}_{\tF\tA\tS\tC_k}\otimes \ket{\theta}_{\tE_{s_j}\tC_{s_j}}\otimes \cdots\otimes \ket{\theta}_{\tE_{k-1}\tC_{k-1}}\otimes \ket{0}_{\tE_k}\otimes \ket{\theta}_{\tE_{k+1}\tC_{k+1}}\otimes \cdots\otimes \ket{\theta}_{\tE_{e_j}\tC_{e_j}}
\eal
is a purification of $\tau^j_{\tS\tC_{s_j}\cdots \tC_{s_j}}$.

Since the state $\tau_{\tS\tS'}\otimes \theta_{\tE_{s_j}\tC_{s_j}}\otimes \cdots\otimes \theta_{\tE_{e_j}\tC_{e_j}}$ is a purification of $\theta_S\otimes \theta_{\tC_{s_j}}\otimes \theta_{\tC_{e_j}}$, due to Uhlmann's theorem, there exists an isometry $U_j:\tS'\tE_{s_j}\ldots\tE_{e_j}\to K\tA\tF\tE_{s_j}\ldots\tE_{e_j}$ for any $j=1,\ldots, 2^{n(R_1-R_2)}$ such that 
\bal
P\qty(\ketbra{\tau^j}{\tau^j},U_j\tau_{\tS\tS'}\otimes \theta_{\tE_{s_j}\tC_{s_j}}\otimes \cdots\otimes \theta_{\tE_{e_j}\tC_{e_j}}U_j^\dagger)\leq \sqrt{2\cdot 2^{-n\sup_{1\leq \alpha \leq 2}\frac{\alpha-1}{\alpha}(R_2-\sI^\downarrow_\alpha(C:S)_\theta)},}
\eal
also vanishes asymptotically.

The isometry discussed above gives us the encoder as follows.
Let $\mM_n\coloneqq\qty{1,\ldots, M_n}$ be a message set with $M_n=2^{n(R_1-R_2)}$. For a message $m\in\mM_n$, the sender applies the isometry $U_m$, which is indeed possible because the isometry is independent of the channel $\mN$.
Then, Alice sends the register $\tA\tS$ through the channel $\mN^{\otimes n }_{AS\to B}$.
Due to the discussion above and the monotonicity of the purified distance under CPTP maps, the purified distance between the resulting state $\Theta'^m_{\tB\tC_{1}\cdots\tC_{M_1}}$ after the application of $\mN^{\otimes n}$ and the ideal state 
\bal
\Theta^m_{\tB\tC_{1}\cdots\tC_{M_1}}\coloneqq\frac{1}{M_2}\sum_{k\in\mB(m)}\mN^{\otimes n}_{\tA\tS\to \tB}(\theta_{\tA\tS\tC_k})\otimes \theta_{\tC_1}\otimes \cdots\otimes \theta_{\tC_{k-1}}\otimes \theta_{\tC_{k+1}}\otimes \cdots\otimes  \theta_{\tC_{M_1}}
\eal
vanishes in the asymptotic limit $n\to\infty$ as
\bal
P\qty(\Theta'^m_{\tB\tC_{1}\cdots\tC_{M_1}}, \Theta^m_{\tB\tC_{1}\cdots\tC_{M_1}})\leq \sqrt{2\cdot 2^{-n\sup_{1\leq \alpha \leq 2}\frac{\alpha-1}{\alpha}(R_2-\sI^\downarrow_\alpha(C:S)_\theta)}gv }.
\eal

Next, we construct the decoder.
Note that information about $\theta_{ASC}$ is available to Bob, which allows him to construct the semi-universal correlation detector $P^n(a,\theta_{C})$, where we take the sequence $\qty{a_n}_n$ as $a_n\coloneqq R_1+1/\sqrt{n}$.
For any $l=1,\ldots, M_1$, we define an operator $\Gamma^l_{\tB\tC_{1}\cdots\tC_{M_1}}$ as 
\bal
\Gamma^l_{\tB\tC_{1}\cdots\tC_{M_1}}\coloneqq P^n_{\tB\tC_l}(a,\theta_{C})\otimes I_{\tC_1}\otimes \cdots\otimes  I_{\tC_{l-1}}\otimes I_{\tC_{l+1}}\otimes \cdots\otimes I_{\tC_{M_1}}.
\eal
From this, we can construct the decoder $\qty{\Lambda^l_{\tB\tC_{1}\cdots\tC_{M_1}}}_l$ as 
\bal
\widetilde{\Lambda}^l_{\tB\tC_{1}\cdots\tC_{M_1}}&\coloneqq\qty(\sum_{l'=1}^{M_1}\Gamma^{l'}_{\tB\tC_{1}\cdots\tC_{M_1}})^{-\frac{1}{2}} \Gamma^l_{\tB\tC_{1}\cdots\tC_{M_1}}\qty(\sum_{l'=1}^{M_1}\Gamma^{l'}_{\tB\tC_{1}\cdots\tC_{M_1}})^{-\frac{1}{2}},
\\
\Lambda^l_{\tB\tC_{1}\cdots\tC_{M_1}}&\coloneqq\widetilde{\Lambda}^l_{\tB\tC_{1}\cdots\tC_{M_1}}+\frac{1}{M_1}\qty(I-\sum_{l'=1}^{M_1}\Gamma^{l'}_{\tB\tC_{1}\cdots\tC_{M_1}})^0.
\eal
Upon obtaining the measurement outcome $l$, Bob decodes the message $m\in\mM_n$ such that $l\in\mB(m)$.

Finally, let us show that the decoding error vanishes in the asymptotic limit.
First, note that due to Lemma~\ref{lem: continuity of the prbability}, we have
\bal
{}&\Tr\qty[\qty(\sum_{l\not\in\mB(m)}\Lambda^l_{\tB\tC_{1}\cdots\tC_{M_1}} )\Theta'^m_{\tB\tC_{1}\cdots\tC_{M_1}}]\\
&\leq \qty(\sqrt{\Tr\qty[\qty(\sum_{l\not\in\mB(m)}\Lambda^l_{\tB\tC_{1}\cdots\tC_{M_1}} )\Theta^m_{\tB\tC_{1}\cdots\tC_{M_1}}]}+P(\Theta'^m_{\tB\tC_{1}\cdots\tC_{M_1}},\Theta^m_{\tB\tC_{1}\cdots\tC_{M_1}}))^2.
\eal
The first term is bounded as follows:
\begin{equation}\label{eq: gelfand bound 1}
    \begin{aligned}
        &\Tr\qty[\qty(\sum_{l\not\in\mB(m)}\Lambda^l_{\tB\tC_{1}\cdots\tC_{M_1}} )\Theta^m_{\tB\tC_{1}\cdots\tC_{M_1}}]\\
        &=\frac{1}{M_2}\sum_{k\in\mB(m)}\Tr[\qty(\sum_{l\not\in\mB(m)}\Lambda^l_{\tB\tC_{1}\cdots\tC_{M_1}} )\mN^{\otimes n}_{\tA\tS\to \tB}(\theta_{\tA\tS\tC_k})\otimes \theta_{\tC_1}\otimes \cdots\otimes \theta_{\tC_{k-1}}\otimes \theta_{\tC_{k+1}}\otimes \cdots\otimes  \theta_{\tC_{M_1}} ]\\
        &\leq \frac{1}{M_2}\sum_{k\in\mB(m)}\Tr[\qty(\sum_{l|l\neq k}\Lambda^l_{\tB\tC_{1}\cdots\tC_{M_1}} )\mN^{\otimes n}_{\tA\tS\to \tB}(\theta_{\tA\tS\tC_k})\otimes \theta_{\tC_1}\otimes \cdots\otimes \theta_{\tC_{k-1}}\otimes \theta_{\tC_{k+1}}\otimes \cdots\otimes  \theta_{\tC_{M_1}} ]\\
        &=\Tr\qty[\qty(\sum_{l|l\neq s_m}\Lambda^l_{\tB\tC_{1}\cdots\tC_{M_1}} )\mN^{\otimes n}_{\tA\tS\to \tB}(\theta_{\tA\tS\tC_{s_m}})\otimes \theta_{\tC_1}\otimes \cdots\otimes \theta_{\tC_{s_m-1}}\otimes \theta_{\tC_{s_m+1}}\otimes \cdots\otimes  \theta_{\tC_{M_1}}],
    \end{aligned}
\end{equation}
where the last line is because of the permutation symmetry and $s_m$ denotes an arbitrary representative position.
Now, using the Hayashi-Nagaoka inequality in Lemma~\ref{lem: Hayashi Nagaoka inequality}, we have
\begin{equation}\label{eq: gelfand bound 2}
    \begin{aligned}
        &\Tr\qty[\qty(\sum_{l|l\neq s_m}\Lambda^l_{\tB\tC_{1}\cdots\tC_{M_1}} )\mN^{\otimes n}_{\tA\tS\to \tB}(\theta_{\tA\tS\tC_{s_m}})\otimes \theta_{\tC_1}\otimes \cdots\otimes \theta_{\tC_{s_m-1}}\otimes \theta_{\tC_{s_m+1}}\otimes \cdots\otimes  \theta_{\tC_{M_1}}]\\
        &=\Tr\qty[\qty(I-\Lambda^{s_m}_{\tB\tC_{1}\cdots\tC_{M_1}})\mN^{\otimes n}_{\tA\tS\to \tB}(\theta_{\tA\tS\tC_{s_m}})\otimes \theta_{\tC_1}\otimes \cdots\otimes \theta_{\tC_{s_m-1}}\otimes \theta_{\tC_{s_m+1}}\otimes \cdots\otimes  \theta_{\tC_{M_1}}]\\
        &\leq \Tr\qty[\qty(I-\widetilde{\Lambda}^{s_m}_{\tB\tC_{1}\cdots\tC_{M_1}})\mN^{\otimes n}_{\tA\tS\to \tB}(\theta_{\tA\tS\tC_{s_m}})\otimes \theta_{\tC_1}\otimes \cdots\otimes \theta_{\tC_{s_m-1}}\otimes \theta_{\tC_{s_m+1}}\otimes \cdots\otimes  \theta_{\tC_{M_1}}]\\
        &\leq \Tr\qty[\qty(c_1\qty(I-\Gamma^{s_m}_{\tB\tC_{1}\cdots\tC_{M_1}})+c_2\sum_{l~|~l\neq s_m}\Gamma^{l}_{\tB\tC_{1}\cdots\tC_{M_1}})\mN^{\otimes n}_{\tA\tS\to \tB}(\theta_{\tA\tS\tC_{s_m}})\otimes \theta_{\tC_1}\otimes \cdots\otimes \theta_{\tC_{s_m-1}}\otimes \theta_{\tC_{s_m+1}}\otimes \cdots\otimes  \theta_{\tC_{M_1}}]\\
        &=c_1\Tr\qty[\qty(I-P^n_{\tB\tC_{s_m}}(a_n,\theta_{C}))\mN^{\otimes n}_{\tA\tS\to \tB}(\theta_{\tA\tS\tC_{s_m}})]+c_2 (2^{nR_1}-1)\Tr\qty[P^n_{\tB\tC_{s_m}}(a_n,\theta_{C})\mN^{\otimes n}_{\tA\tS\to \tB}(\theta_{\tA\tS})\otimes \theta_{\tC}].
    \end{aligned}
\end{equation}
Here, in the third line, we used $\widetilde{\Lambda}^{s_m}_{\tB\tC_{1}\cdots\tC_{M_1}}\leq\Lambda^{s_m}_{\tB\tC_{1}\cdots\tC_{M_1}}$.
Finally, due to Proposition~\ref{pro: semi universal detector}, we have
\begin{equation}
    \begin{aligned}
        &c_1\Tr\qty[\qty(I-P^n_{\tB\tB_{s_m}}(a_n,\theta_{C}))\mN^{\otimes n}_{\tA\tS\to \tB}(\theta_{\tA\tS\tC_{s_m}})]+c_2 (2^{nR_1}-1)\Tr\qty[P^n_{\tB\tB_{s_m}}(a_n,\theta_C)\mN^{\otimes n}_{\tA\tS\to \tB}(\theta_{\tA\tS})\otimes \theta_{\tC}]\\
        &\leq c_1\frac{1}{{\rm poly}(n)}2^{-nt\qty(\pI^{\downarrow}_{1-t}(C:B)_{\mN^{AS\to B}\otimes {\rm id}^{C}(\theta_{ASC})}-a_n)}+c_2{\rm poly}(n)2^{-n(a_n-R_1)}, ~\forall t\in(0,1).
    \end{aligned}
\end{equation}
Due to the choice of $\{a_n\}_n$, the error goes to $0$ in the asymptotic limit, which concludes the proof.

\end{proof}

\section{Universal entanglement-assisted quantum Marton's inner bound}\label{sec: universal marton}
We now review the concept of a broadcast channel.
A broadcast channel refers to a channel with a single sender and multiple receivers, a setting traditionally studied in both classical and quantum network information theory~\cite{Cover_broadcast_channel, Bergman_1973_random_coding, Gallager__1974_Capacity, Korner_general_broadcast}.
In the scenario with $L$ receivers, the sender encodes $(m_1,\ldots, m_L)\in\mM_1\times \cdots\times \mM_L$, and the $l$-th receiver needs to decode the $l$-th message $m_l$. (See also FIG.~\ref{fig:placeholder} for the two-receiver scenario.)

In the following, we mainly consider entanglement-assisted classical communication through a quantum broadcast channel with $L$ receivers.

\begin{defnboxed}
    \begin{defn}[Entanglement-assisted classical communication through a quantum broadcast channel with $L$ receivers]
    Let $\mN_{A\to B_1B_2\cdots B_L}$ be a quantum broadcast channel. A tuple $\qty(\theta_{E_1C_1},\ldots, \theta_{E_LC_L},\qty{\mE^{(m_1,\ldots, m_L)}}_{(m_1,\ldots, m_L)\in\mM_1\times\cdots\times \mM_L},\qty{\Lambda_{B_1C_1}^{m_1}}_{m_1\in\mM_1},\ldots,  \qty{\Lambda_{B_LC_L}^{m_L}}_{m_L\in\mM_L})$ is called an $(M_1,\ldots,M_L,\ve)$-code for entanglement-assisted classical communication through a quantum broadcast channel $\mN$ if it satisfies the following:
    For any $l\in[L]$, $\theta_{E_lC_l}$ is the shared quantum state between the sender and the $l$-th receiver, where the sender holds $E_l$, and the $l$-th receiver holds $C_l$. 
    $\mE^{(m_1,\ldots, m_L)}:\mD(\mH_{E_1}\otimes \cdots\otimes \mH_{E_L})\to \mD(\mH_{A}^{\otimes n})$ is the encoder which encodes the tuple of messages $(m_1,\ldots,m_L)\in\mM_1\times \cdots\mM_{L}$. 
    The POVM $\qty{\Lambda_{B_lC_l}^{m_l}}_{m_l\in\mM_l}$ is the decoder for the $l$-th receiver. 
    The decoding error satisfies
    \bal
    1-\Tr\qty[(\Lambda^{m_1}_{B_1C_1}\otimes\cdots\otimes \Lambda^{m_L}_{B_LC_L})\mN\circ\mE^{(m_1,\ldots, m_L)}(\theta_{E_1C_1}\otimes \cdots\otimes \theta_{E_LC_L})]\leq \ve.
    \eal
    A tuple $(R_1,\ldots, R_L)$ is called an achievable rate tuple if, for any $\ve$ and sufficiently large $n$, there exists a $(2^{nR_1}, \ldots, 2^{nR_L},\ve)$-code for the broadcast channel $\mN^{\otimes n}$. The capacity region is defined as the union of the achievable rates.
    \end{defn}
\end{defnboxed}
\begin{figure}
    \centering
    \includegraphics[width=0.6\linewidth]{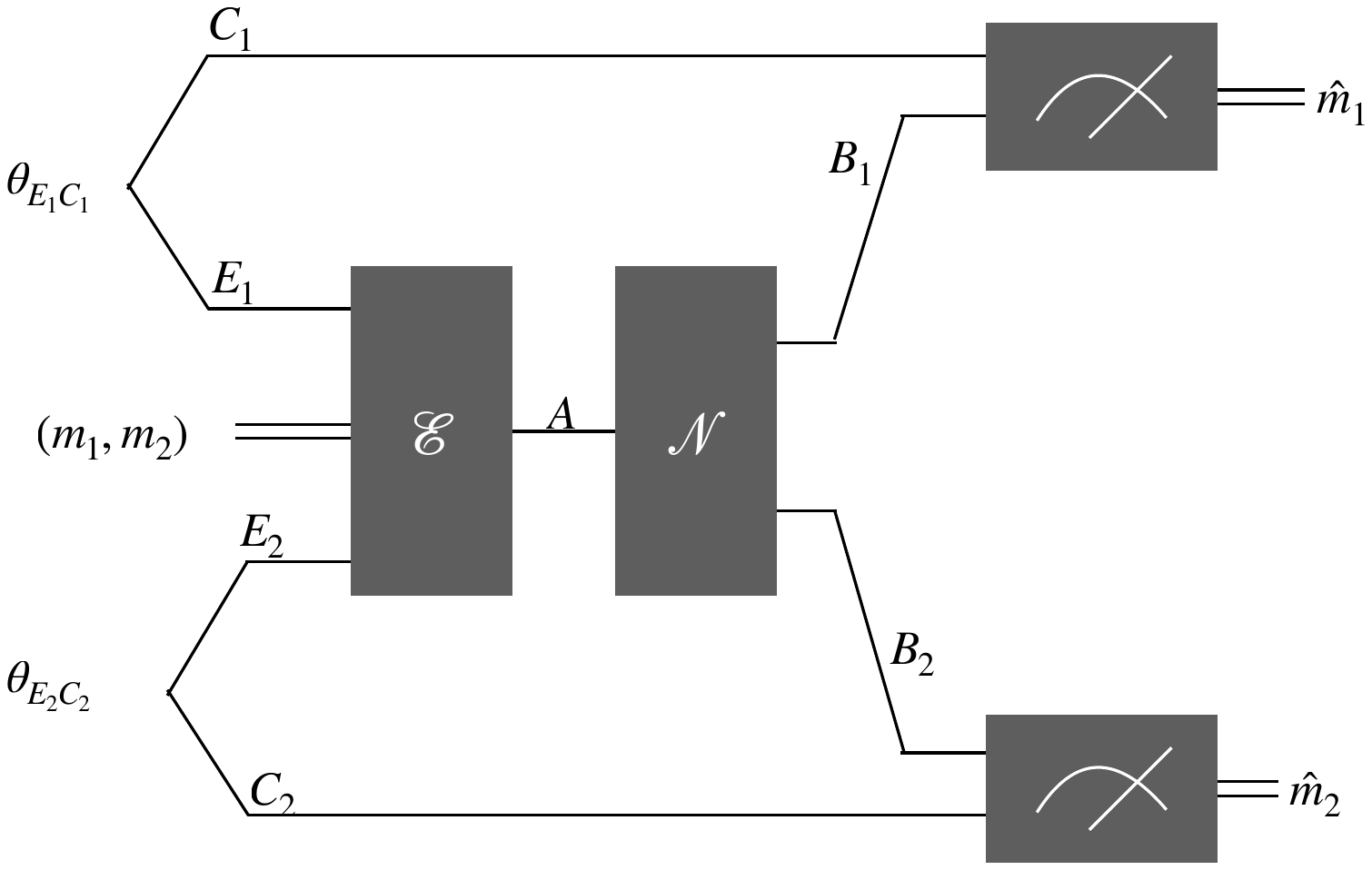}
    \caption{Schematic illustration of an entanglement-assisted quantum broadcast channel with two receivers.}
    \label{fig:placeholder}
\end{figure}

Even in the classical case, a single-letter characterization of the capacity region of a general two-receiver broadcast channel remains unknown. Previous studies have therefore developed achievable inner bounds and characterized the capacity region in special settings, including degraded broadcast channels~\cite{Bergman_1973_random_coding, Gallager__1974_Capacity} and broadcast channels with degraded message sets~\cite{Korner_general_broadcast}. In the latter setting, the sender transmits a common message to both receivers and an additional private message to one of them, without requiring the channel itself to be degraded. Superposition coding has also been extended to classical-quantum broadcast channels~\cite{Yard_2011_quantum_broadcast, Savov_2015_classical_code}. Universal coding has been studied for degraded message sets and related settings involving confidential messages, in both the classical and quantum cases~\cite{Hayashi_2011_universally_attainable, Boche_2020_universal_superposition, Hayashi_2022_universal_cq_superposition}.

For a general two-receiver classical broadcast channel with independent private messages, Marton's inner bound provides a single-letter achievable rate region~\cite{Marton_coding_theorem, ElGamal_proof_of_marton}. Although its single-letter form has recently been shown to be suboptimal in general~\cite{Huang_2026_suboptimality_martons}, its regularization, obtained by taking the closure of the union of the normalized inner bounds over all block lengths, exactly characterizes the capacity region~\cite{Geng_2014_marton_inner}. One-shot extensions incorporating a common message have also been developed~\cite{Liu_2015_one_shot_marton}. In the classical-quantum setting, a Marton-type inner bound for two receivers was established in Ref.~\cite{Savov_2015_classical_code}.

Entanglement-assisted analogs of Marton's inner bound have also been established. The quantum broadcast father protocol yields an achievable region for entanglement-assisted quantum communication, including an extension to arbitrarily many receivers; superdense coding converts this region into one for entanglement-assisted classical communication~\cite{Dupuis_2010_father_broadcast}. A one-shot Marton-type inner bound for entanglement-assisted classical communication with two receivers was subsequently established in Ref.~\cite{Anshu_2019_building}. To the best of our knowledge, however, universal achievability of these entanglement-assisted Marton-type bounds has not previously been established.

Here, we show that the entanglement-assisted Marton-type inner bound is universally achievable for an arbitrary number of receivers, including $L\geq 3$. Our construction combines the semi-universal correlation detector in Proposition~\ref{pro: semi universal detector} with the multipartite convex splitting lemma developed in Ref.~\cite{Cheng_2023_quantum_broadcast}.

\begin{thmboxed}
    \begin{thm}[Universal entanglement-assisted Marton's inner bound]\label{thm: universal L sender Marton's bound}
        Fix a quantum state $\theta_{AC_1\cdots C_L}$. Then, under the assumption that the sender and the receiver are informed of the values of the mutual information $\qty{I(B_l: C_l)_{\mN(\theta)}}_{l\in[L]}$ and the total correlation $\qty{I[S]_{\theta}}_{S\subset [L]}$, there exists a sequence of codes for entanglement-assisted classical communication through a quantum broadcast channel achieving the rate tuple $(R_1,\ldots, R_L)$ for any broadcast channel $\mN$ with $L$ receivers satisfying 
        \begin{equation}\label{eq: marton each R}
            0<R_l<I(B_l:C_l)_{\mN(\theta)},\qquad\forall l\in[L],
        \end{equation}
        \begin{equation}\label{eq: marton sum of R}
            \sum_{l\in S}R_l<\sum_{l\in S}I(B_l:C_l)_{\mN(\theta)}-I[S]_\theta,\qquad \emptyset\neq\forall S\subset[L].
        \end{equation}
        Here, $I[S]_\theta$ is the total correlation of $S\subset [L]$, defined in Eq.~\eqref{eq: def of total correlation}.
        \begin{itemize}
    \item {\bf The sender and the receiver need to know}:
    \begin{enumerate}
        \item The exact values of the mutual information $\qty{I(B_l: C_l)_{\mN(\theta)}}_{l\in[L]}$ $\qty{I(B_l: C_l)_{\mN(\theta)}}_{l\in[L]}$and the total correlation $\qty{I[S]_{\theta}}_{S\subset [L]}$;
        \item the full description of the shared states $\qty{\theta_{E_lC_l}}_{l\in[L]}$ between the sender and the receivers;
    \end{enumerate}
    \item {\bf The sender and the receiver do not need to know}:
    \begin{enumerate}
        \item the full description of the channel $\mN$
    \end{enumerate}
\end{itemize}
    \end{thm}
\end{thmboxed}

Before proceeding to the details, we first show the following lemma.
\begin{lemboxed}
    \begin{lem}\label{lem: existence of r_l}
        For any rate tuple $(R_1,\ldots, R_L)$ satisfying Eq.~\eqref{eq: marton each R} and Eq.~\eqref{eq: marton sum of R}, there exists a tuple of positive numbers $(r_1,\ldots,r_L)$ such that 
        \begin{equation}\label{eq: lem R_l+r_l}
            R_l+r_l<I(B_l:C_l)_{\mN(\theta)},\qquad\forall l\in[L],\\
        \end{equation}
        \begin{equation}\label{eq: lem sum of r_l}
            \sum_{l\in S} r_l>I[S]_\theta, \qquad\emptyset\neq \forall S\subset [L].
        \end{equation}
    \end{lem}
\end{lemboxed}
\begin{proof}[Proof of Lemma~\ref{lem: existence of r_l}]
    From Eq.~\eqref{eq: marton each R} and Eq.~\eqref{eq: marton sum of R}, there exist positive numbers $\ve_S, \ve_l$ indexed by $S\subset[L]$ and $l\in [L]$ such that 
    \bal
    0<R_l&<I(B_l:C_l)_{\mN(\theta)}-\ve_l,\qquad\forall l\in[L],\\
    \sum_{l\in S}R_l&<\sum_{l\in S}I(B_l:C_l)_{\mN(\theta)}-I[S]_\theta-\ve_S,\qquad\emptyset\neq \forall S\subset[L].
    \eal
    Now, let us take
    \bal
    \ve'\coloneqq\frac{1}{2}\min\qty{\min_{S\subset[L]}\frac{\ve_S}{L}, \min_{l\in[L]}\ve_l }>0.
    \eal
    Employing $\ve'$, we take $r_l$ so that 
    \bal
    I(B_l:C_l)_{\mN(\theta)}-R_l-\ve'<r_l<I(B_l:C_l)_{\mN(\theta)}-R_l.
    \eal
    Here, note that Eq.~\eqref{eq: lem R_l+r_l} holds from the definition, and $r_l>0$ because $\ve'\leq \ve_l$ holds from the definition of $\ve'$.
    Now, it suffices to check that Eq.~\eqref{eq: lem sum of r_l} holds. Let us take an arbitrary nonempty set $S\subset[L]$. It follows that
    \bal
    \sum_{l\in S}r_l&>\sum_{l\in S}I(B_l:C_l)_{\mN(\theta)}-\sum_{l\in S}R_l-\abs{S}\ve'\\
    &>\sum_{l\in S}I(B_l:C_l)_{\mN(\theta)}-\sum_{l\in S}I(B_l:C_l)_{\mN(\theta)}+I[S]_\theta+\ve_S-\abs{S}\ve'\\
    &\geq I[S]_\theta+\ve_S-\frac{\abs{S}}{2L}\ve_S>I[S]_\theta,
    \eal
    which concludes the proof.
\end{proof}
To construct the universal coding scheme, one must specify both the communication rates $\qty{R_l}_{l\in[L]}$ and the auxiliary rates $\qty{r_l}_{l\in[L]}$.
The prior knowledge of the mutual information $\qty{I(B_l: C_l)_{\mN(\theta)}}_{l\in[L]}$ and the total correlations $\qty{I[S]_\theta}_{S\subset[L]}$ enables us to find a suitable choice of the auxiliary rates $\qty{r_l}_{l\in[L]}$ through the proof in Lemma~\ref{lem: existence of r_l}.

We prove the main statement.
\begin{proof}[Proof of Theorem~\ref{thm: universal L sender Marton's bound}]
The proof follows the approach of Ref.~\cite{Anshu_2019_building}.
For any $l\in [L]$, let $\ket{\theta^l}_{E_lC_l}$ be a purification of the marginal $\theta_{C_l}=\Tr_{\backslash C_l}\theta_{AC_1\cdots C_L}$ of the fixed state $\theta_{AC_1\cdots C_L}$. Then, before the protocol, we let Alice, the sender, and the $l$-th receiver share $n2^{n(R_l+r_l)}=:nN_l$ copies of $\ket{\theta^l}_{E_lC_l}$.
Now, let us group the registers into blocks of $n$ and write $\tE\coloneqq E^n, \tC_l\coloneqq C_l^n$. The shared state is 
\bal
\theta_{\tE_1^{N_1}\cdots\tE_L^{N_L}\tC_1^{N_1}\cdots\tC_L^{N_L}}\coloneqq(\theta^{1}_{\tE_1\tC_1})^{\otimes N_1}\otimes \cdots\otimes (\theta^{L}_{\tE_1\tC_1})^{\otimes N_L}.
\eal
For each $l\in [L]$, we divide these states into the bins $\mB_l(1),\ldots,\mB_l(2^{nR_l})$, each of which is composed of $2^{nr_l}$ registers  $\tE, \tC$ with labels $s^l_{j_l}\coloneqq(j_l-1)2^{r_l},(j_l-1)2^{r_l}+1,\ldots,e^l_{j_l}\coloneqq j_l2^{r_l}$.

For a tuple $(j_1,\ldots, j_L)\in[2^{nR_1}]\times \cdots\times [2^{nR_L}]$, we define a quantum state
\begin{equation}
\begin{aligned}
&\tau^{(j_1,\ldots, j_L)}_{\tC_1^{\mB(j_1)}\cdots\tC_L^{\mB(j_L)}}\\
&\coloneqq\frac{1}{2^{n\sum_{l=1}^L r_l}}
\sum_{(k_1,\ldots,k_L)\in \mB(j_1)\times \cdots\times\mB(j_L)}\theta_{\tC_{1,k_1}\cdots\tC_{L,k_L}}\otimes \bigotimes_{l=1}^L\theta_{\tC_{l,s^l_{j_l}}}\otimes \cdots\otimes  \theta_{\tC_{l,k_l-1}}\otimes\theta_{\tC_{l,k_l+1}}\otimes \cdots\otimes \theta_{\tC_{l,e^l_{j_l}}}.
\end{aligned}
\end{equation}
Then, if we take a purification $\ket{\Psi}_{FAC_1\cdots C_L}$ of $\theta_{AC_1\cdots C_L}$, the quantum state defined as 
\begin{equation}
\begin{aligned}
\ket{\tau^{(j_1,\ldots, j_L)}}_{K\tF\tA\tE_1^{\mB(j_1)}\cdots\tE_L^{\mB(j_L)}\tC_1^{\mB(j_1)}\cdots\tC_L^{\mB(j_L)}}&\coloneqq\frac{1}{\sqrt{2^{n\sum_{l=1}^L r_l}}}\sum_{(k_1,\ldots,k_L)\in \mB(j_1)\times \cdots\times\mB(j_L)}\ket{(k_1,\ldots,k_L)}_K\ket{\Psi}_{\tF\tA\tC_1\cdots \tC_L}\\
&\otimes \bigotimes_{l=1}^L \ket{\theta}_{\tE_l\tC_l,s^l_{j_l}}\otimes \cdots\otimes \ket{\theta}_{\tE_l\tC_l,k-1}\otimes \ket{0}_{\tE_l,k}\otimes \ket{\theta}_{\tE_l\tC_l,k+1}\otimes \cdots\otimes \ket{\theta}_{\tE_l\tC_l,e^l_{j_l}}
\end{aligned}
\end{equation}
is a purification of $\tau^{(j_1,\ldots, j_L)}_{\tC_1^{\mB(j_1)}\cdots\tC_L^{\mB(j_L)}}$.

Now, noting that Eq.~\eqref{eq: lem sum of r_l} holds, due to the multipartite convex splitting lemma in Lemma~\ref{lem: multipartite convex splitting},
\bal
\left\| \tau^{(j_1,\ldots, j_L)}_{\tC_1^{\mB(j_1)}\cdots\tC_L^{\mB(j_L)}}-\bigotimes_{l=1}^L \theta_{\tC_{l, s^l_{j_l}}}\otimes \cdots\otimes \theta_{\tC_{l, e^l_{j_l}}} \right\|_1\leq \sum_{\emptyset\neq S\subset[L]}2^{\abs{S}}2^{-n\sup_{1\leq \alpha\leq 2}\frac{\alpha-1}{\alpha}(\sum_{l\in S} r_l-\sI_\alpha^\downarrow[S])}
\eal
holds. Due to Lemma~\ref{lem: existence of r_l}, the term $2^{-n\sup_{1\leq \alpha\leq 2}\frac{\alpha-1}{\alpha}(\sum_{l\in S} r_l-\sI_\alpha^\downarrow[S])}$ decays exponentially in $n$.
Moreover, due to the Fuchs--van de Graaf inequality, the purified distance is bounded from above as 
\bal
 P\qty(\tau^{(j_1,\ldots, j_L)}_{\tC_1^{\mB(j_1)}\cdots\tC_L^{\mB(j_L)}},\bigotimes_{l=1}^L \theta_{\tC_{l, s^l_{j_l}}}\otimes \cdots\otimes \theta_{\tC_{l, e^l_{j_l}}} )\leq \sqrt{\sum_{\emptyset\neq S\subset[L]}2^{\abs{S}}2^{-n\sup_{1\leq \alpha\leq 2}\frac{\alpha-1}{\alpha}(\sum_{l\in S} r_l-\sI_\alpha^\downarrow[S])}},
\eal
whose decay is also exponential.
Since the shared state between the sender and the receivers can be seen as a purification of $\bigotimes_{l=1}^L \theta_{\tC_{l, s^l_{j_l}}}\otimes \cdots\otimes \theta_{\tC_{l, e^l_{j_l}}}$, due to Uhlmann's theorem, there exists an isometry
\bal
U^{(j_1,\ldots, j_L)}:\tE_1^{\mB(j_1)}\cdots\tE_L^{\mB(j_L)} \to K\tF\tA\tE_1^{\mB(j_1)}\cdots\tE_L^{\mB(j_L)}
\eal
satisfying
\bal
P\qty(\ketbra{\tau^{(j_1,\ldots, j_L)}}{\tau^{(j_1,\ldots, j_L)}}, U^{(j_1,\ldots, j_L)}\theta_{\tE_{1}^{\mB(j_1)}\cdots\tE_{L }^{\mB(j_L)}\tC_{1}^{\mB(j_1)}\cdots\tC_{L}^{\mB(j_L)}}\qty(U^{(j_1,\ldots, j_L)})^\dagger)= P\qty(\tau^{(j_1,\ldots, j_L)}_{\tC_1^{\mB(j_1)}\cdots\tC_L^{\mB(j_L)}},\bigotimes_{l=1}^L \theta_{\tC_{l, s^l_{j_l}}}\otimes \cdots\otimes \theta_{\tC_{l, e^l_{j_l}}} ).
\eal
This isometry serves as the encoder.

Now, we explain how the universal protocol works.
To send a tuple of messages $(m_1,\ldots, m_L)$, the sender applies the isometry $U^{(m_1,\ldots, m_L)}$ to the systems $\tE_1^{\mB(m_1)}\cdots\tE_L^{\mB(m_L)}$. Then, the sender feeds the register $\tA$ to the broadcast channel. Let $\Theta'^{(m_1,\ldots, m_L)}_{\tB_1\cdots\tB_L\tC_1^{N_1}\cdots\tC_L^{N_L}}$ denote the global output state of the broadcast channel.
Now, due to the monotonicity of the purified distance under CPTP maps, the purified distance between $\Theta'^{(m_1,\ldots, m_L)}_{\tB_1\cdots\tB_L\tC_1^{N_1}\cdots\tC_L^{N_L}}$ and the ideal state
\begin{equation}
    \begin{aligned}
        &\Theta^{(m_1,\ldots, m_L)}_{\tB_1\cdots\tB_L\tC_1^{N_1}\cdots\tC_L^{N_L}}\\
        &\coloneqq\frac{1}{2^{n\sum_{l=1}^L r_l}}
\sum_{(k_1,\ldots,k_L)\in \mB(m_1)\times \cdots\times\mB(m_L)}\mN^{\otimes n}_{\tA\to\tB_1\cdots\tB_L}(\theta_{\tA\tC_{1,k_1}\cdots \tC_{L,k_L}})\otimes \bigotimes_{l=1}^L\theta_{\tC_{l,s^l_{m_l}}}\otimes \cdots\otimes  \theta_{\tC_{l,k_l-1}}\otimes\theta_{\tC_{l,k_l+1}}\otimes \cdots\otimes \theta_{\tC_{l,e^l_{m_l}}}
    \end{aligned}
\end{equation}
is bounded as 
\bal
P\qty(\Theta^{(m_1,\ldots, m_L)}_{\tB_1\cdots\tB_L\tC_1^{N_1}\cdots\tC_L^{N_L}}, \Theta'^{(m_1,\ldots, m_L)}_{\tB_1\cdots\tB_L\tC_1^{N_1}\cdots\tC_L^{N_L}})\leq \sqrt{\sum_{\emptyset\neq S\subset[L]}2^{\abs{S}}2^{-n\sup_{1\leq \alpha\leq 2}\frac{\alpha-1}{\alpha}(\sum_{l\in S} r_l-\sI_\alpha^\downarrow[S])}}
\eal
which decays exponentially.

Now, we construct the decoder. Since we fix $\theta_{AC_1\cdots C_L}$, the receivers have information about their respective marginals $\theta_{C_l}$ of the shared state. Now, for each $l\in[L]$, we fix
\bal
a_l\coloneqq\frac{R_l+r_l+I(B_l:C_l)_{\mN(\theta)}}{2},
\eal
and take the semi-universal correlation detector $P^n_{\tB_l\tC_l}(a_l, \theta_{C_l})$. For any $k'\in[N_l]$, we define an operator 
\bal
\Gamma^{k'}_{\tB_l\tC_{l,1}\cdots\tC_{l,N_l}}\coloneqq P^n_{\tB\tC_{l,k'}}(a_l,\theta_{C_l})\otimes I_{\tC_{l,1}}\otimes \cdots\otimes  I_{\tC_{l,k'-1}}\otimes I_{\tC_{l,k'+1}}\otimes \cdots\otimes I_{\tC_{l,N_l}},
\eal
and a pretty-good measurement $\qty{\Lambda^{k'}_{\tB_l\tC_{l,1}\cdots\tC_{l,N_l}}}_{k'}$
\bal
\widetilde{\Lambda}^{k'}_{\tB_l\tC_{l,1}\cdots\tC_{l,N_l}}&\coloneqq\qty(\sum_{k''=1}^{N_l}\Gamma^{k''}_{\tB_l\tC_{l,1}\cdots\tC_{l,N_l}})^{-\frac{1}{2}} \Gamma^{k'}_{\tB_l\tC_{l,1}\cdots\tC_{l,N_l}}\qty(\sum_{k''=1}^{N_l}\Gamma^{k''}_{\tB_l\tC_{l,1}\cdots\tC_{l,N_l}})^{-\frac{1}{2}},\\
\Lambda^{k'}_{\tB_l\tC_{l,1}\cdots\tC_{l,N_l}}&\coloneqq \widetilde{\Lambda}^{k'}_{\tB_l\tC_{l,1}\cdots\tC_{l,N_l}}+\frac{1}{N_l}\qty(I-\sum_{k''=1}^{N_l}\Gamma^{k''}_{\tB_l\tC_{l,1}\cdots\tC_{l,N_l}})^0.
\eal
The receiver decodes the message $\hat{m}_l$ such that the measurement outcome $k'$ is in the block $\mB(\hat{m}_l)$.

Now, we show that the decoding error vanishes in the asymptotic limit.
First, note that 
\bal
{\rm Pr}[(\hat{M}_1,\ldots,\hat{M}_L)\neq (m_1,\ldots, m_L)~|~({M}_1,\ldots,{M}_L)=(m_1,\ldots, m_L)]\leq \sum_{l\in[L]}{\rm Pr}[(\hat{M}_l\neq m_l)~|~({M}_1,\ldots,{M}_L)=(m_1,\ldots, m_L)]
\eal
holds.
Now, we upper bound each term.
For any $l$, we have
\begin{equation}
    \begin{aligned}
        &{\rm Pr}[(\hat{M}_l\neq m_l)~|~({M}_1,\ldots,{M}_L)=(m_1,\ldots, m_L)]\\
        &=\Tr\qty[\qty(\sum_{k'\not \in\mB(m_l)}\Lambda^{k'}_{\tB_l\tC_{l,1}\cdots\tC_{l,N_l}})\Theta'^{(m_1,\ldots, m_L)}_{\tB_1\cdots\tB_L\tC_1^{N_1}\cdots\tC_L^{N_L}}]\\
        &\leq \qty(\qty(\Tr\qty[\qty(\sum_{k'\not \in\mB(m_l)}\Lambda^{k'}_{\tB_l\tC_{l,1}\cdots\tC_{l,N_l}})\Theta^{(m_1,\ldots, m_L)}_{\tB_1\cdots\tB_L\tC_1^{N_1}\cdots\tC_L^{N_L}}])^2+P\qty(\Theta^{(m_1,\ldots, m_L)}_{\tB_1\cdots\tB_L\tC_1^{N_1}\cdots\tC_L^{N_L}},\Theta'^{(m_1,\ldots, m_L)}_{\tB_1\cdots\tB_L\tC_1^{N_1}\cdots\tC_L^{N_L}}))^2
    \end{aligned}
\end{equation}
It suffices to upper bound $\Tr\qty[\qty(\sum_{k'\not \in\mB(m_l)}\Lambda^{k'}_{\tB_l\tC_{l,1}\cdots\tC_{l,N_l}})\Theta^{(m_1,\ldots, m_L)}_{\tB_1\cdots\tB_L\tC_1^{N_1}\cdots\tC_L^{N_L}}]$. It holds that 
\begin{equation}\label{eq: broadcast bound}
    \begin{aligned}
        &\Tr\qty[\qty(\sum_{k'\not \in\mB(m_l)}\Lambda^{k'}_{\tB_l\tC_{l,1}\cdots\tC_{l,N_l}})\Theta^{(m_1,\ldots, m_L)}_{\tB_1\cdots\tB_L\tC_1^{N_1}\cdots\tC_L^{N_L}}]\\
        &=\frac{1}{2^{nr_l}}\sum_{k\in\mB(m_l)}
        \Tr\qty[\qty(\sum_{k'\not\in\mB(m_l)}\Lambda^{k'}_{\tB_l\tC_{l,1}\cdots\tC_{l,N_l}})\Tr_{\backslash \tB_l\tC_{l,k}}\qty[\mN^{\otimes n}_{\tA\to\tB_1\cdots\tB_L}(\theta^{\otimes n}_{\tA\tC_{l, k}})]\theta_{\tC_{l,s^l_{m_l}}}\otimes \cdots\otimes  \theta_{\tC_{l,k-1}}\otimes\theta_{\tC_{l,k+1}}\otimes \cdots\otimes \theta_{\tC_{l,e^l_{m_l}}}].
    \end{aligned}
\end{equation}
Now, we can follow the same strategy as in Eq.~\eqref{eq: gelfand bound 1} and Eq.~\eqref{eq: gelfand bound 2} to bound Eq.~\eqref{eq: broadcast bound} as
\begin{equation}
    \begin{aligned}
        &\Tr\qty[\qty(\sum_{k'\not \in\mB(m_l)}\Lambda^{k'}_{\tB_l\tC_{l,1}\cdots\tC_{l,N_l}})\Theta^{(m_1,\ldots, m_L)}_{\tB_1\cdots\tB_L\tC_1^{N_1}\cdots\tC_L^{N_L}}]\\
        &\leq c_1\Tr\qty[\qty(I-P^n_{\tB\tC_l}(a_l,\theta_{C_l}))\Tr_{\backslash \tB_l\tC_{l}}\qty[\mN^{\otimes n}_{\tA\to\tB_1\cdots\tB^n}(\theta^{\otimes n}_{\tA\tC_{l}})]]+c_2 \qty(2^{n(R_l+r_l)}-1)\Tr\qty[P^n_{\tB\tC_l}(a_l,\theta_{C_l})\Tr_{\backslash \tB_l}\qty[\mN^{\otimes n}_{\tA\to\tB_1\cdots\tB^n}(\theta^{\otimes n}_{\tA})]\otimes \theta_{\tC_l}].
    \end{aligned}
\end{equation}
Finally, due to Proposition~\ref{pro: semi universal detector}, we have
\begin{equation}
    \begin{aligned}
        &\Tr\qty[\qty(\sum_{k'\not \in\mB(m_l)}\Lambda^{k'}_{\tB_l\tC_{l,1}\cdots\tC_{l,N_l}})\Theta^{(m_1,\ldots, m_L)}_{\tB_1\cdots\tB_L\tC_1^{N_1}\cdots\tC_L^{N_L}}]\\
        &\leq c_1\frac{1}{{\rm poly}(n)}2^{-nt\qty(\pI^{\downarrow}_{1-t}(C_l:B_l)_{\mN(\theta)}-a_l)}+c_2{\rm poly}(n)2^{-n(a_l-R_l-r_l)}, ~\forall t\in(0,1).
    \end{aligned}
\end{equation}
Due to the definition of $a_l$, we can show that the error vanishes asymptotically, which concludes the proof.

\end{proof}

\section{Discussion}
In this work, we have shown that correlation detection—the task of distinguishing a correlated state from the corresponding product of its marginals—can be performed optimally without complete knowledge of the state under test. Building on this result, we used the universal correlation detector to construct a universal decoder and, by combining position-based decoding with convex splitting, obtained universal coding schemes for a variety of information-transmission tasks.
Remarkably, our strategy provides simple constructions of universal codes for Gel'fand–Pinsker channels and universal codes achieving Marton's inner bound for broadcast channels with $L\geq 2$ receivers, which are newly established in this work.

Our results provide a simple and unified framework for constructing universal communication protocols by exploiting permutation symmetry. The connection established here between universal correlation detection and universal decoding may also facilitate the development of universal protocols for other communication settings. Beyond communication, the universal correlation detector may find applications in quantum resource theories and quantum algorithms. More broadly, the ideas developed in this work may offer a useful route toward constructing state-agnostic protocols for other tasks in quantum information theory, such as learning correlations or universal resource distillation~\cite{Matsumoto_universal_entanglement, Watanabe_universal, Rizzo_2026_universal_magic, Lami_2026_universal_quantum_resource_distillation}.

\section*{Disclosure of AI use}
ChatGPT 5.6 Sol and 6 Astra were used to search for previous papers, assist in preparing the manuscript, and finding technical mistakes. The idea originated with the authors.

\section*{Acknowledgments}
We acknowledge the support of JSPS KAKENHI Grant Nos.\ 24K16975, 25K00924, 26H02015, 26KJ0965, the Japan Science and Technology Agency (JST) CREST Grant No.\ JPMJCR23I3, NEXUS Grant Number JPMJNX26C2, PRESTO Grant No.\ JPMJPR24FA, RIKEN iTHEMS, RIKEN Pioneering Project `Mathematical foundation of quantum information' (PI Yasuyuki Kawahigashi), and the World-Leading Innovative Graduate Study Program for Advanced Basic Science Course (WINGS-ABC) at the University of Tokyo.
The authors are indebted to Bartosz Regula and Vishal Singh for fruitful discussions.
KW thanks Masahito Hayashi for bringing the paper~\cite{Dasgupta_2025_universal_tester} to his attention during the poster session at QIP 2025, where they met for the first time.
Part of the discussion of this work took place at the YITP workshop YITP-W-26-07.

\bibliographystyle{apsrmp4-2}
\bibliography{myref_revised}

\begin{thebibliography}{93}%
\makeatletter
\providecommand \@ifxundefined [1]{%
 \@ifx{#1\undefined}
}%
\providecommand \@ifnum [1]{%
 \ifnum #1\expandafter \@firstoftwo
 \else \expandafter \@secondoftwo
 \fi
}%
\providecommand \@ifx [1]{%
 \ifx #1\expandafter \@firstoftwo
 \else \expandafter \@secondoftwo
 \fi
}%
\providecommand \natexlab [1]{#1}%
\providecommand \emph  [1]{``#1''}%
\providecommand \bibnamefont  [1]{#1}%
\providecommand \bibfnamefont [1]{#1}%
\providecommand \citenamefont [1]{#1}%
\providecommand \href@noop [0]{\@secondoftwo}%
\providecommand \href [0]{\begingroup \@sanitize@url \@href}%
\providecommand \@href[1]{\@@startlink{#1}\@@href}%
\providecommand \@@href[1]{\endgroup#1\@@endlink}%
\providecommand \@sanitize@url [0]{\catcode `\\12\catcode `\$12\catcode `\&12\catcode `\#12\catcode `\^12\catcode `\_12\catcode `\%12\relax}%
\providecommand \@@startlink[1]{}%
\providecommand \@@endlink[0]{}%
\providecommand \url  [0]{\begingroup\@sanitize@url \@url }%
\providecommand \@url [1]{\endgroup\@href {#1}{\urlprefix }}%
\providecommand \urlprefix  [0]{URL }%
\providecommand \Eprint [0]{\href }%
\providecommand \doibase [0]{http://dx.doi.org/}%
\providecommand \selectlanguage [0]{\@gobble}%
\providecommand \bibinfo  [0]{\@secondoftwo}%
\providecommand \bibfield  [0]{\@secondoftwo}%
\providecommand \translation [1]{[#1]}%
\providecommand \BibitemOpen [0]{}%
\providecommand \bibitemStop [0]{}%
\providecommand \bibitemNoStop [0]{.\EOS\space}%
\providecommand \EOS [0]{\spacefactor3000\relax}%
\providecommand \BibitemShut  [1]{\csname bibitem#1\endcsname}%
\let\auto@bib@innerbib\@empty
\bibitem [{\citenamefont {Horodecki}\ \emph {et~al.}(2009)\citenamefont {Horodecki}, \citenamefont {Horodecki}, \citenamefont {Horodecki},\ and\ \citenamefont {Horodecki}}]{Horodecki_2009_quantum_entanglement}%
  \BibitemOpen
  \bibfield  {author} {\bibinfo {author} {\bibfnamefont {R.}~\bibnamefont {Horodecki}}, \bibinfo {author} {\bibfnamefont {P.}~\bibnamefont {Horodecki}}, \bibinfo {author} {\bibfnamefont {M.}~\bibnamefont {Horodecki}}, \ and\ \bibinfo {author} {\bibfnamefont {K.}~\bibnamefont {Horodecki}},\ }\bibfield  {title} {\emph {\bibinfo {title} {Quantum entanglement},}\ }\href {http://dx.doi.org/10.1103/revmodphys.81.865} {\bibfield  {journal} {\bibinfo  {journal} {Rev. Mod. Phys.}\ }\textbf {\bibinfo {volume} {81}},\ \bibinfo {pages} {865–942} (\bibinfo {year} {2009})}\BibitemShut {NoStop}%
\bibitem [{\citenamefont {Henderson}\ and\ \citenamefont {Vedral}(2001)}]{L_Henderson_2001}%
  \BibitemOpen
  \bibfield  {author} {\bibinfo {author} {\bibfnamefont {L.}~\bibnamefont {Henderson}}\ and\ \bibinfo {author} {\bibfnamefont {V.}~\bibnamefont {Vedral}},\ }\bibfield  {title} {\emph {\bibinfo {title} {Classical, quantum and total correlations},}\ }\href {http://dx.doi.org/10.1088/0305-4470/34/35/315} {\bibfield  {journal} {\bibinfo  {journal} {J. Phys. A Math. Gen.}\ }\textbf {\bibinfo {volume} {34}},\ \bibinfo {pages} {6899} (\bibinfo {year} {2001})}\BibitemShut {NoStop}%
\bibitem [{\citenamefont {Ollivier}\ and\ \citenamefont {Zurek}(2001)}]{Harold_quantum_discord}%
  \BibitemOpen
  \bibfield  {author} {\bibinfo {author} {\bibfnamefont {H.}~\bibnamefont {Ollivier}}\ and\ \bibinfo {author} {\bibfnamefont {W.~H.}\ \bibnamefont {Zurek}},\ }\bibfield  {title} {\emph {\bibinfo {title} {Quantum discord: A measure of the quantumness of correlations},}\ }\href {http://dx.doi.org/10.1103/PhysRevLett.88.017901} {\bibfield  {journal} {\bibinfo  {journal} {Phys. Rev. Lett.}\ }\textbf {\bibinfo {volume} {88}},\ \bibinfo {pages} {017901} (\bibinfo {year} {2001})}\BibitemShut {NoStop}%
\bibitem [{\citenamefont {Wilde}(2016)}]{Wilde_2016_book}%
  \BibitemOpen
  \bibfield  {author} {\bibinfo {author} {\bibfnamefont {M.~M.}\ \bibnamefont {Wilde}},\ }\href {http://dx.doi.org/10.1017/9781316809976} {\emph {\bibinfo {title} {Quantum Information Theory}}}\ (\bibinfo  {publisher} {Cambridge University Press},\ \bibinfo {year} {2016})\BibitemShut {NoStop}%
\bibitem [{\citenamefont {Khatri}\ and\ \citenamefont {Wilde}(2024{\natexlab{a}})}]{khatri_2024_book}%
  \BibitemOpen
  \bibfield  {author} {\bibinfo {author} {\bibfnamefont {S.}~\bibnamefont {Khatri}}\ and\ \bibinfo {author} {\bibfnamefont {M.~M.}\ \bibnamefont {Wilde}},\ }\href {https://arxiv.org/abs/2011.04672} {\emph {\bibinfo {title} {Principles of quantum communication theory: A modern approach},}\ } (\bibinfo {year} {2024}{\natexlab{a}}),\ \Eprint {http://arxiv.org/abs/2011.04672} {arXiv:2011.04672 [quant-ph]} \BibitemShut {NoStop}%
\bibitem [{\citenamefont {Chitambar}\ and\ \citenamefont {Gour}(2019)}]{Chitamber_gour}%
  \BibitemOpen
  \bibfield  {author} {\bibinfo {author} {\bibfnamefont {E.}~\bibnamefont {Chitambar}}\ and\ \bibinfo {author} {\bibfnamefont {G.}~\bibnamefont {Gour}},\ }\bibfield  {title} {\emph {\bibinfo {title} {Quantum resource theories},}\ }\href {http://dx.doi.org/10.1103/RevModPhys.91.025001} {\bibfield  {journal} {\bibinfo  {journal} {Rev. Mod. Phys.}\ }\textbf {\bibinfo {volume} {91}},\ \bibinfo {pages} {025001} (\bibinfo {year} {2019})}\BibitemShut {NoStop}%
\bibitem [{\citenamefont {Gour}(2025)}]{Gour_2025_book}%
  \BibitemOpen
  \bibfield  {author} {\bibinfo {author} {\bibfnamefont {G.}~\bibnamefont {Gour}},\ }\href {http://dx.doi.org/10.1017/9781009560870} {\emph {\bibinfo {title} {Quantum Resource Theories}}}\ (\bibinfo  {publisher} {Cambridge University Press},\ \bibinfo {year} {2025})\BibitemShut {NoStop}%
\bibitem [{\citenamefont {Hayashi}\ and\ \citenamefont {Tomamichel}(2016)}]{Hayashi_2016_correlation_detection}%
  \BibitemOpen
  \bibfield  {author} {\bibinfo {author} {\bibfnamefont {M.}~\bibnamefont {Hayashi}}\ and\ \bibinfo {author} {\bibfnamefont {M.}~\bibnamefont {Tomamichel}},\ }\bibfield  {title} {\emph {\bibinfo {title} {Correlation detection and an operational interpretation of the rényi mutual information},}\ }\href {http://dx.doi.org/10.1063/1.4964755} {\bibfield  {journal} {\bibinfo  {journal} {J. Math. Phys.}\ }\textbf {\bibinfo {volume} {57}},\ \bibinfo {pages} {102201} (\bibinfo {year} {2016})}\BibitemShut {NoStop}%
\bibitem [{\citenamefont {Girardi}\ \emph {et~al.}(2025{\natexlab{a}})\citenamefont {Girardi}, \citenamefont {Oufkir}, \citenamefont {Regula}, \citenamefont {Tomamichel}, \citenamefont {Berta},\ and\ \citenamefont {Lami}}]{girardi2025umlaut_information}%
  \BibitemOpen
  \bibfield  {author} {\bibinfo {author} {\bibfnamefont {F.}~\bibnamefont {Girardi}}, \bibinfo {author} {\bibfnamefont {A.}~\bibnamefont {Oufkir}}, \bibinfo {author} {\bibfnamefont {B.}~\bibnamefont {Regula}}, \bibinfo {author} {\bibfnamefont {M.}~\bibnamefont {Tomamichel}}, \bibinfo {author} {\bibfnamefont {M.}~\bibnamefont {Berta}}, \ and\ \bibinfo {author} {\bibfnamefont {L.}~\bibnamefont {Lami}},\ }\href {https://arxiv.org/abs/2503.18910} {\emph {\bibinfo {title} {Umlaut information},}\ } (\bibinfo {year} {2025}{\natexlab{a}}),\ \Eprint {http://arxiv.org/abs/2503.18910} {arXiv:2503.18910 [cs.IT]} \BibitemShut {NoStop}%
\bibitem [{\citenamefont {Girardi}\ \emph {et~al.}(2025{\natexlab{b}})\citenamefont {Girardi}, \citenamefont {Oufkir}, \citenamefont {Regula}, \citenamefont {Tomamichel}, \citenamefont {Berta},\ and\ \citenamefont {Lami}}]{girardi2025quantum_umlaut}%
  \BibitemOpen
  \bibfield  {author} {\bibinfo {author} {\bibfnamefont {F.}~\bibnamefont {Girardi}}, \bibinfo {author} {\bibfnamefont {A.}~\bibnamefont {Oufkir}}, \bibinfo {author} {\bibfnamefont {B.}~\bibnamefont {Regula}}, \bibinfo {author} {\bibfnamefont {M.}~\bibnamefont {Tomamichel}}, \bibinfo {author} {\bibfnamefont {M.}~\bibnamefont {Berta}}, \ and\ \bibinfo {author} {\bibfnamefont {L.}~\bibnamefont {Lami}},\ }\href {https://arxiv.org/abs/2503.21479} {\emph {\bibinfo {title} {Quantum umlaut information},}\ } (\bibinfo {year} {2025}{\natexlab{b}}),\ \Eprint {http://arxiv.org/abs/2503.21479} {arXiv:2503.21479 [quant-ph]} \BibitemShut {NoStop}%
\bibitem [{\citenamefont {Dasgupta}\ \emph {et~al.}(2025)\citenamefont {Dasgupta}, \citenamefont {Warsi},\ and\ \citenamefont {Hayashi}}]{Dasgupta_2025_universal_tester}%
  \BibitemOpen
  \bibfield  {author} {\bibinfo {author} {\bibfnamefont {A.}~\bibnamefont {Dasgupta}}, \bibinfo {author} {\bibfnamefont {N.~A.}\ \bibnamefont {Warsi}}, \ and\ \bibinfo {author} {\bibfnamefont {M.}~\bibnamefont {Hayashi}},\ }\bibfield  {title} {\emph {\bibinfo {title} {Universal tester for multiple independence testing and classical-quantum arbitrarily varying multiple access channel},}\ }\href {http://dx.doi.org/10.1109/tit.2025.3538870} {\bibfield  {journal} {\bibinfo  {journal} {IEEE Trans. Inf. Theory}\ }\textbf {\bibinfo {volume} {71}},\ \bibinfo {pages} {3719–3765} (\bibinfo {year} {2025})}\BibitemShut {NoStop}%
\bibitem [{\citenamefont {Khatri}\ \emph {et~al.}(2021)\citenamefont {Khatri}, \citenamefont {Kaur}, \citenamefont {Guha},\ and\ \citenamefont {Wilde}}]{khatri_2021_second_order}%
  \BibitemOpen
  \bibfield  {author} {\bibinfo {author} {\bibfnamefont {S.}~\bibnamefont {Khatri}}, \bibinfo {author} {\bibfnamefont {E.}~\bibnamefont {Kaur}}, \bibinfo {author} {\bibfnamefont {S.}~\bibnamefont {Guha}}, \ and\ \bibinfo {author} {\bibfnamefont {M.~M.}\ \bibnamefont {Wilde}},\ }\href {https://arxiv.org/abs/1910.03883} {\emph {\bibinfo {title} {Second-order coding rates for key distillation in quantum key distribution},}\ } (\bibinfo {year} {2021}),\ \Eprint {http://arxiv.org/abs/1910.03883} {arXiv:1910.03883 [quant-ph]} \BibitemShut {NoStop}%
\bibitem [{\citenamefont {Holevo}(1998)}]{Holevo_1996_capacity}%
  \BibitemOpen
  \bibfield  {author} {\bibinfo {author} {\bibfnamefont {A.}~\bibnamefont {Holevo}},\ }\bibfield  {title} {\emph {\bibinfo {title} {The capacity of the quantum channel with general signal states},}\ }\href {http://dx.doi.org/10.1109/18.651037} {\bibfield  {journal} {\bibinfo  {journal} {IEEE Trans. Inf. Theory}\ }\textbf {\bibinfo {volume} {44}},\ \bibinfo {pages} {269} (\bibinfo {year} {1998})}\BibitemShut {NoStop}%
\bibitem [{\citenamefont {Schumacher}\ and\ \citenamefont {Westmoreland}(1997)}]{SW_theorem}%
  \BibitemOpen
  \bibfield  {author} {\bibinfo {author} {\bibfnamefont {B.}~\bibnamefont {Schumacher}}\ and\ \bibinfo {author} {\bibfnamefont {M.~D.}\ \bibnamefont {Westmoreland}},\ }\bibfield  {title} {\emph {\bibinfo {title} {Sending classical information via noisy quantum channels},}\ }\href {http://dx.doi.org/10.1103/PhysRevA.56.131} {\bibfield  {journal} {\bibinfo  {journal} {Phys. Rev. A}\ }\textbf {\bibinfo {volume} {56}},\ \bibinfo {pages} {131} (\bibinfo {year} {1997})}\BibitemShut {NoStop}%
\bibitem [{\citenamefont {Bennett}\ \emph {et~al.}(2002)\citenamefont {Bennett}, \citenamefont {Shor}, \citenamefont {Smolin},\ and\ \citenamefont {Thapliyal}}]{Bennett_2002_Entanglement_assisted}%
  \BibitemOpen
  \bibfield  {author} {\bibinfo {author} {\bibfnamefont {C.}~\bibnamefont {Bennett}}, \bibinfo {author} {\bibfnamefont {P.}~\bibnamefont {Shor}}, \bibinfo {author} {\bibfnamefont {J.}~\bibnamefont {Smolin}}, \ and\ \bibinfo {author} {\bibfnamefont {A.}~\bibnamefont {Thapliyal}},\ }\bibfield  {title} {\emph {\bibinfo {title} {Entanglement-assisted capacity of a quantum channel and the reverse shannon theorem},}\ }\href {http://dx.doi.org/10.1109/TIT.2002.802612} {\bibfield  {journal} {\bibinfo  {journal} {IEEE Trans. Inf. Theory}\ }\textbf {\bibinfo {volume} {48}},\ \bibinfo {pages} {2637} (\bibinfo {year} {2002})}\BibitemShut {NoStop}%
\bibitem [{\citenamefont {Holevo}(2002)}]{Holevo_2002_entanglement_assisted}%
  \BibitemOpen
  \bibfield  {author} {\bibinfo {author} {\bibfnamefont {A.~S.}\ \bibnamefont {Holevo}},\ }\bibfield  {title} {\emph {\bibinfo {title} {On entanglement-assisted classical capacity},}\ }\href {http://dx.doi.org/10.1063/1.1495877} {\bibfield  {journal} {\bibinfo  {journal} {J. Math. Phys.}\ }\textbf {\bibinfo {volume} {43}},\ \bibinfo {pages} {4326–4333} (\bibinfo {year} {2002})}\BibitemShut {NoStop}%
\bibitem [{\citenamefont {Dupuis}(2009)}]{Dupuis_2009_capacity_side_information}%
  \BibitemOpen
  \bibfield  {author} {\bibinfo {author} {\bibfnamefont {F.}~\bibnamefont {Dupuis}},\ }\bibfield  {title} {\emph {\bibinfo {title} {The capacity of quantum channels with side information at the transmitter},}\ }in\ \href {http://dx.doi.org/10.1109/isit.2009.5205591} {\emph {\bibinfo {booktitle} {2009 IEEE International Symposium on Information Theory}}}\ (\bibinfo  {publisher} {IEEE},\ \bibinfo {year} {2009})\ p.\ \bibinfo {pages} {948–952}\BibitemShut {NoStop}%
\bibitem [{\citenamefont {Devetak}(2005)}]{Devetak_2004_private_classical_capacity}%
  \BibitemOpen
  \bibfield  {author} {\bibinfo {author} {\bibfnamefont {I.}~\bibnamefont {Devetak}},\ }\bibfield  {title} {\emph {\bibinfo {title} {The private classical capacity and quantum capacity of a quantum channel},}\ }\href {http://dx.doi.org/10.1109/TIT.2004.839515} {\bibfield  {journal} {\bibinfo  {journal} {IEEE Trans. Inf. Theory}\ }\textbf {\bibinfo {volume} {51}},\ \bibinfo {pages} {44} (\bibinfo {year} {2005})}\BibitemShut {NoStop}%
\bibitem [{\citenamefont {Lloyd}(1997)}]{Lloyd_1997_capacity}%
  \BibitemOpen
  \bibfield  {author} {\bibinfo {author} {\bibfnamefont {S.}~\bibnamefont {Lloyd}},\ }\bibfield  {title} {\emph {\bibinfo {title} {Capacity of the noisy quantum channel},}\ }\href {http://dx.doi.org/10.1103/physreva.55.1613} {\bibfield  {journal} {\bibinfo  {journal} {Phys. Rev. A}\ }\textbf {\bibinfo {volume} {55}},\ \bibinfo {pages} {1613–1622} (\bibinfo {year} {1997})}\BibitemShut {NoStop}%
\bibitem [{\citenamefont {Anshu}\ \emph {et~al.}(2019{\natexlab{a}})\citenamefont {Anshu}, \citenamefont {Jain},\ and\ \citenamefont {Warsi}}]{Anshu_2019_building}%
  \BibitemOpen
  \bibfield  {author} {\bibinfo {author} {\bibfnamefont {A.}~\bibnamefont {Anshu}}, \bibinfo {author} {\bibfnamefont {R.}~\bibnamefont {Jain}}, \ and\ \bibinfo {author} {\bibfnamefont {N.~A.}\ \bibnamefont {Warsi}},\ }\bibfield  {title} {\emph {\bibinfo {title} {Building blocks for communication over noisy quantum networks},}\ }\href {http://dx.doi.org/10.1109/tit.2018.2851297} {\bibfield  {journal} {\bibinfo  {journal} {IEEE Trans. Inf. Theory}\ }\textbf {\bibinfo {volume} {65}},\ \bibinfo {pages} {1287–1306} (\bibinfo {year} {2019}{\natexlab{a}})}\BibitemShut {NoStop}%
\bibitem [{\citenamefont {Anshu}\ \emph {et~al.}(2019{\natexlab{b}})\citenamefont {Anshu}, \citenamefont {Jain},\ and\ \citenamefont {Warsi}}]{Anshu_2019_on_the_near}%
  \BibitemOpen
  \bibfield  {author} {\bibinfo {author} {\bibfnamefont {A.}~\bibnamefont {Anshu}}, \bibinfo {author} {\bibfnamefont {R.}~\bibnamefont {Jain}}, \ and\ \bibinfo {author} {\bibfnamefont {N.~A.}\ \bibnamefont {Warsi}},\ }\bibfield  {title} {\emph {\bibinfo {title} {On the near-optimality of one-shot classical communication over quantum channels},}\ }\href {http://dx.doi.org/10.1063/1.5039796} {\bibfield  {journal} {\bibinfo  {journal} {J. Math. Phys.}\ }\textbf {\bibinfo {volume} {60}},\ \bibinfo {pages} {012204} (\bibinfo {year} {2019}{\natexlab{b}})}\BibitemShut {NoStop}%
\bibitem [{\citenamefont {Wilde}(2017)}]{Wilde_2017_position_based}%
  \BibitemOpen
  \bibfield  {author} {\bibinfo {author} {\bibfnamefont {M.~M.}\ \bibnamefont {Wilde}},\ }\bibfield  {title} {\emph {\bibinfo {title} {Position-based coding and convex splitting for private communication over quantum channels},}\ }\href {http://dx.doi.org/10.1007/s11128-017-1718-4} {\bibfield  {journal} {\bibinfo  {journal} {Quantum Inf. Process.}\ }\textbf {\bibinfo {volume} {16}},\ \bibinfo {pages} {264} (\bibinfo {year} {2017})}\BibitemShut {NoStop}%
\bibitem [{\citenamefont {Umegaki}(1962)}]{Umegaki_relative}%
  \BibitemOpen
  \bibfield  {author} {\bibinfo {author} {\bibfnamefont {H.}~\bibnamefont {Umegaki}},\ }\bibfield  {title} {\emph {\bibinfo {title} {Conditional expectation in an operator algebra, iv (entropy and information)},}\ }\href {http://dx.doi.org/10.2996/kmj/1138844604} {\bibfield  {journal} {\bibinfo  {journal} {Kodai Math. Semin. Rep.}\ }\textbf {\bibinfo {volume} {14}},\ \bibinfo {pages} {59} (\bibinfo {year} {1962})}\BibitemShut {NoStop}%
\bibitem [{\citenamefont {Petz}(1986)}]{Petz_1986_quasi_entropy}%
  \BibitemOpen
  \bibfield  {author} {\bibinfo {author} {\bibfnamefont {D.}~\bibnamefont {Petz}},\ }\bibfield  {title} {\emph {\bibinfo {title} {Quasi-entropies for finite quantum systems},}\ }\href {http://dx.doi.org/https://doi.org/10.1016/0034-4877(86)90067-4} {\bibfield  {journal} {\bibinfo  {journal} {Rep. Math. Phys.}\ }\textbf {\bibinfo {volume} {23}},\ \bibinfo {pages} {57} (\bibinfo {year} {1986})}\BibitemShut {NoStop}%
\bibitem [{\citenamefont {Müller-Lennert}\ \emph {et~al.}(2013)\citenamefont {Müller-Lennert}, \citenamefont {Dupuis}, \citenamefont {Szehr}, \citenamefont {Fehr},\ and\ \citenamefont {Tomamichel}}]{Muller_Lennert_2013_on_quantum_renyi}%
  \BibitemOpen
  \bibfield  {author} {\bibinfo {author} {\bibfnamefont {M.}~\bibnamefont {Müller-Lennert}}, \bibinfo {author} {\bibfnamefont {F.}~\bibnamefont {Dupuis}}, \bibinfo {author} {\bibfnamefont {O.}~\bibnamefont {Szehr}}, \bibinfo {author} {\bibfnamefont {S.}~\bibnamefont {Fehr}}, \ and\ \bibinfo {author} {\bibfnamefont {M.}~\bibnamefont {Tomamichel}},\ }\bibfield  {title} {\emph {\bibinfo {title} {On quantum {R}ényi entropies: A new generalization and some properties},}\ }\href {http://dx.doi.org/10.1063/1.4838856} {\bibfield  {journal} {\bibinfo  {journal} {J. Math. Phys.}\ }\textbf {\bibinfo {volume} {54}},\ \bibinfo {pages} {122203} (\bibinfo {year} {2013})}\BibitemShut {NoStop}%
\bibitem [{\citenamefont {Wilde}\ \emph {et~al.}(2014)\citenamefont {Wilde}, \citenamefont {Winter},\ and\ \citenamefont {Yang}}]{Wilde_2014_strong_converse}%
  \BibitemOpen
  \bibfield  {author} {\bibinfo {author} {\bibfnamefont {M.~M.}\ \bibnamefont {Wilde}}, \bibinfo {author} {\bibfnamefont {A.}~\bibnamefont {Winter}}, \ and\ \bibinfo {author} {\bibfnamefont {D.}~\bibnamefont {Yang}},\ }\bibfield  {title} {\emph {\bibinfo {title} {Strong converse for the classical capacity of entanglement-breaking and hadamard channels via a sandwiched {R}ényi relative entropy},}\ }\href {http://dx.doi.org/10.1007/s00220-014-2122-x} {\bibfield  {journal} {\bibinfo  {journal} {Commun. Math. Phys.}\ }\textbf {\bibinfo {volume} {331}},\ \bibinfo {pages} {593–622} (\bibinfo {year} {2014})}\BibitemShut {NoStop}%
\bibitem [{\citenamefont {Modi}\ \emph {et~al.}(2010)\citenamefont {Modi}, \citenamefont {Paterek}, \citenamefont {Son}, \citenamefont {Vedral},\ and\ \citenamefont {Williamson}}]{Modi_2010_unified}%
  \BibitemOpen
  \bibfield  {author} {\bibinfo {author} {\bibfnamefont {K.}~\bibnamefont {Modi}}, \bibinfo {author} {\bibfnamefont {T.}~\bibnamefont {Paterek}}, \bibinfo {author} {\bibfnamefont {W.}~\bibnamefont {Son}}, \bibinfo {author} {\bibfnamefont {V.}~\bibnamefont {Vedral}}, \ and\ \bibinfo {author} {\bibfnamefont {M.}~\bibnamefont {Williamson}},\ }\bibfield  {title} {\emph {\bibinfo {title} {Unified view of quantum and classical correlations},}\ }\href {http://dx.doi.org/10.1103/physrevlett.104.080501} {\bibfield  {journal} {\bibinfo  {journal} {Phys. Rev. Lett.}\ }\textbf {\bibinfo {volume} {104}},\ \bibinfo {pages} {080501} (\bibinfo {year} {2010})}\BibitemShut {NoStop}%
\bibitem [{\citenamefont {Avis}\ \emph {et~al.}(2008)\citenamefont {Avis}, \citenamefont {Hayden},\ and\ \citenamefont {Savov}}]{Avis_2008_distributed_compression}%
  \BibitemOpen
  \bibfield  {author} {\bibinfo {author} {\bibfnamefont {D.}~\bibnamefont {Avis}}, \bibinfo {author} {\bibfnamefont {P.}~\bibnamefont {Hayden}}, \ and\ \bibinfo {author} {\bibfnamefont {I.}~\bibnamefont {Savov}},\ }\bibfield  {title} {\emph {\bibinfo {title} {Distributed compression and multiparty squashed entanglement},}\ }\href {http://dx.doi.org/10.1088/1751-8113/41/11/115301} {\bibfield  {journal} {\bibinfo  {journal} {J. Phys. A Math. Theor.}\ }\textbf {\bibinfo {volume} {41}},\ \bibinfo {pages} {115301} (\bibinfo {year} {2008})}\BibitemShut {NoStop}%
\bibitem [{\citenamefont {Sharma}\ and\ \citenamefont {Warsi}(2013)}]{Sharma_2013_fundamental_bound}%
  \BibitemOpen
  \bibfield  {author} {\bibinfo {author} {\bibfnamefont {N.}~\bibnamefont {Sharma}}\ and\ \bibinfo {author} {\bibfnamefont {N.~A.}\ \bibnamefont {Warsi}},\ }\bibfield  {title} {\emph {\bibinfo {title} {Fundamental bound on the reliability of quantum information transmission},}\ }\href {http://dx.doi.org/10.1103/physrevlett.110.080501} {\bibfield  {journal} {\bibinfo  {journal} {Phys. Rev. Lett.}\ }\textbf {\bibinfo {volume} {110}},\ \bibinfo {pages} {080501} (\bibinfo {year} {2013})}\BibitemShut {NoStop}%
\bibitem [{\citenamefont {Berta}\ \emph {et~al.}(2025)\citenamefont {Berta}, \citenamefont {Cheng},\ and\ \citenamefont {Gao}}]{Cheng_2023_quantum_broadcast}%
  \BibitemOpen
  \bibfield  {author} {\bibinfo {author} {\bibfnamefont {M.}~\bibnamefont {Berta}}, \bibinfo {author} {\bibfnamefont {H.-C.}\ \bibnamefont {Cheng}}, \ and\ \bibinfo {author} {\bibfnamefont {L.}~\bibnamefont {Gao}},\ }\bibfield  {title} {\emph {\bibinfo {title} {Quantum {Broadcast} {Channel} {Simulation} via {Multipartite} {Convex} {Splitting}},}\ }\href {http://dx.doi.org/10.1007/s00220-024-05191-4} {\bibfield  {journal} {\bibinfo  {journal} {Communications in Mathematical Physics}\ }\textbf {\bibinfo {volume} {406}},\ \bibinfo {pages} {36} (\bibinfo {year} {2025})}\BibitemShut {NoStop}%
\bibitem [{\citenamefont {Hayashi}(2009)}]{Hayashi_2009_universal_coding}%
  \BibitemOpen
  \bibfield  {author} {\bibinfo {author} {\bibfnamefont {M.}~\bibnamefont {Hayashi}},\ }\bibfield  {title} {\emph {\bibinfo {title} {Universal coding for classical-quantum channel},}\ }\href {http://dx.doi.org/10.1007/s00220-009-0825-1} {\bibfield  {journal} {\bibinfo  {journal} {Commun. Math. Phys.}\ }\textbf {\bibinfo {volume} {289}},\ \bibinfo {pages} {1087–1098} (\bibinfo {year} {2009})}\BibitemShut {NoStop}%
\bibitem [{\citenamefont {Matsuura}\ \emph {et~al.}(2025)\citenamefont {Matsuura}, \citenamefont {Hayashi},\ and\ \citenamefont {Hsieh}}]{Matsuura_2025_universal_resolvability}%
  \BibitemOpen
  \bibfield  {author} {\bibinfo {author} {\bibfnamefont {T.}~\bibnamefont {Matsuura}}, \bibinfo {author} {\bibfnamefont {M.}~\bibnamefont {Hayashi}}, \ and\ \bibinfo {author} {\bibfnamefont {M.-H.}\ \bibnamefont {Hsieh}},\ }\href {https://arxiv.org/abs/2510.02883} {\emph {\bibinfo {title} {Universal classical-quantum channel resolvability and private channel coding},}\ } (\bibinfo {year} {2025}),\ \Eprint {http://arxiv.org/abs/2510.02883} {arXiv:2510.02883 [quant-ph]} \BibitemShut {NoStop}%
\bibitem [{\citenamefont {Fang}\ and\ \citenamefont {Hayashi}(2026)}]{Fang_2026_error_exponent}%
  \BibitemOpen
  \bibfield  {author} {\bibinfo {author} {\bibfnamefont {K.}~\bibnamefont {Fang}}\ and\ \bibinfo {author} {\bibfnamefont {M.}~\bibnamefont {Hayashi}},\ }\bibfield  {title} {\emph {\bibinfo {title} {Error exponents of quantum state discrimination with composite correlated hypotheses},}\ }\href {http://dx.doi.org/10.1109/tit.2026.3684314} {\bibfield  {journal} {\bibinfo  {journal} {IEEE Trans. Inf. Theory}\ }\textbf {\bibinfo {volume} {72}},\ \bibinfo {pages} {4140–4165} (\bibinfo {year} {2026})}\BibitemShut {NoStop}%
\bibitem [{\citenamefont {Cover}\ and\ \citenamefont {Thomas}(2006)}]{cover_1999_elements}%
  \BibitemOpen
  \bibfield  {author} {\bibinfo {author} {\bibfnamefont {T.~M.}\ \bibnamefont {Cover}}\ and\ \bibinfo {author} {\bibfnamefont {J.~A.}\ \bibnamefont {Thomas}},\ }\href {http://dx.doi.org/10.1002/047174882X} {\emph {\bibinfo {title} {Elements of Information Theory}}},\ \bibinfo {edition} {2nd}\ ed.\ (\bibinfo  {publisher} {John Wiley \& Sons},\ \bibinfo {year} {2006})\BibitemShut {NoStop}%
\bibitem [{\citenamefont {Anshu}\ \emph {et~al.}(2017)\citenamefont {Anshu}, \citenamefont {Devabathini},\ and\ \citenamefont {Jain}}]{Anshu2017quantum}%
  \BibitemOpen
  \bibfield  {author} {\bibinfo {author} {\bibfnamefont {A.}~\bibnamefont {Anshu}}, \bibinfo {author} {\bibfnamefont {V.~K.}\ \bibnamefont {Devabathini}}, \ and\ \bibinfo {author} {\bibfnamefont {R.}~\bibnamefont {Jain}},\ }\bibfield  {title} {\emph {\bibinfo {title} {Quantum communication using coherent rejection sampling},}\ }\href {http://dx.doi.org/10.1103/PhysRevLett.119.120506} {\bibfield  {journal} {\bibinfo  {journal} {Phys. Rev. Lett.}\ }\textbf {\bibinfo {volume} {119}},\ \bibinfo {pages} {120506} (\bibinfo {year} {2017})}\BibitemShut {NoStop}%
\bibitem [{\citenamefont {Cheng}\ and\ \citenamefont {Gao}(2025)}]{Cheng_2023_tight_convex_splitting}%
  \BibitemOpen
  \bibfield  {author} {\bibinfo {author} {\bibfnamefont {H.-C.}\ \bibnamefont {Cheng}}\ and\ \bibinfo {author} {\bibfnamefont {L.}~\bibnamefont {Gao}},\ }\bibfield  {title} {\emph {\bibinfo {title} {Tight one-shot analysis for convex splitting with applications in quantum information theory},}\ }\href {http://dx.doi.org/10.1109/TIT.2025.3612051} {\bibfield  {journal} {\bibinfo  {journal} {IEEE Transactions on Information Theory}\ }\textbf {\bibinfo {volume} {71}},\ \bibinfo {pages} {8573} (\bibinfo {year} {2025})}\BibitemShut {NoStop}%
\bibitem [{\citenamefont {Fuchs}\ and\ \citenamefont {van~de Graaf}(1999)}]{Fuchs_van_de_graaf}%
  \BibitemOpen
  \bibfield  {author} {\bibinfo {author} {\bibfnamefont {C.}~\bibnamefont {Fuchs}}\ and\ \bibinfo {author} {\bibfnamefont {J.}~\bibnamefont {van~de Graaf}},\ }\bibfield  {title} {\emph {\bibinfo {title} {Cryptographic distinguishability measures for quantum-mechanical states},}\ }\href {http://dx.doi.org/10.1109/18.761271} {\bibfield  {journal} {\bibinfo  {journal} {IEEE Trans. Inf. Theory}\ }\textbf {\bibinfo {volume} {45}},\ \bibinfo {pages} {1216} (\bibinfo {year} {1999})}\BibitemShut {NoStop}%
\bibitem [{\citenamefont {Hayashi}\ and\ \citenamefont {Nagaoka}(2002)}]{Hayashi_Nagaoka}%
  \BibitemOpen
  \bibfield  {author} {\bibinfo {author} {\bibfnamefont {M.}~\bibnamefont {Hayashi}}\ and\ \bibinfo {author} {\bibfnamefont {H.}~\bibnamefont {Nagaoka}},\ }\bibfield  {title} {\emph {\bibinfo {title} {A general formula for the classical capacity of a general quantum channel},}\ }in\ \href {http://dx.doi.org/10.1109/ISIT.2002.1023343} {\emph {\bibinfo {booktitle} {IEEE Int. Symp. Inf. Theory - Proc.,}}}\ (\bibinfo {year} {2002})\ p.~\bibinfo {pages} {71}\BibitemShut {NoStop}%
\bibitem [{\citenamefont {Anshu}\ \emph {et~al.}(2019{\natexlab{c}})\citenamefont {Anshu}, \citenamefont {Jain},\ and\ \citenamefont {Warsi}}]{Anshu_2019_hypothesis_testing}%
  \BibitemOpen
  \bibfield  {author} {\bibinfo {author} {\bibfnamefont {A.}~\bibnamefont {Anshu}}, \bibinfo {author} {\bibfnamefont {R.}~\bibnamefont {Jain}}, \ and\ \bibinfo {author} {\bibfnamefont {N.~A.}\ \bibnamefont {Warsi}},\ }\bibfield  {title} {\emph {\bibinfo {title} {A hypothesis testing approach for communication over entanglement-assisted compound quantum channel},}\ }\href {http://dx.doi.org/10.1109/tit.2018.2876280} {\bibfield  {journal} {\bibinfo  {journal} {IEEE Trans. Inf. Theory}\ }\textbf {\bibinfo {volume} {65}},\ \bibinfo {pages} {2623–2636} (\bibinfo {year} {2019}{\natexlab{c}})}\BibitemShut {NoStop}%
\bibitem [{\citenamefont {Polyanskiy}\ \emph {et~al.}(2010)\citenamefont {Polyanskiy}, \citenamefont {Poor},\ and\ \citenamefont {Verdu}}]{Polyanskiy_2010_channel_coding}%
  \BibitemOpen
  \bibfield  {author} {\bibinfo {author} {\bibfnamefont {Y.}~\bibnamefont {Polyanskiy}}, \bibinfo {author} {\bibfnamefont {H.~V.}\ \bibnamefont {Poor}}, \ and\ \bibinfo {author} {\bibfnamefont {S.}~\bibnamefont {Verdu}},\ }\bibfield  {title} {\emph {\bibinfo {title} {Channel coding rate in the finite blocklength regime},}\ }\href {http://dx.doi.org/10.1109/TIT.2010.2043769} {\bibfield  {journal} {\bibinfo  {journal} {IEEE Trans. Inf. Theory}\ }\textbf {\bibinfo {volume} {56}},\ \bibinfo {pages} {2307} (\bibinfo {year} {2010})}\BibitemShut {NoStop}%
\bibitem [{\citenamefont {Cheng}(2023)}]{Cheng_2023_simple_and_tighter}%
  \BibitemOpen
  \bibfield  {author} {\bibinfo {author} {\bibfnamefont {H.-C.}\ \bibnamefont {Cheng}},\ }\bibfield  {title} {\emph {\bibinfo {title} {Simple and tighter derivation of achievability for classical communication over quantum channels},}\ }\href {http://dx.doi.org/10.1103/PRXQuantum.4.040330} {\bibfield  {journal} {\bibinfo  {journal} {PRX Quantum}\ }\textbf {\bibinfo {volume} {4}},\ \bibinfo {pages} {040330} (\bibinfo {year} {2023})}\BibitemShut {NoStop}%
\bibitem [{\citenamefont {Liu}\ \emph {et~al.}(2019)\citenamefont {Liu}, \citenamefont {Bu},\ and\ \citenamefont {Takagi}}]{Liu_one_shot}%
  \BibitemOpen
  \bibfield  {author} {\bibinfo {author} {\bibfnamefont {Z.-W.}\ \bibnamefont {Liu}}, \bibinfo {author} {\bibfnamefont {K.}~\bibnamefont {Bu}}, \ and\ \bibinfo {author} {\bibfnamefont {R.}~\bibnamefont {Takagi}},\ }\bibfield  {title} {\emph {\bibinfo {title} {One-shot operational quantum resource theory},}\ }\href {http://dx.doi.org/10.1103/PhysRevLett.123.020401} {\bibfield  {journal} {\bibinfo  {journal} {Phys. Rev. Lett.}\ }\textbf {\bibinfo {volume} {123}},\ \bibinfo {pages} {020401} (\bibinfo {year} {2019})}\BibitemShut {NoStop}%
\bibitem [{\citenamefont {Takagi}\ \emph {et~al.}(2022)\citenamefont {Takagi}, \citenamefont {Regula},\ and\ \citenamefont {Wilde}}]{Takagi_One_shot}%
  \BibitemOpen
  \bibfield  {author} {\bibinfo {author} {\bibfnamefont {R.}~\bibnamefont {Takagi}}, \bibinfo {author} {\bibfnamefont {B.}~\bibnamefont {Regula}}, \ and\ \bibinfo {author} {\bibfnamefont {M.~M.}\ \bibnamefont {Wilde}},\ }\bibfield  {title} {\emph {\bibinfo {title} {One-shot yield-cost relations in general quantum resource theories},}\ }\href {http://dx.doi.org/10.1103/PRXQuantum.3.010348} {\bibfield  {journal} {\bibinfo  {journal} {PRX Quantum}\ }\textbf {\bibinfo {volume} {3}},\ \bibinfo {pages} {010348} (\bibinfo {year} {2022})}\BibitemShut {NoStop}%
\bibitem [{\citenamefont {Regula}\ and\ \citenamefont {Takagi}(2021)}]{Regula_Takagi_2021}%
  \BibitemOpen
  \bibfield  {author} {\bibinfo {author} {\bibfnamefont {B.}~\bibnamefont {Regula}}\ and\ \bibinfo {author} {\bibfnamefont {R.}~\bibnamefont {Takagi}},\ }\bibfield  {title} {\emph {\bibinfo {title} {One-shot manipulation of dynamical quantum resources},}\ }\href {http://dx.doi.org/10.1103/physrevlett.127.060402} {\bibfield  {journal} {\bibinfo  {journal} {Phys. Rev. Lett.}\ }\textbf {\bibinfo {volume} {127}},\ \bibinfo {pages} {060402} (\bibinfo {year} {2021})}\BibitemShut {NoStop}%
\bibitem [{\citenamefont {Regula}\ \emph {et~al.}(2020)\citenamefont {Regula}, \citenamefont {Bu}, \citenamefont {Takagi},\ and\ \citenamefont {Liu}}]{Regula_benchmarking}%
  \BibitemOpen
  \bibfield  {author} {\bibinfo {author} {\bibfnamefont {B.}~\bibnamefont {Regula}}, \bibinfo {author} {\bibfnamefont {K.}~\bibnamefont {Bu}}, \bibinfo {author} {\bibfnamefont {R.}~\bibnamefont {Takagi}}, \ and\ \bibinfo {author} {\bibfnamefont {Z.-W.}\ \bibnamefont {Liu}},\ }\bibfield  {title} {\emph {\bibinfo {title} {Benchmarking one-shot distillation in general quantum resource theories},}\ }\href {http://dx.doi.org/10.1103/PhysRevA.101.062315} {\bibfield  {journal} {\bibinfo  {journal} {Phys. Rev. A}\ }\textbf {\bibinfo {volume} {101}},\ \bibinfo {pages} {062315} (\bibinfo {year} {2020})}\BibitemShut {NoStop}%
\bibitem [{\citenamefont {Hayashi}\ and\ \citenamefont {Yamasaki}(2025)}]{hayashi_generalized_2025}%
  \BibitemOpen
  \bibfield  {author} {\bibinfo {author} {\bibfnamefont {M.}~\bibnamefont {Hayashi}}\ and\ \bibinfo {author} {\bibfnamefont {H.}~\bibnamefont {Yamasaki}},\ }\bibfield  {title} {\emph {\bibinfo {title} {The generalized quantum {Stein}’s lemma and the second law of quantum resource theories},}\ }\href {http://dx.doi.org/10.1038/s41567-025-03047-9} {\bibfield  {journal} {\bibinfo  {journal} {Nat. Phys.}\ }\textbf {\bibinfo {volume} {21}},\ \bibinfo {pages} {1988} (\bibinfo {year} {2025})}\BibitemShut {NoStop}%
\bibitem [{\citenamefont {Lami}(2025)}]{Lami_2025_gqsl}%
  \BibitemOpen
  \bibfield  {author} {\bibinfo {author} {\bibfnamefont {L.}~\bibnamefont {Lami}},\ }\bibfield  {title} {\emph {\bibinfo {title} {A solution of the generalized quantum {S}tein's lemma},}\ }\href {http://dx.doi.org/10.1109/tit.2025.3543610} {\bibfield  {journal} {\bibinfo  {journal} {IEEE Trans. Inf. Theory}\ }\textbf {\bibinfo {volume} {71}},\ \bibinfo {pages} {4454–4484} (\bibinfo {year} {2025})}\BibitemShut {NoStop}%
\bibitem [{\citenamefont {Lami}\ \emph {et~al.}(2026{\natexlab{a}})\citenamefont {Lami}, \citenamefont {Berta},\ and\ \citenamefont {Regula}}]{lami_2024_asymptotic_quantification}%
  \BibitemOpen
  \bibfield  {author} {\bibinfo {author} {\bibfnamefont {L.}~\bibnamefont {Lami}}, \bibinfo {author} {\bibfnamefont {M.}~\bibnamefont {Berta}}, \ and\ \bibinfo {author} {\bibfnamefont {B.}~\bibnamefont {Regula}},\ }\bibfield  {title} {\emph {\bibinfo {title} {Asymptotic quantification of entanglement with a single copy},}\ }\href {http://dx.doi.org/10.1038/s41567-026-03182-x} {\bibfield  {journal} {\bibinfo  {journal} {Nat. Phys.}\ }\textbf {\bibinfo {volume} {22}},\ \bibinfo {pages} {439–445} (\bibinfo {year} {2026}{\natexlab{a}})}\BibitemShut {NoStop}%
\bibitem [{\citenamefont {Hiai}\ and\ \citenamefont {Petz}(1991)}]{hiai_1991_proper}%
  \BibitemOpen
  \bibfield  {author} {\bibinfo {author} {\bibfnamefont {F.}~\bibnamefont {Hiai}}\ and\ \bibinfo {author} {\bibfnamefont {D.}~\bibnamefont {Petz}},\ }\bibfield  {title} {\emph {\bibinfo {title} {The proper formula for relative entropy and its asymptotics in quantum probability},}\ }\href {http://dx.doi.org/10.1007/BF02100287} {\bibfield  {journal} {\bibinfo  {journal} {Commun. Math. Phys.}\ }\textbf {\bibinfo {volume} {143}},\ \bibinfo {pages} {99} (\bibinfo {year} {1991})}\BibitemShut {NoStop}%
\bibitem [{\citenamefont {Ogawa}\ and\ \citenamefont {Nagaoka}(2000)}]{Ogawa_2000_strong}%
  \BibitemOpen
  \bibfield  {author} {\bibinfo {author} {\bibfnamefont {T.}~\bibnamefont {Ogawa}}\ and\ \bibinfo {author} {\bibfnamefont {H.}~\bibnamefont {Nagaoka}},\ }\bibfield  {title} {\emph {\bibinfo {title} {Strong converse and {S}tein's lemma in quantum hypothesis testing},}\ }\href {http://dx.doi.org/10.1109/18.887855} {\bibfield  {journal} {\bibinfo  {journal} {IEEE Trans. Inf. Theory}\ }\textbf {\bibinfo {volume} {46}},\ \bibinfo {pages} {2428} (\bibinfo {year} {2000})}\BibitemShut {NoStop}%
\bibitem [{\citenamefont {Nagaoka}(2006)}]{nagaoka_2006_converse_theorem}%
  \BibitemOpen
  \bibfield  {author} {\bibinfo {author} {\bibfnamefont {H.}~\bibnamefont {Nagaoka}},\ }\href {https://arxiv.org/abs/quant-ph/0611289} {\emph {\bibinfo {title} {The converse part of the theorem for quantum {H}oeffding bound},}\ } (\bibinfo {year} {2006}),\ \Eprint {http://arxiv.org/abs/quant-ph/0611289} {arXiv:quant-ph/0611289} \BibitemShut {NoStop}%
\bibitem [{\citenamefont {Hayashi}(2007)}]{Hayashi_2007_error_exponent}%
  \BibitemOpen
  \bibfield  {author} {\bibinfo {author} {\bibfnamefont {M.}~\bibnamefont {Hayashi}},\ }\bibfield  {title} {\emph {\bibinfo {title} {Error exponent in asymmetric quantum hypothesis testing and its application to classical-quantum channel coding},}\ }\href {http://dx.doi.org/10.1103/physreva.76.062301} {\bibfield  {journal} {\bibinfo  {journal} {Phys. Rev. A}\ }\textbf {\bibinfo {volume} {76}} (\bibinfo {year} {2007}),\ 10.1103/physreva.76.062301}\BibitemShut {NoStop}%
\bibitem [{\citenamefont {Audenaert}\ \emph {et~al.}(2008)\citenamefont {Audenaert}, \citenamefont {Nussbaum}, \citenamefont {Szkoła},\ and\ \citenamefont {Verstraete}}]{audenaert_2008}%
  \BibitemOpen
  \bibfield  {author} {\bibinfo {author} {\bibfnamefont {K.~M.~R.}\ \bibnamefont {Audenaert}}, \bibinfo {author} {\bibfnamefont {M.}~\bibnamefont {Nussbaum}}, \bibinfo {author} {\bibfnamefont {A.}~\bibnamefont {Szkoła}}, \ and\ \bibinfo {author} {\bibfnamefont {F.}~\bibnamefont {Verstraete}},\ }\bibfield  {title} {\emph {\bibinfo {title} {Asymptotic {{Error Rates}} in {{Quantum Hypothesis Testing}}},}\ }\href {http://dx.doi.org/10.1007/s00220-008-0417-5} {\bibfield  {journal} {\bibinfo  {journal} {Commun. Math. Phys.}\ }\textbf {\bibinfo {volume} {279}},\ \bibinfo {pages} {251} (\bibinfo {year} {2008})}\BibitemShut {NoStop}%
\bibitem [{\citenamefont {Cai}\ \emph {et~al.}(2004)\citenamefont {Cai}, \citenamefont {Winter},\ and\ \citenamefont {Yeung}}]{cai_quantum_2004}%
  \BibitemOpen
  \bibfield  {author} {\bibinfo {author} {\bibfnamefont {N.}~\bibnamefont {Cai}}, \bibinfo {author} {\bibfnamefont {A.}~\bibnamefont {Winter}}, \ and\ \bibinfo {author} {\bibfnamefont {R.~W.}\ \bibnamefont {Yeung}},\ }\bibfield  {title} {\emph {\bibinfo {title} {Quantum privacy and quantum wiretap channels},}\ }\href {http://dx.doi.org/10.1007/s11122-005-0002-x} {\bibfield  {journal} {\bibinfo  {journal} {Probl. Inf. Transm.}\ }\textbf {\bibinfo {volume} {40}},\ \bibinfo {pages} {318} (\bibinfo {year} {2004})}\BibitemShut {NoStop}%
\bibitem [{\citenamefont {Renes}\ and\ \citenamefont {Renner}(2011)}]{Renes_Noisy_channel}%
  \BibitemOpen
  \bibfield  {author} {\bibinfo {author} {\bibfnamefont {J.~M.}\ \bibnamefont {Renes}}\ and\ \bibinfo {author} {\bibfnamefont {R.}~\bibnamefont {Renner}},\ }\bibfield  {title} {\emph {\bibinfo {title} {Noisy channel coding via privacy amplification and information reconciliation},}\ }\href {http://dx.doi.org/10.1109/TIT.2011.2162226} {\bibfield  {journal} {\bibinfo  {journal} {IEEE Trans. Inf. Theory}\ }\textbf {\bibinfo {volume} {57}},\ \bibinfo {pages} {7377} (\bibinfo {year} {2011})}\BibitemShut {NoStop}%
\bibitem [{\citenamefont {Radhakrishnan}\ \emph {et~al.}(2017)\citenamefont {Radhakrishnan}, \citenamefont {Sen},\ and\ \citenamefont {Warsi}}]{Radhakrishnan_2017_one_shot_private}%
  \BibitemOpen
  \bibfield  {author} {\bibinfo {author} {\bibfnamefont {J.}~\bibnamefont {Radhakrishnan}}, \bibinfo {author} {\bibfnamefont {P.}~\bibnamefont {Sen}}, \ and\ \bibinfo {author} {\bibfnamefont {N.~A.}\ \bibnamefont {Warsi}},\ }\href {https://arxiv.org/abs/1703.01932} {\emph {\bibinfo {title} {One-shot private classical capacity of quantum wiretap channel: Based on one-shot quantum covering lemma},}\ } (\bibinfo {year} {2017}),\ \Eprint {http://arxiv.org/abs/1703.01932} {arXiv:1703.01932 [quant-ph]} \BibitemShut {NoStop}%
\bibitem [{\citenamefont {Boche}\ \emph {et~al.}(2014)\citenamefont {Boche}, \citenamefont {Cai}, \citenamefont {Cai},\ and\ \citenamefont {Deppe}}]{Boche_2014_secrecy_capacity_compound}%
  \BibitemOpen
  \bibfield  {author} {\bibinfo {author} {\bibfnamefont {H.}~\bibnamefont {Boche}}, \bibinfo {author} {\bibfnamefont {M.}~\bibnamefont {Cai}}, \bibinfo {author} {\bibfnamefont {N.}~\bibnamefont {Cai}}, \ and\ \bibinfo {author} {\bibfnamefont {C.}~\bibnamefont {Deppe}},\ }\bibfield  {title} {\emph {\bibinfo {title} {Secrecy capacities of compound quantum wiretap channels and applications},}\ }\href {http://dx.doi.org/10.1103/physreva.89.052320} {\bibfield  {journal} {\bibinfo  {journal} {Phys. Rev. A}\ }\textbf {\bibinfo {volume} {89}},\ \bibinfo {pages} {052320} (\bibinfo {year} {2014})}\BibitemShut {NoStop}%
\bibitem [{\citenamefont {Datta}\ and\ \citenamefont {Hsieh}(2010)}]{Datta_2010_universal_private}%
  \BibitemOpen
  \bibfield  {author} {\bibinfo {author} {\bibfnamefont {N.}~\bibnamefont {Datta}}\ and\ \bibinfo {author} {\bibfnamefont {M.-H.}\ \bibnamefont {Hsieh}},\ }\bibfield  {title} {\emph {\bibinfo {title} {Universal coding for transmission of private information},}\ }\href {http://dx.doi.org/10.1063/1.3521499} {\bibfield  {journal} {\bibinfo  {journal} {J. Math. Phys.}\ }\textbf {\bibinfo {volume} {51}},\ \bibinfo {pages} {122202} (\bibinfo {year} {2010})}\BibitemShut {NoStop}%
\bibitem [{\citenamefont {Horodecki}\ \emph {et~al.}(2005)\citenamefont {Horodecki}, \citenamefont {Horodecki}, \citenamefont {Horodecki},\ and\ \citenamefont {Oppenheim}}]{Horodecki_2005_secret_key}%
  \BibitemOpen
  \bibfield  {author} {\bibinfo {author} {\bibfnamefont {K.}~\bibnamefont {Horodecki}}, \bibinfo {author} {\bibfnamefont {M.}~\bibnamefont {Horodecki}}, \bibinfo {author} {\bibfnamefont {P.}~\bibnamefont {Horodecki}}, \ and\ \bibinfo {author} {\bibfnamefont {J.}~\bibnamefont {Oppenheim}},\ }\bibfield  {title} {\emph {\bibinfo {title} {Secure key from bound entanglement},}\ }\href {http://dx.doi.org/10.1103/physrevlett.94.160502} {\bibfield  {journal} {\bibinfo  {journal} {Phys. Rev. Lett.}\ }\textbf {\bibinfo {volume} {94}},\ \bibinfo {pages} {160502} (\bibinfo {year} {2005})}\BibitemShut {NoStop}%
\bibitem [{\citenamefont {Devetak}\ and\ \citenamefont {Winter}(2005)}]{Devetak_2005_distillation}%
  \BibitemOpen
  \bibfield  {author} {\bibinfo {author} {\bibfnamefont {I.}~\bibnamefont {Devetak}}\ and\ \bibinfo {author} {\bibfnamefont {A.}~\bibnamefont {Winter}},\ }\bibfield  {title} {\emph {\bibinfo {title} {Distillation of secret key and entanglement from quantum states},}\ }\href {http://dx.doi.org/10.1098/rspa.2004.1372} {\bibfield  {journal} {\bibinfo  {journal} {Proc. R. Soc. A Math. Phys. Eng. Sci.}\ }\textbf {\bibinfo {volume} {461}},\ \bibinfo {pages} {207–235} (\bibinfo {year} {2005})}\BibitemShut {NoStop}%
\bibitem [{\citenamefont {Christandl}\ \emph {et~al.}(2007)\citenamefont {Christandl}, \citenamefont {Ekert}, \citenamefont {Horodecki}, \citenamefont {Horodecki}, \citenamefont {Oppenheim},\ and\ \citenamefont {Renner}}]{Christandl_2007_unifying_classical_quantum}%
  \BibitemOpen
  \bibfield  {author} {\bibinfo {author} {\bibfnamefont {M.}~\bibnamefont {Christandl}}, \bibinfo {author} {\bibfnamefont {A.}~\bibnamefont {Ekert}}, \bibinfo {author} {\bibfnamefont {M.}~\bibnamefont {Horodecki}}, \bibinfo {author} {\bibfnamefont {P.}~\bibnamefont {Horodecki}}, \bibinfo {author} {\bibfnamefont {J.}~\bibnamefont {Oppenheim}}, \ and\ \bibinfo {author} {\bibfnamefont {R.}~\bibnamefont {Renner}},\ }\bibfield  {title} {\emph {\bibinfo {title} {Unifying {Classical} and {Quantum} {Key} {Distillation}},}\ }in\ \href {http://dx.doi.org/10.1007/978-3-540-70936-7_25} {\emph {\bibinfo {booktitle} {Theory of {Cryptography}}}},\ \bibinfo {editor} {edited by\ \bibinfo {editor} {\bibfnamefont {S.~P.}\ \bibnamefont {Vadhan}}}\ (\bibinfo  {publisher} {Springer},\ \bibinfo {address} {Berlin, Heidelberg},\ \bibinfo {year} {2007})\ pp.\ \bibinfo {pages} {456--478}\BibitemShut {NoStop}%
\bibitem [{\citenamefont {Renner}\ and\ \citenamefont {König}(2005)}]{Renner_2004_universally_composable}%
  \BibitemOpen
  \bibfield  {author} {\bibinfo {author} {\bibfnamefont {R.}~\bibnamefont {Renner}}\ and\ \bibinfo {author} {\bibfnamefont {R.}~\bibnamefont {König}},\ }\bibfield  {title} {\emph {\bibinfo {title} {Universally {Composable} {Privacy} {Amplification} {Against} {Quantum} {Adversaries}},}\ }in\ \href {http://dx.doi.org/10.1007/978-3-540-30576-7_22} {\emph {\bibinfo {booktitle} {Theory of {Cryptography}}}},\ \bibinfo {editor} {edited by\ \bibinfo {editor} {\bibfnamefont {J.}~\bibnamefont {Kilian}}}\ (\bibinfo  {publisher} {Springer},\ \bibinfo {address} {Berlin, Heidelberg},\ \bibinfo {year} {2005})\ pp.\ \bibinfo {pages} {407--425}\BibitemShut {NoStop}%
\bibitem [{\citenamefont {Ben-Or}\ \emph {et~al.}(2005)\citenamefont {Ben-Or}, \citenamefont {Horodecki}, \citenamefont {Leung}, \citenamefont {Mayers},\ and\ \citenamefont {Oppenheim}}]{Benor_2004_universal_composable}%
  \BibitemOpen
  \bibfield  {author} {\bibinfo {author} {\bibfnamefont {M.}~\bibnamefont {Ben-Or}}, \bibinfo {author} {\bibfnamefont {M.}~\bibnamefont {Horodecki}}, \bibinfo {author} {\bibfnamefont {D.~W.}\ \bibnamefont {Leung}}, \bibinfo {author} {\bibfnamefont {D.}~\bibnamefont {Mayers}}, \ and\ \bibinfo {author} {\bibfnamefont {J.}~\bibnamefont {Oppenheim}},\ }\bibfield  {title} {\emph {\bibinfo {title} {The {Universal} {Composable} {Security} of {Quantum} {Key} {Distribution}},}\ }in\ \href {http://dx.doi.org/10.1007/978-3-540-30576-7_21} {\emph {\bibinfo {booktitle} {Theory of {Cryptography}}}},\ \bibinfo {editor} {edited by\ \bibinfo {editor} {\bibfnamefont {J.}~\bibnamefont {Kilian}}}\ (\bibinfo  {publisher} {Springer},\ \bibinfo {address} {Berlin, Heidelberg},\ \bibinfo {year} {2005})\ pp.\ \bibinfo {pages} {386--406}\BibitemShut {NoStop}%
\bibitem [{\citenamefont {Portmann}\ and\ \citenamefont {Renner}(2022)}]{Portmann_2022_security}%
  \BibitemOpen
  \bibfield  {author} {\bibinfo {author} {\bibfnamefont {C.}~\bibnamefont {Portmann}}\ and\ \bibinfo {author} {\bibfnamefont {R.}~\bibnamefont {Renner}},\ }\bibfield  {title} {\emph {\bibinfo {title} {Security in quantum cryptography},}\ }\href {http://dx.doi.org/10.1103/revmodphys.94.025008} {\bibfield  {journal} {\bibinfo  {journal} {Rev. Mod. Phys.}\ }\textbf {\bibinfo {volume} {94}},\ \bibinfo {pages} {025008} (\bibinfo {year} {2022})}\BibitemShut {NoStop}%
\bibitem [{\citenamefont {Boche}\ and\ \citenamefont {Janßen}(2016)}]{Boche_2016_secret_key}%
  \BibitemOpen
  \bibfield  {author} {\bibinfo {author} {\bibfnamefont {H.}~\bibnamefont {Boche}}\ and\ \bibinfo {author} {\bibfnamefont {G.}~\bibnamefont {Janßen}},\ }\bibfield  {title} {\emph {\bibinfo {title} {Distillation of secret-key from a class of compound memoryless quantum sources},}\ }\href {http://dx.doi.org/10.1063/1.4960217} {\bibfield  {journal} {\bibinfo  {journal} {J. Math. Phys.}\ }\textbf {\bibinfo {volume} {57}},\ \bibinfo {pages} {082201} (\bibinfo {year} {2016})}\BibitemShut {NoStop}%
\bibitem [{\citenamefont {Khatri}\ and\ \citenamefont {Wilde}(2024{\natexlab{b}})}]{Khatri_textbook}%
  \BibitemOpen
  \bibfield  {author} {\bibinfo {author} {\bibfnamefont {S.}~\bibnamefont {Khatri}}\ and\ \bibinfo {author} {\bibfnamefont {M.~M.}\ \bibnamefont {Wilde}},\ }\href {https://arxiv.org/abs/2011.04672} {\emph {\bibinfo {title} {Principles of quantum communication theory: A modern approach},}\ } (\bibinfo {year} {2024}{\natexlab{b}}),\ \Eprint {http://arxiv.org/abs/2011.04672} {arXiv:2011.04672 [quant-ph]} \BibitemShut {NoStop}%
\bibitem [{\citenamefont {Boche}\ \emph {et~al.}(2017)\citenamefont {Boche}, \citenamefont {Janßen},\ and\ \citenamefont {Kaltenstadler}}]{boche_entanglement-assisted_2017}%
  \BibitemOpen
  \bibfield  {author} {\bibinfo {author} {\bibfnamefont {H.}~\bibnamefont {Boche}}, \bibinfo {author} {\bibfnamefont {G.}~\bibnamefont {Janßen}}, \ and\ \bibinfo {author} {\bibfnamefont {S.}~\bibnamefont {Kaltenstadler}},\ }\bibfield  {title} {\emph {\bibinfo {title} {Entanglement-assisted classical capacities of compound and arbitrarily varying quantum channels},}\ }\href {http://dx.doi.org/10.1007/s11128-017-1538-6} {\bibfield  {journal} {\bibinfo  {journal} {Quantum Inf. Process.}\ }\textbf {\bibinfo {volume} {16}},\ \bibinfo {pages} {88} (\bibinfo {year} {2017})}\BibitemShut {NoStop}%
\bibitem [{\citenamefont {Berta}\ \emph {et~al.}(2017)\citenamefont {Berta}, \citenamefont {Gharibyan},\ and\ \citenamefont {Walter}}]{Berta_2017_compound_entanglement_assisted}%
  \BibitemOpen
  \bibfield  {author} {\bibinfo {author} {\bibfnamefont {M.}~\bibnamefont {Berta}}, \bibinfo {author} {\bibfnamefont {H.}~\bibnamefont {Gharibyan}}, \ and\ \bibinfo {author} {\bibfnamefont {M.}~\bibnamefont {Walter}},\ }\bibfield  {title} {\emph {\bibinfo {title} {Entanglement-assisted capacities of compound quantum channels},}\ }\href {http://dx.doi.org/10.1109/TIT.2017.2672981} {\bibfield  {journal} {\bibinfo  {journal} {IEEE Trans. Inf. Theory}\ }\textbf {\bibinfo {volume} {63}},\ \bibinfo {pages} {3306} (\bibinfo {year} {2017})}\BibitemShut {NoStop}%
\bibitem [{\citenamefont {Devetak}\ \emph {et~al.}(2004)\citenamefont {Devetak}, \citenamefont {Harrow},\ and\ \citenamefont {Winter}}]{Devetak_2004_family}%
  \BibitemOpen
  \bibfield  {author} {\bibinfo {author} {\bibfnamefont {I.}~\bibnamefont {Devetak}}, \bibinfo {author} {\bibfnamefont {A.~W.}\ \bibnamefont {Harrow}}, \ and\ \bibinfo {author} {\bibfnamefont {A.}~\bibnamefont {Winter}},\ }\bibfield  {title} {\emph {\bibinfo {title} {A family of quantum protocols},}\ }\href {http://dx.doi.org/10.1103/physrevlett.93.230504} {\bibfield  {journal} {\bibinfo  {journal} {Phys. Rev. Lett.}\ }\textbf {\bibinfo {volume} {93}},\ \bibinfo {pages} {230504} (\bibinfo {year} {2004})}\BibitemShut {NoStop}%
\bibitem [{\citenamefont {Devetak}(2006)}]{Devetak_2006_triangle_of_duality}%
  \BibitemOpen
  \bibfield  {author} {\bibinfo {author} {\bibfnamefont {I.}~\bibnamefont {Devetak}},\ }\bibfield  {title} {\emph {\bibinfo {title} {Triangle of dualities between quantum communication protocols},}\ }\href {http://dx.doi.org/10.1103/physrevlett.97.140503} {\bibfield  {journal} {\bibinfo  {journal} {Phys. Rev. Lett.}\ }\textbf {\bibinfo {volume} {97}},\ \bibinfo {pages} {140503} (\bibinfo {year} {2006})}\BibitemShut {NoStop}%
\bibitem [{\citenamefont {Abeyesinghe}\ \emph {et~al.}(2009)\citenamefont {Abeyesinghe}, \citenamefont {Devetak}, \citenamefont {Hayden},\ and\ \citenamefont {Winter}}]{Abeyesinghe_2009}%
  \BibitemOpen
  \bibfield  {author} {\bibinfo {author} {\bibfnamefont {A.}~\bibnamefont {Abeyesinghe}}, \bibinfo {author} {\bibfnamefont {I.}~\bibnamefont {Devetak}}, \bibinfo {author} {\bibfnamefont {P.}~\bibnamefont {Hayden}}, \ and\ \bibinfo {author} {\bibfnamefont {A.}~\bibnamefont {Winter}},\ }\bibfield  {title} {\emph {\bibinfo {title} {The mother of all protocols: restructuring quantum information’s family tree},}\ }\href {http://dx.doi.org/10.1098/rspa.2009.0202} {\bibfield  {journal} {\bibinfo  {journal} {Proc. R. Soc. A Math. Phys. Eng. Sci.}\ }\textbf {\bibinfo {volume} {465}},\ \bibinfo {pages} {2537–2563} (\bibinfo {year} {2009})}\BibitemShut {NoStop}%
\bibitem [{\citenamefont {Bravo-Prieto}\ \emph {et~al.}(2026)\citenamefont {Bravo-Prieto}, \citenamefont {Gong},\ and\ \citenamefont {Mele}}]{bravoprieto_2026_quantum_memory}%
  \BibitemOpen
  \bibfield  {author} {\bibinfo {author} {\bibfnamefont {C.}~\bibnamefont {Bravo-Prieto}}, \bibinfo {author} {\bibfnamefont {W.}~\bibnamefont {Gong}}, \ and\ \bibinfo {author} {\bibfnamefont {A.~A.}\ \bibnamefont {Mele}},\ }\href {https://arxiv.org/abs/2607.13476} {\emph {\bibinfo {title} {Quantum memory advantage for quantum process tomography},}\ } (\bibinfo {year} {2026}),\ \Eprint {http://arxiv.org/abs/2607.13476} {arXiv:2607.13476 [quant-ph]} \BibitemShut {NoStop}%
\bibitem [{\citenamefont {Shirokov}(2017)}]{Shirokov_2017}%
  \BibitemOpen
  \bibfield  {author} {\bibinfo {author} {\bibfnamefont {M.~E.}\ \bibnamefont {Shirokov}},\ }\bibfield  {title} {\emph {\bibinfo {title} {Tight uniform continuity bounds for the quantum conditional mutual information, for the holevo quantity, and for capacities of quantum channels},}\ }\href {http://dx.doi.org/10.1063/1.4987135} {\bibfield  {journal} {\bibinfo  {journal} {J. Math. Phys.}\ }\textbf {\bibinfo {volume} {58}},\ \bibinfo {pages} {102202} (\bibinfo {year} {2017})}\BibitemShut {NoStop}%
\bibitem [{\citenamefont {Dupuis}(2010)}]{dupuis_2010_phd_thesis}%
  \BibitemOpen
  \bibfield  {author} {\bibinfo {author} {\bibfnamefont {F.}~\bibnamefont {Dupuis}},\ }\href {https://arxiv.org/abs/1004.1641} {\emph {\bibinfo {title} {The decoupling approach to quantum information theory},}\ } (\bibinfo {year} {2010}),\ \Eprint {http://arxiv.org/abs/1004.1641} {arXiv:1004.1641 [quant-ph]} \BibitemShut {NoStop}%
\bibitem [{\citenamefont {Cover}(1972)}]{Cover_broadcast_channel}%
  \BibitemOpen
  \bibfield  {author} {\bibinfo {author} {\bibfnamefont {T.}~\bibnamefont {Cover}},\ }\bibfield  {title} {\emph {\bibinfo {title} {Broadcast channels},}\ }\href {http://dx.doi.org/10.1109/TIT.1972.1054727} {\bibfield  {journal} {\bibinfo  {journal} {IEEE Trans. Inf. Theory}\ }\textbf {\bibinfo {volume} {18}},\ \bibinfo {pages} {2} (\bibinfo {year} {1972})}\BibitemShut {NoStop}%
\bibitem [{\citenamefont {Bergmans}(1973)}]{Bergman_1973_random_coding}%
  \BibitemOpen
  \bibfield  {author} {\bibinfo {author} {\bibfnamefont {P.}~\bibnamefont {Bergmans}},\ }\bibfield  {title} {\emph {\bibinfo {title} {Random coding theorem for broadcast channels with degraded components},}\ }\href {http://dx.doi.org/10.1109/TIT.1973.1054980} {\bibfield  {journal} {\bibinfo  {journal} {IEEE Trans. Inf. Theory}\ }\textbf {\bibinfo {volume} {19}},\ \bibinfo {pages} {197} (\bibinfo {year} {1973})}\BibitemShut {NoStop}%
\bibitem [{\citenamefont {{Gallager}}(1974)}]{Gallager__1974_Capacity}%
  \BibitemOpen
  \bibfield  {author} {\bibinfo {author} {\bibfnamefont {R.~G.}\ \bibnamefont {{Gallager}}},\ }\bibfield  {title} {\emph {\bibinfo {title} {Capacity and coding for degraded broadcast channels},}\ }\href {https://www.mathnet.ru/php/archive.phtml?wshow=paper&jrnid=ppi&paperid=1036&option_lang=eng} {\bibfield  {journal} {\bibinfo  {journal} {Probl. Inf. Transm.}\ }\textbf {\bibinfo {volume} {10}},\ \bibinfo {pages} {185} (\bibinfo {year} {1974})}\BibitemShut {NoStop}%
\bibitem [{\citenamefont {Korner}\ and\ \citenamefont {Marton}(1977)}]{Korner_general_broadcast}%
  \BibitemOpen
  \bibfield  {author} {\bibinfo {author} {\bibfnamefont {J.}~\bibnamefont {Korner}}\ and\ \bibinfo {author} {\bibfnamefont {K.}~\bibnamefont {Marton}},\ }\bibfield  {title} {\emph {\bibinfo {title} {General broadcast channels with degraded message sets},}\ }\href {http://dx.doi.org/10.1109/TIT.1977.1055655} {\bibfield  {journal} {\bibinfo  {journal} {IEEE Trans. Inf. Theory}\ }\textbf {\bibinfo {volume} {23}},\ \bibinfo {pages} {60} (\bibinfo {year} {1977})}\BibitemShut {NoStop}%
\bibitem [{\citenamefont {Yard}\ \emph {et~al.}(2011)\citenamefont {Yard}, \citenamefont {Hayden},\ and\ \citenamefont {Devetak}}]{Yard_2011_quantum_broadcast}%
  \BibitemOpen
  \bibfield  {author} {\bibinfo {author} {\bibfnamefont {J.}~\bibnamefont {Yard}}, \bibinfo {author} {\bibfnamefont {P.}~\bibnamefont {Hayden}}, \ and\ \bibinfo {author} {\bibfnamefont {I.}~\bibnamefont {Devetak}},\ }\bibfield  {title} {\emph {\bibinfo {title} {Quantum broadcast channels},}\ }\href {http://dx.doi.org/10.1109/tit.2011.2165811} {\bibfield  {journal} {\bibinfo  {journal} {IEEE Trans. Inf. Theory}\ }\textbf {\bibinfo {volume} {57}},\ \bibinfo {pages} {7147–7162} (\bibinfo {year} {2011})}\BibitemShut {NoStop}%
\bibitem [{\citenamefont {Savov}\ and\ \citenamefont {Wilde}(2015)}]{Savov_2015_classical_code}%
  \BibitemOpen
  \bibfield  {author} {\bibinfo {author} {\bibfnamefont {I.}~\bibnamefont {Savov}}\ and\ \bibinfo {author} {\bibfnamefont {M.~M.}\ \bibnamefont {Wilde}},\ }\bibfield  {title} {\emph {\bibinfo {title} {Classical codes for quantum broadcast channels},}\ }\href {http://dx.doi.org/10.1109/tit.2015.2485998} {\bibfield  {journal} {\bibinfo  {journal} {IEEE Trans. Inf. Theory}\ }\textbf {\bibinfo {volume} {61}},\ \bibinfo {pages} {7017–7028} (\bibinfo {year} {2015})}\BibitemShut {NoStop}%
\bibitem [{\citenamefont {Hayashi}\ and\ \citenamefont {Matsumoto}(2011)}]{Hayashi_2011_universally_attainable}%
  \BibitemOpen
  \bibfield  {author} {\bibinfo {author} {\bibfnamefont {M.}~\bibnamefont {Hayashi}}\ and\ \bibinfo {author} {\bibfnamefont {R.}~\bibnamefont {Matsumoto}},\ }\href {https://arxiv.org/abs/1104.4285} {\emph {\bibinfo {title} {Universally attainable error and information exponents, and equivocation rate for the broadcast channels with confidential messages},}\ } (\bibinfo {year} {2011}),\ \Eprint {http://arxiv.org/abs/1104.4285} {arXiv:1104.4285 [cs.IT]} \BibitemShut {NoStop}%
\bibitem [{\citenamefont {Boche}\ \emph {et~al.}(2020)\citenamefont {Boche}, \citenamefont {Janßen},\ and\ \citenamefont {Saeedinaeeni}}]{Boche_2020_universal_superposition}%
  \BibitemOpen
  \bibfield  {author} {\bibinfo {author} {\bibfnamefont {H.}~\bibnamefont {Boche}}, \bibinfo {author} {\bibfnamefont {G.}~\bibnamefont {Janßen}}, \ and\ \bibinfo {author} {\bibfnamefont {S.}~\bibnamefont {Saeedinaeeni}},\ }\bibfield  {title} {\emph {\bibinfo {title} {Universal superposition codes: Capacity regions of compound quantum broadcast channel with confidential messages},}\ }\href {http://dx.doi.org/10.1063/1.5139622} {\bibfield  {journal} {\bibinfo  {journal} {J. Math. Phys.}\ }\textbf {\bibinfo {volume} {61}},\ \bibinfo {pages} {042204} (\bibinfo {year} {2020})}\BibitemShut {NoStop}%
\bibitem [{\citenamefont {Hayashi}\ and\ \citenamefont {Cai}(2022)}]{Hayashi_2022_universal_cq_superposition}%
  \BibitemOpen
  \bibfield  {author} {\bibinfo {author} {\bibfnamefont {M.}~\bibnamefont {Hayashi}}\ and\ \bibinfo {author} {\bibfnamefont {N.}~\bibnamefont {Cai}},\ }\bibfield  {title} {\emph {\bibinfo {title} {Universal classical-quantum superposition coding and universal classical-quantum multiple access channel coding},}\ }\href {http://dx.doi.org/10.1109/tit.2021.3131575} {\bibfield  {journal} {\bibinfo  {journal} {IEEE Trans. Inf. Theory}\ }\textbf {\bibinfo {volume} {68}},\ \bibinfo {pages} {1822–1850} (\bibinfo {year} {2022})}\BibitemShut {NoStop}%
\bibitem [{\citenamefont {Marton}(1979)}]{Marton_coding_theorem}%
  \BibitemOpen
  \bibfield  {author} {\bibinfo {author} {\bibfnamefont {K.}~\bibnamefont {Marton}},\ }\bibfield  {title} {\emph {\bibinfo {title} {A coding theorem for the discrete memoryless broadcast channel},}\ }\href {http://dx.doi.org/10.1109/TIT.1979.1056046} {\bibfield  {journal} {\bibinfo  {journal} {IEEE Trans. Inf. Theory}\ }\textbf {\bibinfo {volume} {25}},\ \bibinfo {pages} {306} (\bibinfo {year} {1979})}\BibitemShut {NoStop}%
\bibitem [{\citenamefont {El~Gamal}\ and\ \citenamefont {van~der Meulen}(1981)}]{ElGamal_proof_of_marton}%
  \BibitemOpen
  \bibfield  {author} {\bibinfo {author} {\bibfnamefont {A.}~\bibnamefont {El~Gamal}}\ and\ \bibinfo {author} {\bibfnamefont {E.}~\bibnamefont {van~der Meulen}},\ }\bibfield  {title} {\emph {\bibinfo {title} {A proof of marton's coding theorem for the discrete memoryless broadcast channel (corresp.)},}\ }\href {http://dx.doi.org/10.1109/TIT.1981.1056302} {\bibfield  {journal} {\bibinfo  {journal} {IEEE Trans. Inf. Theory}\ }\textbf {\bibinfo {volume} {27}},\ \bibinfo {pages} {120} (\bibinfo {year} {1981})}\BibitemShut {NoStop}%
\bibitem [{\citenamefont {Huang}\ \emph {et~al.}(2026)\citenamefont {Huang}, \citenamefont {Liu},\ and\ \citenamefont {Liu}}]{Huang_2026_suboptimality_martons}%
  \BibitemOpen
  \bibfield  {author} {\bibinfo {author} {\bibfnamefont {M.}~\bibnamefont {Huang}}, \bibinfo {author} {\bibfnamefont {Y.}~\bibnamefont {Liu}}, \ and\ \bibinfo {author} {\bibfnamefont {Y.}~\bibnamefont {Liu}},\ }\href {https://arxiv.org/abs/2608.19869} {\emph {\bibinfo {title} {Sub-optimality of marton's inner bound for the two-receiver broadcast channel},}\ } (\bibinfo {year} {2026}),\ \Eprint {http://arxiv.org/abs/2608.19869} {arXiv:2608.19869 [cs.IT]} \BibitemShut {NoStop}%
\bibitem [{\citenamefont {Geng}\ \emph {et~al.}(2014)\citenamefont {Geng}, \citenamefont {Gohari}, \citenamefont {Nair},\ and\ \citenamefont {Yu}}]{Geng_2014_marton_inner}%
  \BibitemOpen
  \bibfield  {author} {\bibinfo {author} {\bibfnamefont {Y.}~\bibnamefont {Geng}}, \bibinfo {author} {\bibfnamefont {A.}~\bibnamefont {Gohari}}, \bibinfo {author} {\bibfnamefont {C.}~\bibnamefont {Nair}}, \ and\ \bibinfo {author} {\bibfnamefont {Y.}~\bibnamefont {Yu}},\ }\bibfield  {title} {\emph {\bibinfo {title} {On marton’s inner bound and its optimality for classes of product broadcast channels},}\ }\href {http://dx.doi.org/10.1109/tit.2013.2285925} {\bibfield  {journal} {\bibinfo  {journal} {IEEE Trans. Inf. Theory}\ }\textbf {\bibinfo {volume} {60}},\ \bibinfo {pages} {22–41} (\bibinfo {year} {2014})}\BibitemShut {NoStop}%
\bibitem [{\citenamefont {Liu}\ \emph {et~al.}(2015)\citenamefont {Liu}, \citenamefont {Cuff},\ and\ \citenamefont {Verdu}}]{Liu_2015_one_shot_marton}%
  \BibitemOpen
  \bibfield  {author} {\bibinfo {author} {\bibfnamefont {J.}~\bibnamefont {Liu}}, \bibinfo {author} {\bibfnamefont {P.}~\bibnamefont {Cuff}}, \ and\ \bibinfo {author} {\bibfnamefont {S.}~\bibnamefont {Verdu}},\ }\bibfield  {title} {\emph {\bibinfo {title} {One-shot mutual covering lemma and marton’s inner bound with a common message},}\ }in\ \href {http://dx.doi.org/10.1109/isit.2015.7282697} {\emph {\bibinfo {booktitle} {2015 IEEE International Symposium on Information Theory (ISIT)}}}\ (\bibinfo  {publisher} {IEEE},\ \bibinfo {year} {2015})\ p.\ \bibinfo {pages} {1457–1461}\BibitemShut {NoStop}%
\bibitem [{\citenamefont {Dupuis}\ \emph {et~al.}(2010)\citenamefont {Dupuis}, \citenamefont {Hayden},\ and\ \citenamefont {Li}}]{Dupuis_2010_father_broadcast}%
  \BibitemOpen
  \bibfield  {author} {\bibinfo {author} {\bibfnamefont {F.}~\bibnamefont {Dupuis}}, \bibinfo {author} {\bibfnamefont {P.}~\bibnamefont {Hayden}}, \ and\ \bibinfo {author} {\bibfnamefont {K.}~\bibnamefont {Li}},\ }\bibfield  {title} {\emph {\bibinfo {title} {A father protocol for quantum broadcast channels},}\ }\href {http://dx.doi.org/10.1109/tit.2010.2046217} {\bibfield  {journal} {\bibinfo  {journal} {IEEE Trans. Inf. Theory}\ }\textbf {\bibinfo {volume} {56}},\ \bibinfo {pages} {2946–2956} (\bibinfo {year} {2010})}\BibitemShut {NoStop}%
\bibitem [{\citenamefont {Matsumoto}\ and\ \citenamefont {Hayashi}(2007)}]{Matsumoto_universal_entanglement}%
  \BibitemOpen
  \bibfield  {author} {\bibinfo {author} {\bibfnamefont {K.}~\bibnamefont {Matsumoto}}\ and\ \bibinfo {author} {\bibfnamefont {M.}~\bibnamefont {Hayashi}},\ }\bibfield  {title} {\emph {\bibinfo {title} {Universal distortion-free entanglement concentration},}\ }\href {http://dx.doi.org/10.1103/PhysRevA.75.062338} {\bibfield  {journal} {\bibinfo  {journal} {Phys. Rev. A}\ }\textbf {\bibinfo {volume} {75}},\ \bibinfo {pages} {062338} (\bibinfo {year} {2007})}\BibitemShut {NoStop}%
\bibitem [{\citenamefont {Watanabe}\ and\ \citenamefont {Takagi}(2026)}]{Watanabe_universal}%
  \BibitemOpen
  \bibfield  {author} {\bibinfo {author} {\bibfnamefont {K.}~\bibnamefont {Watanabe}}\ and\ \bibinfo {author} {\bibfnamefont {R.}~\bibnamefont {Takagi}},\ }\bibfield  {title} {\emph {\bibinfo {title} {Universal work extraction in quantum thermodynamics},}\ }\href {http://dx.doi.org/10.1038/s41467-026-69143-3} {\bibfield  {journal} {\bibinfo  {journal} {Nat. Commun.}\ }\textbf {\bibinfo {volume} {17}},\ \bibinfo {pages} {1857} (\bibinfo {year} {2026})}\BibitemShut {NoStop}%
\bibitem [{\citenamefont {Rizzo}\ and\ \citenamefont {Leone}(2026)}]{Rizzo_2026_universal_magic}%
  \BibitemOpen
  \bibfield  {author} {\bibinfo {author} {\bibfnamefont {J.}~\bibnamefont {Rizzo}}\ and\ \bibinfo {author} {\bibfnamefont {L.}~\bibnamefont {Leone}},\ }\href {https://arxiv.org/abs/2608.13376} {\emph {\bibinfo {title} {Universal magic state concentration},}\ } (\bibinfo {year} {2026}),\ \Eprint {http://arxiv.org/abs/2608.13376} {arXiv:2608.13376 [quant-ph]} \BibitemShut {NoStop}%
\bibitem [{\citenamefont {Lami}\ \emph {et~al.}(2026{\natexlab{b}})\citenamefont {Lami}, \citenamefont {Regula},\ and\ \citenamefont {Takagi}}]{Lami_2026_universal_quantum_resource_distillation}%
  \BibitemOpen
  \bibfield  {author} {\bibinfo {author} {\bibfnamefont {L.}~\bibnamefont {Lami}}, \bibinfo {author} {\bibfnamefont {B.}~\bibnamefont {Regula}}, \ and\ \bibinfo {author} {\bibfnamefont {R.}~\bibnamefont {Takagi}},\ }\href {https://arxiv.org/abs/2605.15174} {\emph {\bibinfo {title} {Universal quantum resource distillation via composite generalised quantum stein's lemma},}\ } (\bibinfo {year} {2026}{\natexlab{b}}),\ \Eprint {http://arxiv.org/abs/2605.15174} {arXiv:2605.15174 [quant-ph]} \BibitemShut {NoStop}%
\end{thebibliography}%
\end{document}